\documentclass[11pt, oneside]{article}
\usepackage{etoolbox}

\usepackage{graphicx}	
\usepackage{amssymb}
\usepackage{amsfonts,latexsym,amsthm,amssymb,amsmath,amscd,euscript}
\usepackage{bbm}
\usepackage{framed}
\usepackage{fullpage}
\usepackage{mathrsfs}
\usepackage{physics}
\usepackage[hypertexnames=false]{hyperref}
\usepackage{enumitem}
\usepackage{pifont}

\hypersetup{
  colorlinks=true,
  linkcolor=blue,
  filecolor=blue,
  citecolor = black,      
  urlcolor=cyan,
}
\usepackage{accents}
\usepackage{setspace}
\hypersetup{colorlinks=true,citecolor=blue,urlcolor=black,linkbordercolor={1 0 0}}
\usepackage{tikz-cd}
\newcount\Comments  
\ifnum\Comments=1
\usepackage[inline]{showlabels}

\fi
\usepackage{turnstile}
\usepackage{mathpartir}
\usepackage{tikz}
\usepackage{xspace}

\usepackage{algorithm}
\usepackage{verbatim}
\usepackage[noend]{algpseudocode}

\usepackage{array}
\newcolumntype{L}[1]{>{\raggedright\arraybackslash}p{#1}}

\newcommand{\nc}{\newcommand}

\usepackage[nameinlink,capitalize]{cleveref}
\crefformat{equation}{#2(#1)#3}
\Crefformat{equation}{#2(#1)#3}

\Crefformat{figure}{#2Figure #1#3}
\Crefname{assumption}{Assumption}{Assumptions}
\Crefformat{assumption}{#2Assumption #1#3}
   \Crefname{question}{Question}{Questions}
   \Crefformat{question}{#2Question #1#3}
   \Crefname{claim}{Claim}{Claims}
   \Crefformat{claim}{#2Claim #1#3}
   \Crefname{problem}{Problem}{Problems}
  \Crefformat{problem}{#2Problem #1#3}
   \Crefname{conjecture}{Conjecture}{Conjectures}
  \Crefformat{conjecture}{#2Conjecture #1#3}
\Crefname{subsection}{Section}{Sections}
\Crefname{subsubsection}{Section}{Sections}
\crefformat{subsubsection}{#2Section #1#3}
\Crefformat{subsubsection}{#2Section #1#3}

\Crefname{fact}{Fact}{Facts}

\nc{\sups}[1]{^{\scriptscriptstyle{#1}}}
\nc{\subs}[1]{_{\scriptscriptstyle{#1}}}

\newcommand{\wb}{\widebar}

\nc{\Critic}{\texttt{Critic}\xspace}
\nc{\PSDPUCB}{\texttt{PSDP-UCB}\xspace}
\nc{\LSVIUCB}{\texttt{LSVI-UCB}\xspace}
\nc{\Actor}{\texttt{Actor}\xspace}
\nc{\EstFeature}{\texttt{EstFeature}\xspace}
\nc{\ExpFTPL}{\texttt{ExpFTPL}\xspace}
\nc{\dist}{\mathrm{dist}}
\nc{\Bquad}{B^{\mathsf{quad}}}

\newcommand{\secpar}{\lambda}

\newcommand{\F}{\mathbb{F}}
\newcommand{\Z}{\mathbb{Z}}

\makeatletter
\newtheorem*{rep@theorem}{\rep@title}
\newcommand{\newreptheorem}[2]{%
\newenvironment{rep#1}[1]{%
 \def\rep@title{#2 \ref{##1}}%
 \begin{rep@theorem}}%
 {\end{rep@theorem}}}
\makeatother

\makeatletter
\newcommand\xlabel[2][]{\phantomsection\def\@currentlabelname{#1}\label{#2}}
\makeatother

{
\theoremstyle{plain}
\newtheorem{theorem}{Theorem}
\newtheorem*{theorem*}{Theorem}
\newtheorem{lemma}[theorem]{Lemma}
\newtheorem{corollary}[theorem]{Corollary}
\newtheorem{conjecture}[theorem]{Conjecture}
\newtheorem*{introconjecture}{Conjecture}
\newtheorem{proposition}[theorem]{Proposition}
\newtheorem*{proposition*}{Proposition}
\newtheorem{fact}[theorem]{Fact}

\theoremstyle{definition}
\newtheorem{definition}{Definition}

\numberwithin{theorem}{section}
\numberwithin{definition}{section}
}

\nc{\DMO}{\DeclareMathOperator}

\DMO{\prox}{prox}
\DMO{\UCB}{UCB}
\DMO{\LCB}{LCB}
\nc{\phidiff}{\phi\sups{\Delta}}
\nc{\pexp}{q_{\mathrm{exp}}}
\nc{\nn}{\nonumber}
\nc{\rk}{\mathrm{rk}}
\nc{\brk}[3]{{\rm br}_{#1}^{#2}({#3})}
\nc{\co}{{\rm co}}
\nc{\br}[2]{{\rm br}^{#1}({#2})}
\nc{\depth}[1]{{\rm d}({#1})}
\nc{\tA}{\textsc{A}}
\nc{\child}[2]{{\rm ch}_{#1}({#2})}
\nc{\parent}[1]{{\rm pa}({#1})}
\nc{\dg}{\dagger}
\nc{\bB}{\mathbf{B}}
\nc{\Span}{\mathsf{span}}
\nc{\unif}{\mathsf{unif}}
\nc{\indsig}[2]{\mathcal{I}_{#1}({#2})}
\nc{\total}{{\rm fin}}
\nc{\early}{{\rm pre}}
\nc{\zsink}{z_{\rm sink}}
\nc{\lowv}{{\rm low}}
\nc{\ol}{\overline}
\nc{\ul}{\underline}
\nc{\madec}[3]{\texttt{ma-dec}_{#1}({#2}, {#3})}
\nc{\madeco}[1]{\texttt{ma-dec}_{#1}}
\nc{\madecd}[3]{\texttt{ma-dec}^{\texttt{d}}_{#1}({#2}, {#3})}
\nc{\SF}{\mathscr{F}}
\nc{\SH}{\mathscr{H}}
\nc{\SP}{\mathscr{P}}
\nc{\SPc}{\wb{\mathscr{P}}}
\nc{\SB}{\mathscr{B}}
\nc{\SC}{\mathscr{C}}
\nc{\BS}{\mathbb{S}}
\nc{\PiMarkov}{\Pi^{\rm markov}}
\nc{\trunc}[2]{\mathsf{trunc}_{#2}({#1})}
\nc{\sbl}{of strong Bellman type\xspace}
\nc{\inormal}[1][\Phi, u,v]{\til{N}_{{#1}}}

\nc{\gamvec}{\gamma}
\nc{\til}{\widetilde}
\nc{\td}{\tilde}
\nc{\wh}{\widehat}
\nc{\old}[1]{\ifnum\Comments=1 {\color{brown}  [OLD: #1]}\fi}
\nc{\noah}[1]{\ifnum\Comments=1 {\color{purple} [ng: #1]}\fi}
\nc{\mir}[1]{\ifnum\Comments=1 {\color{teal} [mc: #1]}\fi}
\nc{\sam}[1]{\ifnum\Comments=1 {\color{red} [sg: #1]}\fi}
\nc{\dan}[1]{\ifnum\Comments=1 {\color{magenta} [dw: #1]}\fi}
\nc{\BP}{\mathbb{P}}
\nc{\BI}{\mathbb{I}}
\nc{\midpoint}[1][\Phi,\phi_1,\phi_2]{\mu^{\star}_{{#1}}}

\nc{\fools}[3]{\MF_{#3}({#1}, {#2})}
\nc{\fool}[2]{\MF({#1},{#2})}
\nc{\clip}[2]{{\rm clip}\left[ \left. {#1} \right| {#2} \right]}
\nc{\imax}{\omega}
\DMO{\conv}{conv}
\nc{\MH}{\mathcal{H}}
\nc{\MV}{\mathcal{V}}
\nc{\MC}{\mathcal{C}}
\nc{\MI}{\mathcal{I}}
\nc{\st}{\star}
\nc{\lng}{\langle}
\nc{\rng}{\rangle}
\DMO{\OOPT}{opt}
\nc{\dopt}[2]{\ell_{\OOPT}({#1},{#2})}
\nc{\MG}{\mathcal{G}}
\nc{\MP}{\mathcal{P}}
\nc{\PP}{\mathbb{P}}
\nc{\TT}{\mathbb{T}}
\nc{\TTmax}{\TT_{\max}}
\DMO{\REG}{Reg}
\DMO{\WREG}{wReg}
\nc{\reg}[2]{{\Delta}_{{#1}}({#2})}
\nc{\wreg}[2]{{\Delta}^{\rm w}_{{#1}}({#2})}
\nc{\Reg}[2]{{\REG}_{{#1}}({#2})}
\nc{\wReg}[2]{{\WREG}_{{#1}}({#2})}
\DMO{\Gap}{Gap}
\DMO{\GD}{GD}
\DMO{\GDA}{GDA}
\DMO{\EG}{EG}
\nc{\TE}{\til{\E}}
\nc{\Var}{\mathbb{V}}
\DMO{\Cov}{Cov}
\DMO{\OGDA}{OGDA}
\DMO{\Unif}{Unif}
\nc{\Qu}{\ul{Q}}
\nc{\Qo}{\ol{Q}}
\nc{\Ro}{\ol{R}}
\nc{\Vu}{\ul{V}}
\nc{\Vo}{\ol{V}}
\nc{\RanQ}{\Delta Q}
\nc{\RanV}{\Delta V}
\nc{\clipQ}{\Delta \breve{Q}}
\nc{\frzQ}{\Delta \mathring{Q}}
\nc{\clipV}{\Delta \breve{V}}
\nc{\clipdelta}{\breve{\delta}}
\nc{\cliptheta}{\breve{\theta}}
\nc{\delmin}{\Delta_{{\rm min}}}
\nc{\delmins}[1]{\Delta_{{\rm min},{#1}}}
\nc{\gapfinal}[1]{\max \left\{ \frac{\frzQ_{{#1}}^{k^\st}(x,a)}{2H}, \frac{\delmin}{4H} \right\}}
\nc{\post}[2]{R({#1}; {#2})}
\nc{\posts}[3]{R_{#3}({#1}; {#2})}

\nc{\algnst}[1]{\begin{align*}#1\end{align*}}
\nc{\algn}[1]{\begin{align}#1\end{align}}
\nc{\matx}[1]{\left(\begin{matrix}#1\end{matrix}\right)}

\nc{\nuu}{\nu}

\nc{\bel}[1]{\mathbf{b}({#1})}
\nc{\nbel}[1]{\bar{\mathbf{b}}({#1})}
\nc{\sbel}[2]{\mathbf{b}'_{#1}({#2})}
\nc{\nsbel}[2]{\bar{\mathbf{b}}'_{#1}({#2})}

\nc{\bv}{\mathbf{v}}
\nc{\bone}{\mathbf{1}}
\nc{\bX}{\mathbf{X}}
\nc{\bY}{\mathbf{Y}}
\nc{\bG}{\mathbf{G}}
\nc{\bz}{\mathbf{z}}
\nc{\bw}{\mathbf{w}}
\nc{\bA}{\mathbf{A}}
\nc{\bJ}{\mathbf{J}}
\nc{\bK}{\mathbf{K}}
\nc{\bb}{\mathbf{b}}
\nc{\ba}{\mathbf{a}}
\nc{\bc}{\mathbf{c}}
\nc{\bC}{\mathbf{C}}
\nc{\BR}{\mathbb R}
\nc{\BA}{\mathbb{A}}
\nc{\BC}{\mathbb C}
\nc{\bx}{\mathbf{x}}
\nc{\bS}{\mathbf{S}}
\nc{\bM}{\mathbf{M}}
\nc{\bR}{\mathbf{R}}
\nc{\bN}{\mathbf{N}}
\nc{\NN}{\mathbb{N}}
\nc{\by}{\mathbf{y}}
\nc{\sy}{y}
\nc{\sx}{x}

\nc{\MO}{\mathcal O}
\nc{\MU}{\mathcal{U}}
\nc{\ME}{\mathcal{E}}
\nc{\MN}{\mathcal{N}}
\nc{\MK}{\mathcal{K}}
\nc{\MM}{\mathcal{M}}
\nc{\MS}{\mathcal{S}}
\nc{\MT}{\mathcal{T}}
\nc{\BF}{\mathbb F}
\nc{\BQ}{\mathbb Q}
\nc{\MX}{\mathcal{X}}
\nc{\MA}{\mathcal{A}}
\nc{\MD}{\mathcal{D}}
\nc{\MB}{\mathcal{B}}
\nc{\MZ}{\mathcal{Z}}
\nc{\MJ}{\mathcal{J}}
\nc{\MW}{\mathcal{W}}
\nc{\MY}{\mathcal{Y}}
\nc{\BZ}{\mathbb Z}
\nc{\BN}{\mathbb N}
\nc{\ep}{\epsilon}
\nc{\epbe}{\varepsilon_{\mathsf{BE}}}
\nc{\epout}{\varepsilon_{\mathsf{outlier}}}
\nc{\bellc}[1][h]{\MT_{#1}^\circ}
\nc{\vep}{\varepsilon}
\nc{\gapfn}[1]{\varepsilon_{#1}}
\nc{\ggapfn}[2]{\varphi_{#1}({#2})}
\nc{\epsahk}{\gapfn{0}}
\nc{\BH}{\mathbb H}
\nc{\BG}{\mathbb{G}}
\nc{\D}{\Delta}
\nc{\MF}{\mathcal{F}}
\nc{\One}[1]{\mathbbm{1}\{{#1}\}}
\nc{\bOne}{\mathbf{1}}
\nc{\Aopt}{\mathcal{A}^{\rm opt}}
\nc{\Amul}{\mathcal{A}^{\rm mul}}

\nc{\SQ}{\mathsf Q}

\nc{\DO}{\accentset{\circ}{\D}}
\nc{\mf}{\mathfrak}
\nc{\mfp}{\mathfrak{p}}
\nc{\mfq}{\mf{q}}
\nc{\mfx}{\mf{s}}
\nc{\Sp}{\mbox{Spec}}
\nc{\Spm}{\mbox{Specm}}
\nc{\hookuparrow}{\mathrel{\rotatebox[origin=c]{90}{$\hookrightarrow$}}}
\nc{\hookdownarrow}{\mathrel{\rotatebox[origin=c]{-90}{$\hookrightarrow$}}}
\nc{\hra}{\hookrightarrow}
\nc{\tra}{\twoheadrightarrow}
\nc{\sgn}{{\rm sgn}}
\nc{\aut}{{\rm Aut}}
\nc{\Hom}{{\rm Hom}}
\nc{\img}{{\rm Im}}
\DMO{\id}{Id}
\DMO{\supp}{supp}
\DMO{\KL}{KL}
\nc{\kld}[2]{d_{\mathsf{KL}}({#1}||{#2})}
\nc{\ren}[2]{D_2({#1}||{#2})}
\nc{\chisq}[2]{\chi^2({#1}||{#2})}
\nc{\tvd}[2]{d_{\mathsf{TV}}({#1}, {#2})}
\nc{\hell}[2]{d_{\mathsf{H}}^2({#1}, {#2})}
\nc{\dbi}[3][\pi]{D_{\mathsf{bi}}^{#1}({#2} \| {#3})}
\DMO{\BSS}{BSS}
\DMO{\BES}{BES}
\DMO{\BGS}{BGS}
\DMO{\poly}{poly}
\nc{\indep}{\perp}
\DMO{\sink}{sink}
\nc{\fp}[1]{\MP_1({#1})}
\nc{\BO}{\mathbb{O}}
\nc{\BT}{\mathbb{T}}

\nc{\RR}{\mathbb{R}}
\nc{\Gradient}{\nabla}
\DMO{\diag}{diag}
\nc{\EE}{\mathbb{E}}
\nc{\MQ}{\mathcal{Q}}
\nc{\ML}{\mathcal{L}}
\nc{\cPhi}{\bar \Phi}

\DeclareMathOperator*{\PR}{Pr}
\renewcommand{\Pr}{\PR}
\DeclareMathOperator*{\E}{\mathbb{E}}
\nc{\ra}{\rightarrow}

\nc{\pmhc}[1]{\{-1,1\}^{#1}}
\nc{\Dbnd}{D}
\nc{\Bbnd}{B}

\nc{\Key}{\mathsf{KeyGen}}
\nc{\Enc}{\mathsf{Encode}}
\nc{\Encemb}{\mathsf{EncodeEmb}}
\nc{\Dec}{\mathsf{Decode}}
\nc{\sk}{\mathsf{sk}}
\nc{\pk}{\mathsf{pk}}
\nc{\lpk}{\ell_{\mathsf{pk}}}
\nc{\lsk}{\ell_{\mathsf{sk}}}
\nc{\msg}{\mathsf{m}}
\nc{\Adv}{\mathsf{Adv}}
\nc{\Red}{\mathsf{Red}}
\nc{\negl}{\mathsf{negl}}
\nc{\Ber}{\mathrm{Ber}}
\nc{\PRFPRC}{\mathsf{PRF\text{-}PRC}}
\nc{\wt}{\mathrm{wt}}
\nc{\res}[2]{{#1}_{#2}}
\nc{\bzero}{\mathbf{0}}
\nc{\Bin}{\mathrm{Bin}}
\nc{\Hyp}{\mathrm{Hyp}}

\nc{\Nrho}[1][\rho]{{N}_{#1}}
\nc{\Trho}[1][\rho]{\mathsf{T}_{#1}}
\nc{\hc}[1][n]{\{0,1\}^{#1}}
\nc{\Stab}{\mathbf{Stab}}
\nc{\bW}{\mathbf{W}}
\nc{\NS}{{\mathbf{NS}}}

\nc{\KeyS}{\mathsf{KeyGen_{Sub}}}
\nc{\EncS}{\mathsf{Encode_{Sub}}}
\nc{\DecS}{\mathsf{Decode_{Sub}}}
\nc{\WeightPerturb}{\mathsf{WeightPerturb}}
\nc{\Unique}{\mathsf{Unique}}
\nc{\PRCS}{\mathsf{PRC_{Sub}}}
\nc{\PRC}{\mathsf{PRC}}
\nc{\PRCI}{\mathsf{PRC_{Idx}}}
\nc{\SampleUnique}{\mathsf{SampleUnique}}
\nc{\PerturbDifference}{\mathsf{PerturbDifference}}

\nc{\Model}{\mathsf{Model}}
\nc{\Modelo}{\overline{\Model}}
\nc{\prompt}{\mathtt{PROMPT}}
\nc{\Setup}{\mathsf{Setup}}
\nc{\Detect}{\mathsf{Detect}}
\nc{\Sigprc}{\Sigma_{\mathsf{PRC}}}
\nc{\Wat}{\mathsf{Wat}}
\nc{\term}{\mathtt{END}}
\nc{\tok}{\mathsf{t}}
\nc{\True}{\textsf{True}}
\nc{\False}{\textsf{False}}
\nc{\Eemb}{\ME^{\mathsf{Emb}}}
\nc{\hist}{\mathsf{hist}}
\nc{\hh}{\mathsf{h}}
\nc{\freq}{\mathsf{freq}}
\nc{\ff}{\mathsf{f}}

\nc{\Hemp}[1]{H_{\mathsf{e}}^{#1}}
\nc{\Hempt}[1]{\bar{H}_{\mathsf{e}}^{#1}}
\nc{\Hemptil}[1]{\tilde{H}_{\mathsf{e}}^{#1}}
\nc{\Spread}[1]{S^{#1}}
\nc{\Hmean}[1]{H_{\mathsf{m}}^{#1}}
\nc{\partition}[1][n,q]{P^{\mathsf{ptn}}_{#1}}
\nc{\Crob}{C_{\mathsf{rob}}}
\nc{\Lmax}{L_{\mathsf{max}}}
\nc{\skwat}{\sk_{\mathsf{Wat}}}
\nc{\EmbedToken}{\mathsf{EmbedChar}}
\nc{\len}{\mathrm{len}}
\nc{\Esub}{\ME_{\mathsf{sub}}}
\nc{\Ecomp}{\ME_{\mathsf{comp}}}
\nc{\comp}{\mathsf{c}}
\nc{\SE}{\mathscr{E}}
\nc{\alphb}{q}
\nc{\tAdv}{\widetilde{\Adv}}
\nc{\Funif}{{F_{\mathsf{Unif}}}}
\nc{\Alg}{\mathsf{Alg}}
\nc{\Majority}{\mathsf{Maj}}
\nc{\Dist}{\mathsf{Dist}}
\nc{\edit}{edit\xspace}
\nc{\Edit}{Edit\xspace}
\nc{\Wcomp}{\MW^{\mathsf{comp}}}

\nc{\channSC}{\text{SC}}

\nc{\Dpp}{\MD^{\mathsf{pp}}}
\nc{\Dpc}{\MD^{\mathsf{pc}}}

\nc{\INS}{\mathsf{INS}}
\nc{\CNS}{\mathsf{CNS}}
\nc{\cdist}{\stackrel{\mathrm{c}}{\sim}}
\nc{\SU}{\mathscr{U}}
\nc{\rr}{\bar{n}}

\newcommand{\from}{\leftarrow}
\nc{\KeyGen}{\mathsf{KeyGen}}
\nc{\ED}{D_{\mathsf{E}}}
\nc{\Ham}{D_{\mathsf{H}}}
\nc{\bin}{\mathsf{bin}}
\nc{\EDball}{\mathcal{B}_{\mathsf{E}}}
\nc{\SEDball}{\mathcal{B}_{\mathsf{H,E}}}
\nc{\LEDball}{\mathcal{B}_{\mathsf{len,ED}}}
\nc{\epED}{\varepsilon_{\mathsf{E}}}
\nc{\epedit}{\epED}
\nc{\epham}{\varepsilon_{\mathsf{H}}}
\nc{\epDec}{\varepsilon_{\mathsf{Dec}}}
\nc{\Eedit}{\mathscr{E}^{\mathsf{E}}}
\nc{\Ehamedit}{\mathscr{E}^{\mathsf{H,E}}}
\nc{\Egood}{\ME_{\mathsf{good}}}
\nc{\PermEnc}{\mathsf{PermEncode}}
\nc{\dham}{d_{\mathsf{H}}}
\nc{\dedit}{d_{\mathsf{E}}}
\nc{\pham}{p_{\mathsf{H}}}
\nc{\pedit}{p_{\mathsf{E}}}
\nc{\pDec}{p_{\mathsf{Dec}}}
\nc{\pSub}{p_{\mathsf{Sub}}}
\nc{\Rclean}{R_{\mathsf{Clean}}}
\nc{\adv}{\mathcal{A}}
\nc{\bi}{\mathbf{i}}

\nc{\bj}{\mathbf{j}}
\nc{\bp}{\mathbf{p}}

\nc{\FRS}{\mathsf{FRS}}
\nc{\RS}{\mathsf{RS}}
\nc{\CDec}{C_{\mathsf{Dec}}}
\nc{\trec}{t_{\mathsf{rec}}}
\nc{\ListRecov}{\mathsf{ListRecovery}}
\nc{\Imax}{I_{\mathsf{max}}}
\nc{\Eadv}{\ME_{\mathsf{adv}}}

\nc{\Dnap}{\MD^{\mathsf{nap}}}
\renewcommand{\v}{{\bf v}}

\newcommand{\heading}[1]{\vspace{0.25cm} \noindent \textbf{#1}}

\nc{\cmark}{\ding{51}}%
\nc{\xmark}{\ding{55}}%

\title{Improved Pseudorandom Codes from Permuted Puzzles}
\date{}
\author{
Miranda Christ\\ 
{\small Columbia University}\\
{\small \texttt{mchrist@cs.columbia.edu}}
\and 
Noah Golowich\\
{\small Microsoft Research}\\
{\small \texttt{noah.golowich@austin.utexas.edu}}
\and 
Sam Gunn\\
{\small UC Berkeley}\\
{\small \texttt{gunn@berkeley.edu}}
\and 
Ankur Moitra\\
{\small MIT}\\
{\small \texttt{moitra@mit.edu}}
\and 
Daniel Wichs\\
{\small Northeastern University}\\
{\small NTT Research}\\
{\small \texttt{wichs@ccs.neu.edu}}
}

\begin{document}
\maketitle

\begin{abstract}
    Watermarks are an essential tool for identifying AI-generated content. Recently, Christ and Gunn (CRYPTO '24) introduced \emph{pseudorandom error-correcting codes} (PRCs), which are equivalent to watermarks with strong robustness and quality guarantees.
    A PRC is a pseudorandom encryption scheme whose decryption algorithm tolerates a high rate of errors.
    Pseudorandomness ensures quality preservation of the watermark, and error tolerance of decryption translates to the watermark's ability to withstand modification of the content.

    In the short time since the introduction of PRCs, several works (NeurIPS '24, RANDOM '25, STOC '25) have proposed new constructions.
    Curiously, all of these constructions are vulnerable to quasipolynomial-time distinguishing attacks.
    Furthermore, all lack robustness to edits over a constant-sized alphabet, which is necessary for a meaningfully robust LLM watermark. Lastly, they lack robustness to adversaries who know the watermarking detection key.
    Until now, it was not clear whether any of these properties was achievable individually, let alone together.

    We construct pseudorandom codes that achieve all of the above:  plausible subexponential pseudorandomness security, robustness to worst-case edits over a binary alphabet, and robustness against even computationally unbounded adversaries that have the detection key.    Pseudorandomness rests on a new assumption that we formalize, the \emph{permuted codes conjecture}, which states that a distribution of permuted noisy codewords is pseudorandom.
    We show that this conjecture is implied by the permuted puzzles conjecture used previously to construct doubly efficient private information retrieval.
    To give further evidence, we show that the conjecture holds against a broad class of simple distinguishers, including read-once branching programs.
\end{abstract}

\thispagestyle{empty}
\newpage
{\small\tableofcontents}
\thispagestyle{empty}
\newpage

\setcounter{page}{1}

\section{Introduction}
\label{sec:intro}
The proliferation of AI-generated content creates serious challenges for society today, due to increasingly realistic threats such as misinformation and impersonation.
One way to address this challenge is to use \emph{watermarking}, which involves embedding a hidden signal in generated content that allows it to be detected as such.
This content can later be passed to a \emph{watermarking detector} which can identify the content's provenance using the hidden signal. In order to be useful, watermarking schemes need to satisfy certain properties: first, they should not significantly degrade the quality of the content, and second, they should be robust to modifications of the content by an adversary (e.g., one seeking to remove the watermark). In recent years, a promising approach to watermarking which comes with provable guarantees and also has shown empirical promise \cite{GZS25,YZC+25,HLLL25} is that of pseudorandom codes \cite{CG24}.
The goal of this paper is to provide pseudorandom codes, and thereby watermarking schemes, with improved quality and robustness guarantees. 

A \emph{pseudorandom code (PRC)} is a secret-key encryption scheme that is both \emph{pseudorandom}, in that any polynomial number of ciphertexts are computationally indistinguishable from uniform, and \emph{error-correcting} or \emph{robust}, in that decryption succeeds even when ciphertexts are passed through some error channel.
\cite{CG24} showed a direct way to construct a watermark for essentially any kind of AI-generated content from any pseudorandom code.
Error correction of the PRC implies that the watermark persists under content modifications; pseudorandomness translates to \emph{undetectability}, meaning that the output of the watermarking scheme is \emph{computationally indistinguishable} from that of the true language model, which provides a very strong guarantee against  degradations in quality.
In fact, a PRC may be viewed as a special case of a watermarking scheme, meaning that PRCs are equivalent to watermarking schemes in an appropriate sense. 

In the short time since PRCs were introduced, there has been a significant body of work building them \cite{GM24,GG25,DBLP:conf/stoc/AlrabiahACDG25}, applying them to watermarks \cite{GZS25,DBLP:conf/sp/Cohen0S25,GM24}, and investigating their relationship to other cryptographic primitives \cite{blackbox1,blackbox2}.
An important direction is to construct PRCs with stronger pseudorandomness and robustness, which would immediately yield improved watermarks.
For example, the first PRC proposed in \cite{CG24} withstood only \emph{substitution} errors or \emph{random deletions}, resulting in an LLM watermark that tolerates word substitutions and random deletions.
Ideally, a fully satisfactory watermark should be robust to \emph{non-random edits} as well, and \cite{GM24} later built one by constructing the first edit-robust PRC, although it required an unrealistic assumption on the entropy of generated text.
We remark that our focus in this paper is on watermarking schemes for \emph{language models}, for which edits are a particularly natural notion by which to consider modification of content. Nevertheless, PRCs have been shown to be useful for other types of models as well, such as image diffusion models \cite{GZS25}, and our results have implications in those domains as well.

Despite the intense interest in the area, we still lack PRCs with \emph{any one of} the following three basic and desirable properties:

\heading{Property (1): subexponential security.} 
Until now, all PRCs suffered from quasipolynomial-time attacks against pseudorandomness.
Arguably, the biggest open question surrounding PRCs was whether it was possible to achieve subexponential security at all.
The basic issue was that previous constructions all involved hiding a $(\log n)$-sparse structure in a length-$n$ codeword.
This sparsity was crucial for error correction; any larger structure would be destroyed by errors.
But the price of using such sparse patterns was that there necessarily existed $n^{\log n}$-time brute-force search distinguishing attacks.
Therefore, stronger pseudorandomness seemed to call for new techniques, or possibly even be unattainable.

\heading{Property (2): small-alphabet edit robustness.}
LLM watermarks should withstand edits to the text, and doing so requires edit-robust PRCs.
While \cite{GM24} constructed an edit-robust PRC, theirs used a polynomial-size alphabet.
This translates to a strong entropy requirement for their watermark: It can only be detected if there is a large amount of entropy per word in the text.
Unfortunately, realistic text has only about one bit of entropy per word, necessitating an edit-robust PRC with a \emph{constant-sized}, ideally binary, alphabet.

\heading{Property (3): strongly-adaptive robustness.}
Previous works on PRCs considered error channels with varying levels of power.
The weakest are random errors, and the strongest error models from prior work were computationally bounded channels that have some partial information about the decoding key \cite{DBLP:conf/stoc/AlrabiahACDG25}.
Their corresponding watermarks withstand attackers with similar levels of power---in the case of \cite{DBLP:conf/stoc/AlrabiahACDG25}, the attacker can query a watermarked model and the watermark detector.
However, an attacker that knows the detection key can easily find much smaller perturbations that remove the watermark.
In fact, as a consequence of the weak pseudorandomness of existing PRCs, it is often possible for motivated attackers to learn the detection key for practical parameter settings.\footnote{In the most robust image watermark used by \cite{GZS25}, the PRC has parity checks of size 3, resulting in an $n^3$-time key recovery attack: an attacker can iterate over all $n^3$ possible parity checks and test which are correlated with the watermarked images. See the ``Practical undetectability'' section of~\cite{GZS25} for discussion on their choice of parameters.}


Therefore, we would ideally like to ensure robustness even when the attacker has \emph{full knowledge of the watermarking key}.\footnote{An astute reader may be concerned that an attacker that recovers the detection key may also be able to embed watermarks in any content of their choice. Prior works have constructed watermarks that are unforgeable even by an adversary that knows the detection key~\cite{CG24,DBLP:journals/cic/FairozeGJMMW24,lin2026unforgeable}.} 
Such a guarantee, which we call ``strongly-adaptive robustness'' and which corresponds to the worst-case error model typically considered in coding theory, remained open even for robustness to substitution errors prior to this work. This is not merely an academic goal but has important implications for how watermarking schemes can be deployed, since it would allow anyone to verify a watermark rather than just one trusted party, without compromising robustness. 

\subsection{Results}
In this work, \textbf{\emph{we present a new PRC construction, and accompanying watermarking scheme, satisfying all three of the above properties.}}
In order to prove pseudorandomness, some computational hardness assumption is necessary.
The pseudorandomness of our PRCs is based on the \emph{permuted codes conjecture}, which states that the following distribution is pseudorandom for certain codes $C \subseteq \BF_q^n$ and any noise rate $\eta = \Omega(1)$:
\begin{enumerate}
    \item Sample an index permutation $\pi$ over $[n]$ and alphabet permutations $\pi_1, \dots, \pi_n$ over $[q]$.
    \item Sample codewords $c_1, \dots, c_T$ uniformly at random from $C$.
    \item For each codeword $c_i = (c_{i,1},\ldots,c_{i,n})$, permute the location indices according to $\pi$ and the alphabet at each location index $j$ according to $\pi_j$, obtaining $\hat{c}_i = (\pi_1(c_{i,\pi(1)}),\ldots, \pi_n(c_{i,\pi(n)}))$.
    \item Output $\hat{c}_1 + e_1, \dots, \hat{c}_T + e_T$, where the error $e_i = (e_{i,1},\ldots,e_{i,n})$ is sampled with each coordinate $e_{i,j}$ being 0 with probability $1-\eta$ and uniformly random from $\F_q$ otherwise.
\end{enumerate}

\begin{introconjecture}[Permuted codes conjecture, \Cref{conj:permuted-codes}]
    Let $C$ be any family of linear codes over $\F_q$ with dual distance $n^{\Omega(1)}$ and let $\eta \in (0,1)$ be a constant rate of random substitutions.
    The \emph{permuted codes conjecture} says that the above distribution is pseudorandom for any $T = \poly(n)$.
\end{introconjecture}

We state the  general form of the conjecture, since we are unable to find any counterexamples, and wish to provide a broad target for cryptanalysis.
However, for our applications, we will only require the conjecture to hold for  specific codes $C$. The conjecture also plausibly holds with even sub-exponential security (i.e., for some constant $c>0$ and any number of samples $T= 2^{O(n^c)}$,  attackers running in time $O(2^{n^c})$ have at most $2^{-\omega(n^c)}$ advantage in distinguishing the above permuted code distribution from uniform).
This translates to plausible sub-exponential security for undetectability in our watermarking applications.
While \cref{conj:permuted-codes} has not been stated before in its present form, we show that it follows from the ``permuted puzzles'' assumption of \cite{boyle2021security,blackwell2021note} which was originally used to obtain doubly efficient private information retrieval:

\begin{theorem*}[Evidence for the permuted codes conjecture, \Cref{thm:puzzles-from-codes}]
    The ``permuted puzzles'' conjecture \cite{boyle2021security,blackwell2021note} implies the permuted codes conjecture.
\end{theorem*}

We note that earlier versions of the permuted puzzles conjecture were introduced by~\cite{boyle2017can,DBLP:conf/tcc/CanettiHR17} and further studied in~\cite{DBLP:conf/tcc/BoyleHW19}. 
We focus on the variants of~\cite{boyle2021security,blackwell2021note}, because we are able to show they imply the permuted codes conjecture. 

As the permuted codes and permuted puzzles conjectures are relatively new, we perform some cryptanalysis to better our understanding. On the positive side, we show statistical evidence for the permuted codes conjecture over constant-size alphabets: A sample of $O(\log n)$ codewords are jointly statistically uniform (for comparison recall that \cref{conj:permuted-codes} posits that $\poly(n)$ codewords are computationally indistinguishable from uniform).

\begin{theorem*}[Statistical uniformity of a few codewords, \Cref{cor:stat-evidence}]
    Let $C$ be any code with polynomial dual distance over a constant-sized alphabet. Then there exists a $T = \Omega(\log n)$ such that $T$ samples from $C$'s permuted code distribution (per \cref{conj:permuted-codes}) are statistically indistinguishable from $T$ uniformly random strings.
\end{theorem*}

There are some conceptual similarities between the Permuted Codes conjecture and the ``low-degree conjecture,'' which has been extensively studied in the context of algorithmic statistics in recent years (e.g., \cite{hopkins2018sos,kunisky2019notes,buhai2025quasipolynomial}, amongst many others).
Roughly speaking, the low-degree conjecture states that for any distribution $P$ which is permutation invariant and indistinguishable from uniform by low-degree polynomials (which mirrors our requirement of ``high dual distance''), there is \emph{no efficient algorithm} distinguishing $P$ from uniform. We observe that \cref{cor:stat-evidence} in fact establishes a special case of the low-degree conjecture by showing that the distribution $P$ is statistically close to uniform; see \cref{cor:low-degree-consequence}.\footnote{Incidentally, recent independent work \cite{buhai2025quasipolynomial} has established that the low-degree conjecture is false more generally.} 


On the negative side, we consider a modification of the permuted codes conjecture in which one omits the alphabet permutations applied to the individual symbols of the codewords. We show in \cref{thm:rs-dist} that the conjecture would be false in general in this case, by taking $C$ to be the Reed-Solomon code. 
In particular, we show that \emph{even a constant number of codewords would be efficiently distinguishable from random}.
In fact, it turns out that all three of the randomizations (alphabet permutation, index permutation, and noise) are necessary: The permuted codes conjecture would be false if any one of these was omitted.
See \cref{tab:two-of-three} for more details.

\begin{table}
    \centering
    \begin{tabular}{L{.25\linewidth} | c | L{.6\linewidth}}
        Randomization method & Secure? & Explanation \\
        \hline
        Symbol + alphabet permutations & \xmark & This is the ``toy conjecture'' of \cite{boyle2017can}, which was shown to be insecure when instantiated with Reed--Solomon codes in \cite{boyle2021security,blackwell2021note}.\\
        \hline
        Noise + index permutation & \xmark & We prove in \Cref{thm:rs-dist} that this is insecure when instantiated with Reed--Solomon codes.\\
        \hline
        Noise + alphabet permutations & \xmark & Insecure for any efficiently decodable binary linear code: Over $\BF_2$, the alphabet permutations are simply a one-time pad, which can be removed by adding pairs of samples together (though this roughly doubles the noise rate). Any efficient decoder then serves as a distinguisher. \\
        \hline
        Noise + index + alphabet permutations & \cmark & This is either our permuted codes conjecture or the permuted puzzles conjecture, depending on whether the noise is substitutions or erasures.\\
    \end{tabular}
    \caption{Any two of the three randomizing methods employed in the permuted codes and puzzles conjectures are insufficient to guarantee pseudorandomness in general. As in the permuted codes assumption, we consider codes with polynomial dual distance.}
    \label{tab:two-of-three}
\end{table}

\paragraph{Improved PRCs and watermarking schemes from permuted codes.}  The permuted codes conjecture for any specific efficiently (list-)decodable code $C \subseteq \F_q^n$ directly gives PRCs  with strong adaptive robustness against \emph{substitutions} over the alphabet $\F_q$. Essentially the PRC ciphertexts are the outputs of the permuted codeword distribution, and the detection procedure undoes the permutations and applies the decoding algorithm of the underlying code. The rate of substitutions that the PRC tolerates can be set arbitrarily close to the error tolerance of the underlying code $C$. See \cref{sec:sub-robustness} for this basic result. 

The basic result only handles substitutions. As our main positive result for PRCs, we show how to also handle \emph{edits}. Under the permuted codes conjecture, we can obtain PRCs satisfying all three of the desired properties from \cref{sec:intro} above:

\begin{theorem*}[An improved PRC from permuted codes, \Cref{thm:prc-main}]
    Under the permuted codes conjecture, there exists a binary-alphabet PRC that is strongly-adaptive to some constant rate of edits.
    Pseudorandomness plausibly holds against subexponential-time distinguishers.
\end{theorem*}

For the above, we need the permuted codes conjecture to hold specifically for \emph{Reed-Solomon codes}. Note that while the Reed-Solomon code has a large alphabet, the resulting PRCs are over a binary alphabet. By using a standard transformation from PRCs to watermarking schemes \cite{CG24}, we obtain an LLM watermark with the analogous guarantees:

\begin{theorem*}[An improved watermark, \Cref{thm:watermarking-main}]
    Under (a slight variant of) the permuted codes conjecture, for any constant $\alpha > 0$, there exists an undetectable, $\widetilde\Omega(\alpha^7)$-edit-robust watermark requiring per-token entropy of only $\alpha$.
\end{theorem*}

For this result, we need a slight variant of the permuted codes conjecture to hold for \emph{folded Reed-Solomon Codes}, with the caveat that these are technically not linear codes over their underlying alphabet. 

This is the first undetectable LLM watermark that is robust to a constant rate of edits, and which works under a constant per-token entropy rate.
Furthermore, our edits can be worst-case (i.e., the watermark is strongly adaptively robust). 
Previous watermarks either noticeably changed the distribution of text \cite{kuditipudi2024robust}, or required a superconstant (and therefore unrealistic) rate of entropy \cite{GM24}.

\subsection{Related work}

\heading{Watermarks.} The recent line of work on LLM watermarks was largely initiated by \cite{scott,KGW}.
Several works with provable guarantees followed \cite{CGZ24,ZhaoA0W24,kuditipudi2024robust}; however, until \cite{CG24} no watermark was simultaneously undetectable and robust to significant modification.
Subsequent works with both of these properties all use pseudorandom codes \cite{GM24,GZS25,DBLP:conf/sp/Cohen0S25}.

It is impossible to construct a watermark that is robust against arbitrary adversarial removal strategies, as a strong enough adversary can simply create an unwatermarked response itself.
\cite{DBLP:conf/icml/ZhangEF0AB24} show a formal version of this impossibility, where they assume that the adversary can generate, from any given response, a uniformly random response of the same quality.
As a result, we (and other watermarking works) consider robustness only to restricted adversarial removal strategies, e.g. those making a bounded number of edits to the LLM output.

We refer the reader to \cite{sok} for further background on watermarks.

\heading{Quasipolynomial-time distinguishing attacks.}
The first pseudorandom code construction, of Christ and Gunn \cite{CG24}, is essentially a binary random linear code satisfying some planted $(\log n)$-sparse parity checks.
Encryptions of `1' are random codewords with a constant rate of noise, which are pseudorandom under subexponential LPN.
Encryptions of `0' are uniform strings.
Decoding, which tolerates random substitutions, involves checking whether the given string satisfies a $\left(\frac{1}{2} + \frac{1}{\mathsf{poly}(n)}\right)$ fraction of the parity checks for a particular polynomial in $n$.
This is true for encryptions of `1' if and only if the parity checks are sufficiently---in particular, $O(\log n)$---sparse.
However, this sparsity also enables a quasipolynomial-time attack that brute-force guesses a parity check in $n^{O(\log n)}$ time.

Subsequent works follow a similar technique, embedding hidden $(\log n)$-sparse structure.
For example, \cite{GM24} constructs PRCs under the assumption that there exist $(\log n)$-local weak PRFs.
\cite{GG25} constructs PRCs under the assumption that a random hypergraph is indistinguishable from a random hypergraph with a planted structure consisting of $\log n$ hyperedges.
These constructions exhibit the same phenomenon as \cite{CG24}---$\log n$ sparsity is necessary for decoding, but it unfortunately leads to quasipolynomial-time distinguishing attacks that simply brute-force search for this sparse structure.

In fact, in \Cref{sec:quasipoly-dist} we show that a broader class of constructions which does not immediately appear to necessitate $(\log n)$-sized structures always has quasipolynomial-time distinguishing attacks.
This class includes constructions where codewords are comprised in part of a random string $r$ and a predicate $f(r)$ that tolerates a constant rate of errors.

We do not know how to construct PRCs with public encoding using our techniques, unlike the constructions of \cite{CG24,GG25}.

\heading{Robustness notions.}
Robustness can vary on two important axes: the type and number of errors the channel is allowed to introduce, and the information the channel knows.
The most basic notion of a PRC tolerates a constant rate of substitutions, and even constructing this object is challenging. 
However, ideally one wants to tolerate a constant rate of edits, allowing insertions and deletions in addition to substitutions.
Orthogonal to the kind of errors is the knowledge the channel has when choosing these errors.
A weak channel may make only random errors; a stronger channel may be computationally bounded, and choose errors without knowledge of the PRC secret key.
The strongest kind of channel is computationally unbounded and receives the secret key as input---this is the error model we consider in this work.
Note that while this channel is able to distinguish PRC codewords from random, this does not necessarily create an issue for robustness.

Two constructions go beyond substitution robustness. \cite{CG24} propose a binary PRC that is robust to a constant rate of \emph{random} deletions.
\cite{GM24} propose a polynomial-alphabet PRC that is robust to a constant rate of worst-case edits.
However, this large alphabet translates to a watermark that requires unrealistically high-entropy responses.
Accommodating realistic entropy rates requires a \emph{binary} PRC with robustness to a constant rate of edits; this is what we construct in this work.

The only paper we are aware of that considers PRCs for channels which may depend on the secret key in some way is \cite{DBLP:conf/stoc/AlrabiahACDG25}.
They defined a stronger notion of \emph{ideal} security for PRCs, where the channel is computationally bounded but can adaptively query encoding and decoding oracles.
They showed that the original PRC of \cite{CG24} can be easily transformed into one with ideal security.
\mir{I commented out the above--- if you open up the proof of ideal security, our construction is actually plausibly robust to a subexponential-time attacker with encoding and decoding oracle access. We don't need to rely on pseudorandomness.}
In contrast, we present PRCs which are robust even against an error channel that is computationally unbounded and knows the secret key.

\heading{On the hardness of constructing PRCs.}
Two recent concurrent works \cite{blackbox1,blackbox2} show that PRCs are difficult to construct in a formal sense. 
That is, they cannot be black-box constructed from a wide range of cryptographic primitives including random oracles, public-key encryption, and virtual black-box obfuscation.
Therefore particular cryptographic assumptions---such as our permuted codes conjecture or the prior sub-exponentially secure LPN \cite{CG24,GG25}, planted hyperloop \cite{GG25}, or weak PRF \cite{GM24} assumptions---are necessary.
It remains an interesting open question to construct sub-exponentially secure pseudorandom codes from more standard assumptions than the permuted codes conjecture, such as LPN or LWE or variants thereof (sparse LPN, dense--sparse LPN, ring LWE).

\heading{``Heuristic construction'' of \cite{CG24}.} In addition to their LPN-based PRC, \cite{CG24} mentions an alternate construction, which is ``heuristic'' in that it lacks strong evidence for pseudorandomness.
This construction involves taking a binary error-correcting code, permuting its symbols, and applying a constant rate of random substitutions.
\cite{CG24} does not specify the codes for which this construction should be pseudorandom, though they suggested investigating polar codes due to their high rate.

It is fairly immediate that pseudorandomness of their heuristic construction is equivalent to the permuted codes conjecture, when both are restricted to binary codes with polynomial dual distance.
The only difference is that \cite{CG24} omits the alphabet permutation. 
However, for binary codes the alphabet permutation is meaningless as it is simply a one-time pad, which can be removed by adding pairs of codewords.

Furthermore, our result that the alphabet permutation is necessary (\Cref{thm:rs-dist}) implies that the heuristic construction is insecure when instantiated with larger-alphabet codes (though \cite{CG24} considers it only for binary codes).

\section{Preliminaries}

  Let $\text{Bin}(n, p)$ denote the binomial distribution with $n$ independent trials and success probability $p$ of each trial.
  Let the \emph{substitution channel} $\channSC_\eta : \Sigma^\st \to \Sigma^\st$ be the randomized channel which, on input $x \in \Sigma^\st$, outputs a string $y \in \Sigma^\st$ of the same length as $x$ obtained by independently replacing each symbol of $x$ with a uniformly random symbol from $\Sigma$ with probability $\eta$. 
  For a set $X$, we let $S_X$ denote the symmetric group on $X$, i.e., the set of permutations $\pi : X \to X$. We say that a non-negative valued function $f(\lambda)$ of a security parameter $\lambda$ is \emph{negligible} (and write $f(\lambda) \leq \negl(\lambda)$) if $f$ decays faster than any polynomial in $\lambda$, i.e., $f(\lambda) \leq O(1/\poly(\lambda))$ for any polynomial $\poly$. While we state our guarantees with this ``sub-polynomial'' notion of decay, we emphasize that our constructions achieve \emph{sub-exponential} security assuming sub-exponential security  of the permuted codes conjecture (and this follows immediately from our proofs). 

  For a finite set $S$, we write $x_1, \ldots, x_k \gets S$ to mean elements $x_1 \ldots, x_k$ sampled uniformly at random from $S$, with replacement.
  For a distribution $D$, we write $x_1, \ldots, x_k \gets D$ to denote sampling $k$ i.i.d. elements from $D$.

  We will often use $\Sigma = \BF_q$, the finite field of size $q$.
  In cases when the algebraic structure is not important, we will sometimes refer to elements of $\BF_q$ as elements of $[q]$.
  
  \subsection{Pseudorandom codes}
  \begin{definition}
    Let $\Sigma$ be an alphabet and $\ME : \{0,1\}^* \times \Sigma^\st \to \Sigma^\st$ be an error channel on $\Sigma^\st$ that may depend on some auxiliary information $\sk \in \{0,1\}^*$. A \emph{(secret-key) pseudorandom code with strong adaptive robustness to $\ME$} is a tuple of polynomial-time randomized algorithms $(\KeyGen, \Enc, \Dec)$, where:
    \begin{itemize}
    \item (Syntax) For a security parameter $\lambda$ and functions $s(\lambda), n(\lambda)$ denoting the length of the secret key and the block length, we have: $\KeyGen: \{1^\lambda \} \to \{0,1\}^{s(\lambda)}$, $\Enc : \{0,1\}^{s(\lambda)} \to \Sigma^{n(\lambda)}$, and $\Dec: \{0,1\}^{s(\lambda)} \times \Sigma^{n(\lambda)} \to \{ \True, \False \}$.
    \item (\emph{Undetectability}; also referred to as \emph{Pseudorandonmess}) For any polynomial-time distinguisher $\Dist$, it holds that
      \begin{align}
      \left| \Pr_{\substack{\sk \gets \KeyGen(1^\lambda)}} \left[ \Dist^{\Enc(\sk)}(1^\lambda) = 1 \right] - \Pr_{\MU} \left[ \Dist^\MU(1^\secpar) = 1 \right] \right| \leq \negl(\lambda)\nonumber,
      \end{align}
      where $\MU$ denotes the uniform oracle which outputs a uniformly random string in $\Sigma^{n(\lambda)}$ each time it is called. 
    \item (\emph{Soundness}) For any fixed $y \in \Sigma^\st$, it holds that
      \begin{align}
\Pr_{\sk \gets \KeyGen(1^\lambda)} \left( \Dec(\sk, y) = \True \right) \leq \negl(\lambda)\nonumber.
      \end{align}
    \item (Strongly-adaptive robustness) For any $\lambda \in \BN$, we have that
      \begin{align}
      \Pr_{\substack{\sk \gets \KeyGen(1^\lambda) \\ x \gets \Enc(\sk)}} \left(\Dec(\sk, \ME(\sk, x)) = \False \right) \leq \negl(\lambda)\nonumber. 
      \end{align}
    \end{itemize}
    
  \end{definition}

  \subsection{Watermarking schemes}
  \label{sec:watermarking-intro}
  Recall that one of our main motivations behind pseudorandom codes is to obtain \emph{watermarking schemes} for autoregressive language models. An \emph{autoregressive language model} $\Model$ over some alphabet $\Sigma$ is a algorithm $\Model$ which takes as input a prompt $\prompt \in \Sigma^\st$ and a sequence of previous tokens $\tok_1, \ldots, \tok_{i-1}$ and outputs a distribution $\Model(\prompt, \tok_{1:i-1}) \in \Delta(\Sigma)$ over the next token $\tok_i \in \Sigma$. Given an integer $\ell \in \BN$, such an algorithm defines a distribution over sequences of tokens $\tok \in \Sigma^\st$, by repeatedly drawing $\tok_i \sim \Model(\prompt, \tok_{1:i-1})$ for $i \in [\ell]$. A \emph{watermarking scheme} $\MW$ for a language model $\Model$ consists of the following components (see \cref{sec:prc-watermarking-formal} for a formal treatment):
\begin{itemize}
\item $\Setup(\lambda)$, which takes as input a security parameter $\lambda$ and outputs a secret key $\sk$ of length polynomial in $\lambda$;
\item $\Wat(\sk, \prompt)$, which takes as input a secret key $\sk$ and a prompt $\prompt$ and outputs a sequence $\tok \in \Sigma^\ell$;
\item $\Detect(\sk, \tok)$, which takes as input a secret key $\sk$ and a sequence $\tok \in \Sigma^\ell$ and outputs a response in $\{ \True, \False \}$ indicating whether $\Detect$ detects the sequence $\tok$ as being watermarked according to $\sk$. 
\end{itemize}
Generally speaking, we want the watermarking scheme $\MW$ to have properties paralleling that of pseudorandom codes: in particular, we desire (a) \emph{soundness}, meaning that any fixed string is detected as watermarked with negligible probability (\cref{def:wat-sound}), (b) \emph{undetectability}, meaning that strings output by $\Wat(\sk, \prompt)$ are computationally indistinguishable from strings output by repeatedly sampling from $\Model$ (\cref{def:wat-undetect}), and (c) \emph{strong adaptive robustness}, meaning that any channel which corrupts the watermarked string by a small number of edits, \emph{even with knowledge of the secret key}, cannot fool the detection algorithm $\Detect$ (\cref{def:wat-robust}). 

In this work we omit discussion of multi-bit watermarks, although our results immediately yield watermarks with positive information rate.
See \cite[Sections 2.6 and 7.4]{CG24} for the applications of PRCs encoding multiple bits to watermarking with unforgeable public attribution.\footnote{Some works have referred to similar properties as ``public detection'' or ``publicly detectable watermarks.'' In \cite[Sections 2.6 and 7.4]{CG24} it is explained why unforgeable attribution must be treated as an independent property from standard watermark detection.}

\section{Technical overview}
We now overview the technical ideas behind our main contributions. In \cref{sec:comp-asm}, we describe the permuted codes assumption.
We explain how it relates to the ``permuted puzzles'' assumption of prior work, we show that it holds \emph{statistically} for a small number of samples if the alphabet is small, and we give an attack on large-alphabet schemes that do not use the alphabet permutation.
We then reiterate the need for new approaches to PRC constructions by showing a quasipolynomial-time attack against a general blueprint used in prior works.
In \cref{sec:edit-consequences,sec:intro-prc-watermarking}, we outline how this assumption yields binary pseudorandom codes with strongly-adaptive robustness to edits and sub-exponential security.

\subsection{The computational assumption: permuted codes}
\label{sec:comp-asm}
We introduce a family of assumptions which we collectively call \emph{permuted codes} assumptions. Roughly speaking, these assumptions posit that, for codes $C$ satisfying certain properties, taking a random codeword from $C$, randomly permuting it and adding noise is indistinguishable from simply outputting a uniformly random string. To formally state the assumption, we fix a set $\Sigma$ denoting the alphabet, an integer $n \in \BN$ denoting the block length, and let  $C \subseteq \Sigma^n$ denote an arbitrary set (which will typically be taken to be an error-correcting code with large dual distance). 
Given a real number $\eta > 0$, we let the \emph{substitution channel} $\channSC_\eta : \Sigma^\st \to \Sigma^\st$ be the randomized channel which, on input $x \in \Sigma^\st$, outputs a string $y \in \Sigma^\st$ obtained by independently replacing each symbol of $x$ with a uniformly random symbol from $\Sigma$ with probability $\eta$.
That is,
\[
    \channSC_{\eta}(x)_i =
    \begin{cases}
        x_i & \text{with probability $1-\eta$, and}\\
        \text{uniform in $\Sigma$} & \text{with probability $\eta$}.
    \end{cases}
\]
We define $\MD_{n, \Sigma, C, \eta, T}$ as follows:
\begin{enumerate}
    \item Sample random alphabet permutations $\pi_1, \dots, \pi_n \gets S_{\Sigma}$ and index permutation $\sigma \gets S_{[n]}$.
    \item Repeat $T$ times:
    \begin{enumerate}
        \item Sample $c \gets C$.
        \item Define $\hat{c}$ by $\hat{c}_i \gets \pi_i(c_{\sigma(i)})$ for $i \in [n]$.
        \item Output $\channSC_\eta(\hat{c})$.
    \end{enumerate}
\end{enumerate}

The permuted codes assumption says that $\MD_{n,\Sigma, C, \eta, T}$ is pseudorandom.
\begin{definition}[Permuted codes assumption] \label{def:permuted-codes}
    Let $\lambda \in \Z^+$ be a security parameter, and let $n = n(\lambda)$, $q = q(\lambda)$, $T = T(\lambda)$ be polynomially-bounded functions in $\lambda$ denoting the block length, alphabet size, and number of samples, respectively. 
    Let $C = C(\lambda) \subseteq \BF_q^n$ be a family of codes.
    The \emph{permuted codes assumption for $C$ with error $\eta$} states that $\MD_{n,\BF_q,C,\eta,T}$ is computationally indistinguishable from $\Unif((\BF_q^n)^T)$. 
\end{definition}
Note that the uniform distribution $\Unif((\BF_q^n)^T)$ is identical to the distribution $\MD_{n,\BF_q,\BF_q^n,\eta,T}$ where $C = \BF_q^n$ is the trivial code.
We conjecture that the permuted codes assumption holds with any constant error rate for any family of linear codes with polynomial dual distance.
\begin{conjecture}[General permuted codes conjecture] \label{conj:permuted-codes}
    Let $\lambda \in \Z^+$ be a security parameter, and let $n = n(\lambda)$, $q = q(\lambda)$, $T = T(\lambda)$ be polynomially-bounded functions in $\lambda$ denoting the block length, alphabet size, and number of samples, respectively.
    The \emph{permuted codes conjecture} states that if $C = C(\lambda) \subseteq \BF_q^n$ is any family of linear codes with dual distance $d = d(\lambda) = \lambda^{\Omega(1)}$ and $\eta = \Omega(1)$ is any constant error rate, then $\MD_{n,\BF_q,C,\eta,T}$ is computationally indistinguishable from $\Unif((\BF_q^n)^T)$.
\end{conjecture}

In \Cref{sec:sub-robustness}, we show that this conjecture immediately yields a binary PRC that is robust to a constant rate of \emph{substitutions}.
That is, let $C$ be an efficiently decodable binary code that has high dual distance.
Algorithm $\Enc$ of our PRC simply outputs permuted random noisy codewords, which are pseudorandom under the permuted codes assumption.
$\Dec$ inverts the permutation and applies the efficient decoder.
A far more significant challenge, which we discuss in the next section~(\Cref{sec:edit-consequences}) is achieving \emph{edit} robustness.

For our results, we will only need the permuted codes assumption for particular families of codes; below we state the specialization to the important case where $C$ is the family of \emph{Reed-Solomon codes} (\cref{def:rs}).
\begin{conjecture}[Permuted Reed-Solomon conjecture]
  \label{conj:permuted-rs}
Fix any constant $\eta$. For security parameter $\lambda$, let $C$ be the Reed-Solomon code $\RS_{\F_q,n,k}$ with $n = \lambda, q = \lambda - 1, k = \lambda^{1/5}$ (see \cref{def:rs}).\footnote{The choice of $k = \lambda^{1/5}$ is unimportant; it is straightforward to adjust our algorithm to accommodate any $k$ which grows polynomially with $\lambda$.}
Then the distributions $\MD_{n,\BF_q,C,\eta,T}$ and $\Unif((\BF_q^n)^T)$ are computationally indistinguishable for any $T = \poly(\lambda).$
\end{conjecture}

We also specialize \Cref{def:permuted-codes} to folded Reed-Solomon codes~(\Cref{conj:permuted-frs}), which we use to build edit-robust watermarks in~\Cref{sec:edit-robustness}.

\paragraph{Evidence: as a consequence of permuted puzzles}
\cref{def:permuted-codes} (and its specialization \cref{conj:permuted-rs}) are related to various precedents in the literature. The closest connection is to the \emph{permuted puzzles} assumption of \cite{boyle2021security,blackwell2021note}, which is similar to \cref{def:permuted-codes} but applies a single permutation to all of $[n] \times \Sigma$.  To formally state it, given $n \in \BN$, $C \subseteq \Sigma^n$, and $T, m \in \BN$, we define the distribution $\Dpp_{n, \Sigma, C, T, m}$ as follows:
\begin{enumerate}
    \item Sample a random permutation $\pi \gets S_{[n] \times \Sigma}$.
    \item Repeat $T$ times:
    \begin{enumerate}
        \item Sample $c \gets C$.
        \item Sample $i_1, \ldots, i_m \gets [n]$.
        \item Output $( \pi(i_1, c_{i_1}), \ldots, \pi(i_m, c_{i_m}))$.
    \end{enumerate}
\end{enumerate}

\begin{conjecture}[Permuted puzzles, ``General conjecture'' \cite{boyle2021security,blackwell2021note}] \label{conj:perm-puzz}
    Let $\secpar \in \BZ^+$ denote a security parameter, and let $n = n(\lambda), q = q(\lambda), T = T(\lambda), m = m(\lambda) \leq n(\lambda)$ be polynomially growing functions in $\lambda$ denoting the block length, alphabet size, number of samples, and number of symbols per codeword, respectively. If $C = C(\lambda) \subseteq \BF_q^n$ is any family of linear codes with dual distance $d = d(\lambda) = \lambda^{\Omega(1)}$ and $m \le (1 - \Omega(1)) \cdot n$, then $\Dpp_{n, \F_q, C, T, m}$ and $\Dpp_{n, \F_q, \F_q^n, T, m}$ are computationally indistinguishable.
\end{conjecture}
\cref{conj:perm-puzz} was originally introduced for the unrelated purpose of obtaining doubly-efficient private information retrieval.
Interestingly, we can show that the permuted codes conjecture (\cref{conj:permuted-codes}) \emph{is implied by} the permuted puzzles conjecture (\cref{conj:perm-puzz}), meaning that we can base the existence of PRCs on either.
In fact, the permuted puzzles conjecture is equivalent to the version of the permuted codes conjecture where the substitution channel is replaced by the erasure channel.

\begin{proposition}
  \label{thm:puzzles-from-codes}
   Suppose that the permuted puzzles conjecture (\Cref{conj:perm-puzz}) holds.
   Then the permuted codes conjecture (\Cref{conj:permuted-codes}) holds.
\end{proposition}
The proof of \cref{thm:puzzles-from-codes} proceeds by drawing enough samples from either $\Dpp_{n, \F_q, C, T, m}$ or $\Dpp_{n, \F_q, \F_q^n, T, m}$ to learn which elements $z \in [n] \times \Sigma$ never appear in the same sample $(\pi(i_1, c_{i_1}), \ldots, \pi(i_m, c_{i_m}))$. This allows us to convert samples from $\Dpp_{n, \F_q, C, T, m}$ or $\Dpp_{n, \F_q, \F_q^n, T, m}$ into samples from $\MD_{n, \F_q, C, \eta, T}$ or $\Unif((\BF_q^n)^T)$ respectively. See \cref{sec:perm-puzz-to-codes} for the full proof.

\paragraph{Evidence: statistical uniformity of a few samples.}
Next, we discuss a complementary piece of evidence for \cref{def:permuted-codes}, which results from asking the following question: \emph{If we take $T(\lambda)$ to be very small (e.g., $O(\log n)$), can we prove that $\MD_{n,\Sigma,C,\eta,T}$ and $\Unif((\Sigma^n)^T)$ are in fact statistically indistinguishable?}
When the size of the alphabet $\Sigma$ is large (e.g., a sufficiently large polynomial in $n$), it is easy to see that this is not the case, by a simple counting argument.
It is also straightforward to see that such statistical indistinguishability does not hold when $T = \omega(\log n)$, simply because there is not enough entropy in the sampling process of $\MD_{n,\Sigma,C,\eta,T}$.
However, when $|\Sigma| = O(1)$, there exists $T(\lambda) = \Omega(\log d)$ such that the two distributions in question are statistically close:
\begin{proposition}[Informal version of $q=O(1)$ case of \cref{cor:stat-evidence}]
\label{prop:stat-evidence-informal}
    Let $C \subseteq \BF_q^n$ be any code with dual distance $d$ and constant alphabet size $q$.
    Fix a constant $\eta > 0$.
    For sufficiently small constant $c$, any $T \leq c \cdot \log d$ satisfies
    \[
        \tvd{\MD_{n,\BF_q,C,\eta,T}}{\Unif((\BF_q^n)^T)} \le \exp(-\Omega(d)).
    \]
\end{proposition}
The proof of \cref{prop:stat-evidence-informal} proceeds by noting that, due to the permutation $\sigma : [n] \to [n]$ of the $n$ positions, for any distinguisher $\mathsf{Dist}$ between samples from either $\MD_{n,\BF_q,C,\eta,T}$ or $\MD_{n,\BF_q,\BF_q^n,\eta,T}$ we can write $\mathsf{Dist}$ in the following form. For each position $i \in [n]$, we may group the bits at position $i$ across the $T$ samples, and interpret the group as an element of $\BF_q^T$; the distinguisher $\mathsf{Dist}$, in turn, must depend only on the multinomial counts of the $n$ elements of $\BF_q^T$ corresponding to a codeword from either $\MD_{n,\BF_q,C,\eta,T}$ or $(\BF_q^n)^T$. 

In turn, \emph{arbitrary} algorithms which depend only on these counts turn out to be very simple: they can be written as length-$nT$, width-$\binom{n+q^T-1}{q^T-1}$ read-once branching programs, which follows from the fact that there are at most $\binom{n+q^T-1}{q^T-1}$ possibilities for what the collection of counts of the $n$ elements of $\BF_q^T$ can be. By considering a special form of the Fourier decomposition of such branching programs (\cref{lemma:forbes-kelley}) due to \cite{forbes2018pseudorandom}, we can show that the output of such programs are close under the distributions $\MD_{n,\BF_q,C,\eta,T}$ and $(\Unif(\BF_q^n)^T)$.
Very roughly, the reason is that the low-degree Fourier terms are identical due to the high dual distance of $C$, while high-degree Fourier terms are negligible due to the noise introduced by $\channSC_\eta$.
We note that our argument easily implies pseudorandomness against \emph{linear tests}, a benchmark sometimes used to judge new LPN-style cryptographic conjectures \cite{DBLP:conf/crypto/CouteauRR21,DBLP:conf/crypto/BoyleCGIKRS22,DBLP:conf/stoc/DingJK25}.
See \cref{sec:stat-evidence} for the proof.

While the setting of \cref{prop:stat-evidence-informal} only applies in the special case where $|\Sigma|$ is relatively small, as we observe in \cref{sec:sub-robustness}, the permuted codes assumption (\cref{def:permuted-codes}) with these parameter settings implies strongly-adaptive robust pseudorandom codes (with subexponential security, assuming such of \cref{def:permuted-codes}) against \emph{substitutions}, which is not known to be achievable under any other assumptions. To handle \emph{edits}, we need to rely on the permuted codes assumption for  (folded) Reed-Solomon codes with a large alphabet, which are not covered by the above result. 

\paragraph{Cryptanalysis: the alphabet permutation is necessary.}
The permuted codes assumption states that a distribution of ``scrambled'' codewords is pseudorandom.
Recall that this scrambling involves several components: First, the codeword indices are permuted by a index permutation $\sigma \in S_{[n]}$.
Second, an \emph{alphabet permutation} $\pi_i \in S_\Sigma$ is applied to each codeword symbol.
Finally, random substitution errors are applied to this doubly permuted codeword.

A natural question is whether all of these steps are necessary for pseudorandomness.
Prior work on the permuted puzzles conjecture gave a partial answer---a ``toy conjecture'' put forth in \cite{boyle2017can} asserted that the scrambled codewords were pseudorandom \emph{even without} the random substitution errors.
However, this was shown to be false in \cite{boyle2021security,blackwell2021note}, which suggested amending the conjecture to apply random erasures---this is exactly \Cref{conj:perm-puzz}.
The analogous question could be asked about the permutations.
With \emph{no} permutations, any efficiently decodable code would break the conjecture, with the decoding algorithm serving as a distinguisher.
But could the index permutation alone be enough?

We show that the alphabet permutations are necessary for the permuted codes assumption in \Cref{thm:rs-dist}.
That is, there exists a family of codes with polynomial dual distance that can be efficiently distinguished from random when codewords are position-permuted and subjected to a constant rate of random substitutions.
This code family is simply Reed-Solomon codes.

Suppose we are given $m$ position-permuted Reed-Solomon codewords $c^{(1)}, \ldots, c^{(m)}$ with a constant rate of substitutions.
Our attack attempts to find a low-degree nonzero multivariate polynomial $f$ such that $f(c^{(1)}_i, \ldots, c^{(m)}_i) = 0$ for all $i \in [n]$.
If $m$ is small enough, there will be a significant number of $i$'s with no noise across all codewords; these corresponding evaluation points take the form $(p_1(\alpha_i), \ldots, p_m(\alpha_i))$ for low-degree polynomials $p_1, \ldots, p_m$ and some unknown $\alpha_i$.
These points can be annihilated by a far lower-degree polynomial than one would expect of random points.
Therefore, for a proper choice of degree of $f$, this interpolation task is possible for the scrambled codewords but not for random points.

This attack can also be viewed as an attack against the McEliece cryptosystem instantiated with Reed-Solomon codes.
This attack uses only ciphertexts, and does not require the public key.

\paragraph{Cryptanalysis: a quasipolynomial-time attack against a broad class of constructions.}
To reiterate the need for new approaches to PRCs, we show that all constructions following a natural blueprint are vulnerable to quasipolynomial-time attacks.
In such constructions, codewords are binary strings that include a random string $r$ and a noise-tolerant predicate $f(r)$.
That is, codewords are of the form $\sigma(r||f(r)||y)$ where $\sigma \in S_{[n]}$ is fixed across all codewords, $r \gets \{0,1\}^\ell$, and $y$ is drawn from any distribution over $\{0,1\}^{n - \ell - 1}$.\footnote{Abusing notation a bit, we mean that the $i^{\text{th}}$ bit of the codeword is the $\sigma(i)^{\text{th}}$ bit of $(x || f(x) || y)$.}
$f$ satisfies $\Pr_{r, r' \gets \channSC_\eta(r)}[f(r) = f(r')] \geq 1/2 + 1/\poly(n)$, for some constant $\eta > 0$.

In other words, codewords have a planted relation between $r$ and $f(r)$, which is hidden by the permutation.
Since $f$ is noise-tolerant, this relation still holds with significant probability when noise is added, yielding a weak decoder (that can be boosted to a full decoder).
Several existing constructions fall into this class~\cite{CG24,GM24}.

Our attack is fairly simple: Since $f$ is noise-tolerant, it must have a Fourier coefficient at degree $t = O(\log n)$ with weight $n^{-O(\log n)}$.
Therefore, there is a quasipolynomial-time distinguisher that brute-force estimates all Fourier coefficients up to degree $t$.
We prove this in~\Cref{sec:quasipoly-dist}.

Note that our new binary constructions do not fall into this class. Instead, our codewords can be viewed as containing a random string $r \in \{0,1\}^\ell$, where $\ell$ is the dual distance of the code, and a \emph{longer-output function} $F: \{0,1\}^\ell \to \{0,1\}^{n - \ell}$.
No single bit of $F$ is noise-tolerant, and decoding uses the entire output of $F$.
Furthermore, polynomial dual distance of the codes used in our constructions prevents any similar attack.

\subsection{Edit-robust pseudorandom codes from permuted codes assumption}
\label{sec:edit-consequences}
In this section we explain how the permuted Reed-Solomon conjecture (\cref{conj:permuted-rs}), which is a special case of the permuted codes assumption (\cref{def:permuted-codes}), yields an edit-robust pseudorandom code. 

Fix a prime power $q$ denoting the alphabet size (so that $\Sigma = \BF_q$), and let $n = q-1$; we will consider the Reed-Solomon code $\RS_{\F_q,n,k}$ with $k = n^{1/5}$.\footnote{Any $k$ which grows polynomially with $n$ suffices for our applications.} We construct a pseudorandom code $\PRC$ as follows (see \cref{alg:prc-edit-binary} for the full description).

\paragraph{Key Generation.} The secret key $\sk$ consists of the permutations $\sigma \gets S_{[n]}, \pi_1, \dots, \pi_n \gets S_{\BF_q}$ as in the definition of $\MD_{n,\BF_q,C,\eta,T}$ (see \cref{sec:comp-asm}).

\paragraph{Encoding.} To produce a PRC codeword, we first generate one sample (corresponding to one index $t \in [T]$) from the distribution $\MD_{n, \BF_q, C, \eta, T}$, which may be expressed as a vector $z \in \BF_q^n$. 
Note that the data of $z$ can be be equivalently expressed as a list of tuples $(1, z_1), \ldots, (n, z_n)$. 
Roughly speaking, to output a PRC codeword, we sample $m$ indices $i_j \sim \Unif([n])$ and write the corresponding tuples $(i_j, z_{i_j})$ in binary, yielding a string in $\{0,1\}^{\log(nq)}$ for $j \in [m]$, and output their concatenation. (Technically, we need to take some additional care because if $i_j$ is sampled twice it will be followed by (binary) $z_{i_j}$ both times, which is unlikely for a random string. The necessary modifications are straightforward.)

\paragraph{Decoding.} The technical bulk of the algorithm (and the proof that it yields an edit-robust PRC) lies in the decoding algorithm. A key ingredient is the \emph{list recoverability} of Reed-Solomon codes, which generalizes the somewhat better-known notion of \emph{list decoding}. In particular, let us consider the Reed-Solomon code $C = \RS_{\BF_q, n,k}$ and fix parameters $\ell \in \BN, \zeta \in [0,1]$ satisfying $\zeta n \geq  \sqrt{k\ell n}$.  Then the code $C$ is efficiently $(\ell, \zeta)$-\emph{list recoverable}, which means that there is a $\poly(n)$-time algorithm which takes as input $n$ sets $\ML_1, \ldots, \ML_n \subset \BF_q$ satisfying $|\ML_i| \leq \ell$ for each $i \in [n]$, and outputs a list consisting of all codewords $c' \in C$ satisfying $c_i' \in \ML_i$ for at least $\zeta n$ values of $i \in [n]$. (In particular, the number of such codewords is at most $\poly(n)$; see \cref{thm:rs-listrecov} for a formal statement.)

Recall from the encoding algorithm above and the definition of $\MD_{n,\BF_q, C, \eta,T}$ (\cref{def:permuted-codes}) that a PRC codeword $y$ consists of the binary encoding of a collection of pairs $(i_j, \pi_{i_j}(c_{\sigma(i_j)}) + e_{i_j}) \in [n] \times \BF_q$ where $i_j \in [n]$ is an index, $\pi_{i_j}(c_{\sigma(i_j)})$ is a permuted symbol from a noisy Reed-Solomon codeword $c$, and $e_{i_j}$ is a random substitution error (which in particular is $0$ with probability $1-\eta$). 
The crucial insight is as follows: for any channel $\Eadv$ which corrupts its input $y$ by at most $\epedit \cdot |y|$ edits (for a sufficiently small constant $\epedit$) to some value $y'$, many pairs $(i_j, \pi_{i_j}(c_{\sigma(i_j)})+ e_{i_j})$ must remain ``relatively intact'' somewhere inside $y'$.
That is, many pairs will contain at most $O(\log(nq) \cdot \epedit)$ edits.
Moreover, for any string $w \in \{0,1\}^{\log(nq)}$, the number of strings $w'$ of edit distance at most $O(\log(nq) \cdot \epedit)$ from $w$ is at most $(nq)^{O(\epedit \log(1/\epedit))}$. Thus, given an input $y'$ to the decoding algorithm, we can ``brute force'' over all strings close in edit distance to each $O(\log(nq))$-length contiguous substring of $y$, and construct corresponding lists as input to a list recovery algorithm for $C$. Under an appropriate choice of parameters, the list recovery algorithm will output a nonempty set if in fact $y' = \Eadv(y)$ for some output $y$ of the PRC encoding algorithm. 

In more detail, the decoding algorithm acts as follows, given some input string $y'$. We first initialize lists $\ML_1, \ldots, \ML_n \gets \emptyset$ (which will ultimately be the input to a list recovery procedure for $C$). We inspect all contiguous substrings of length  $O(\log(nq))$ inside $y'$, and for each, compute all strings within edit distance $O(\log(nq) \cdot \epedit)$. In turn, for each of these strings, we interpret them as the binary encoding of a pair $(i, \pi_{i}(a))$ for some $a \in \BF_q$ and $i \in [n]$. We then add $\pi_i^{-1}(a)$ to the list $\ML_{\sigma(i)}$ (using our knowledge of $\pi_i, \sigma$ as encoded by the secret key $\sk$). Finally, we run the list recovery algorithm for $C$ on input $\ML_1, \ldots, \ML_n$. If the output list contains any codeword $c' \in C$, we output $\True$; otherwise, we output $\False$. See \cref{alg:prc-edit-binary} for the full description.

\paragraph{Desired properties of the PRC.} The robustness of the decoding algorithm follows from the above discussion: if $y' = \Eadv(y)$ for some PRC codeword $y$, then for many indices $i \in [n]$, the list $\ML_i$ will contain the correct symbol $c^*_i$ from the Reed-Solomon codeword $c^*$ corresponding to $y$ (since we will have ``guessed" it via the appropriate search over all $O(\log(nq) \cdot \epedit)$-edit close strings). Moreover, using our bound of $(nq)^{O(\epedit \log(1/\epedit))}$ on the size of an $O(\log(nq) \cdot \epedit)$-edit distance ball around any string in $\{0,1\}^{\log(nq)}$, we may bound the size of the lists $\ML_i$ by some parameter $\ell$ satisfying $\sqrt{k\ell n} = o(n)$. 
Therefore, by the list recoverability of Reed-Solomon codes, the list recovery algorithm will output $c^*$. See \cref{lem:list-preserve,lem:edit-robustness} for the formal arguments. 

Moreover, soundness of the PRC follows from a simple counting argument which, roughly speaking, computes the number of strings $y''$ which have enough contiguous substrings of length $O(\log(nq))$ which are $O(\log(nq) \cdot \epedit)$-edit close to some contiguous substring of the input $y'$ to the decoding algorithm; see \cref{lem:soundness}. 

Finally, the undetectability of the PRC construction follows directly from \cref{conj:permuted-rs}, since the encoding algorithm above may be interpreted as outputting a (randomized) function $F$ of one sample from $\MD_{n,\BF_q,C,\eta,T}$. Moreover, if we instead pass one sample from $\MD_{n,\BF_q,\BF_q^n,\eta,T}$ to this randomized function $F$, the output distribution may be verified to be the uniform distribution. Thus, any algorithm which distinguishes between PRC codewords and uniform strings may be used to distinguish between samples from $\MD_{n,\BF_q,C,\eta,T}$ and $\Unif((\BF_q^n)^T)$; see \cref{lem:edit-undetectability}. 

\subsection{Edit-robust watermarking schemes using folded Reed-Solomon codes}
\label{sec:intro-prc-watermarking}
Next, we discuss how the ideas from the previous section can be extended to obtain edit-robust watermarking schemes for autoregressive language models. As we will see, this requires PRCs that are robust to a very large fraction of edits, close to $\frac 1 2$. Above we constructed PRCs that are only robust to some small constant fraction of edits. However, we'll see we can extend the ideas to handle a large fraction of edits by relying on folded Reed-Solomon codes.

We begin by describing a generic transformation that converts any pseudorandom code into a watermarking scheme, due to \cite{CG24}. In particular, consider a language model $\Model$ (see \cref{sec:watermarking-intro}) and suppose we are given a pseudorandom code $\PRC = (\KeyGen, \Enc, \Dec)$ with some block length $n$. For simplicity, we assume that the alphabet $\Sigma$ for both $\Model$ and $\PRC$ is $\Sigma = \{0,1\}$, which is essentially without loss of generality (see Section 7.1 of \cite{CG24}). The $\Setup(1^\lambda)$ function of the watermarking scheme simply calls the $\KeyGen$ procedure for $\PRC$, which outputs some $\sk$. To implement the watermarking function $\Wat(\sk, \prompt)$, we first generate a PRC codeword $x \gets \Enc(\sk)$, $x \in \Sigma^n$, and then generate a sequence of tokens $\tok_{1:n}$ by attempting to generate each bit $\tok_i$ so as to maximize its probability of being equal to $x_i$: in particular, if we have that $\Model(\tok_i = 1 \mid \prompt, \tok_{1:i-1}) = p_i$, then we will draw $\tok_i$ from the distribution $\Ber(p_i - (-1)^{x_i}\cdot \min\{ p_i, 1-p_i \})$. This procedure is then repeated for multiple blocks of length $n$ to reach some desired output length $\ell$ (see \cref{alg:watermarking-edit}); for simplicity, we assume that $\ell = n$ in this discussion. It is straightforward to see that the resulting distribution of $\tok_{1:n}$ is identical to that produced by $\Model$ if $x$ is in fact uniformly random; thus, if $x$ is computationally indistinguishable from uniform, the output of $\Wat$ is computationally indistinguishable from that of $\Model$ (i.e., we have undetectability). 

Finally, the $\Detect(\sk, \tok)$ function of the watermarking scheme simply calls the $\Dec(\sk, \tok)$ function for $\PRC$, and outputs the same response. We begin by reviewing the argument in \cite{CG24} which shows that the resulting watermarking scheme is robust to a constant fraction of \emph{substitutions} (assuming that $\PRC$ has such substitution-robustness). To do so, we need to assume that \emph{entropy} of the output of $\Model$ is sufficiently high, i.e., that each token has at least entropy $\alpha$, on average. If this is the case, then the Hamming distance between the PRC codeword $x$ and the output $\tok$ of $\Wat$ is bounded above by $\left( \frac 12 - \Omega(\alpha^2) \right) \cdot n$ (see Lemma 21 of \cite{CG24}, reproduced as \cref{lem:ham-hemp}). Now consider  any substitution channel $\Eadv$ which corrupts at most a fraction $c_0 \cdot \alpha^2$ of its input, where $c_0$ denotes a sufficiently small universal constant, and let $\tok' = \Eadv(\tok)$. Then the Hamming distance between $x$ and $\tok'$ is bounded above by $ \left( \frac 12 - \Omega(\alpha^2) \right) \cdot n$. Thus, if $\PRC$ is robust to any \emph{substitution channel} which corrupts at most a fraction $\left( \frac 12 - \Omega(\alpha^2) \right)$ of its input, then the watermarking scheme is robust to $\Eadv$.

\paragraph{Our contribution: edit robustness via folded Reed-Solomon codes}
Can we extend the above argument to handle the case where $\Eadv$ may introduce a constant fraction of \emph{edits} (under the assumption that $\PRC$ has such robustness to edits)? While the above template naturally extends to handle edits, we need to address the following obstacle. The PRC watermarking scheme we described in \cref{sec:edit-consequences} is robust to any channel which introduces a $\epedit$-fraction of edits, \emph{for a sufficiently small constant $\epedit$}. However, in order to apply the above recipe for converting a PRC to a watermarking scheme, we need the PRC to be  robust to $\left( \frac 12 - \Omega(\alpha^2) \right)$-fraction of edits, for arbitrarily small constant $\alpha$. This would be a significantly stronger result than what we showed above! 

A closer inspection of the argument reveals that the issue is not merely with our analysis, i.e., using tighter bounds on the size of edit distance balls is insufficient to close this gap. In particular, for a given tuple $(i_j, \pi_{i_j}(c_{\sigma(i_j)}) e_{i_j}) \in [n] \times [q]$ corresponding to one permuted symbol of a codeword $c \in \BF_q^n$, the number of strings which could result after a fraction  $\frac 12 - \Omega(\alpha^2)$ of substitutions is at most $2^{\log(nq) \cdot (1 - \Omega(\alpha^4))} = (nq)^{1 - \Omega(\alpha^4)}$. Since, for the Reed-Solomon code, we must have $q \leq n$, this quantity approaches $n^2$ as $\alpha \to 0$, and in particular is $\omega(n)$. Aggregating over all $j \in [m] = \Theta(n)$, we see that the sum of the list sizes $\sum_{i=1}^n |\ML_i|$ for the list recovery procedure must be $n \cdot (nq)^{1 - \Omega(\alpha^4)} = n \cdot \omega(n) = \omega(n^2)$, which means the typical list $\ML_i$ will be of size $\omega(n)$; this is too large for the list recovery guarantee for Reed-Solomon codes to apply (see \cref{thm:rs-listrecov}).  

The core issue here is that we are ``wasting'' bits encoding the indices $i_j$ in the PRC encoding algorithm: if there were a way to avoid having to write them out explicitly, then the sum of the list sizes would be only $n \cdot n^{1 - \Omega(\alpha^4)}$, which means the typical list size would be $n^{1 - \Omega(\alpha^4)}= o(n)$, which would suffice for our needs. Unfortunately, our ability to write out these indices is crucial for edit robustness: if we had simply written the codeword symbols $c_{\sigma(i_j)}$ in order of e.g.~increasing $i_j$, an adversary could introduce insertions or deletions to shift the position of these symbols by $\Theta(n)$ positions, and there is not a clear way to recover from this. 

Instead, the main idea to get around this obstacle is encode multiple symbols in $\BF_q$ using a single index $i_j$. Equivalently, we replace the Reed-Solomon code $\RS_{\F_q,n,k}$ with a \emph{folded Reed-Solomon (FRS) code}, which is the code over alphabet $\BF_q^m$ obtained by grouping together collections of $m$ symbols from codewords of $\RS_{\F_q,n,k}$ (see \cref{def:frs}); here $m$ should be interpreted as a constant, and will be taken to be $\poly(1/\alpha)$ in the context of the preceding discussion. As a FRS code is not linear over its alphabet\footnote{In particular, if we identify $\BF_q^m \simeq \BF_{q^m}$, then the folded Reed-Solomon code is not a linear code over $\BF_{q^m}$. We remark that the FRS code is linear over $\BF_q$.} we cannot rely directly on the permuted codes assumption (\cref{def:permuted-codes}) to ensure undetectabilility of the resulting PRC. Nevertheless, the fact that the FRS code is linear over $\BF_q$ (and has dual distance decreased by a factor of only $m = O(1)$ compared to that of the corresponding Reed-Solomon code) leads us to the natural generalization of \cref{conj:permuted-rs}, in \cref{conj:permuted-frs} (in particular, \cref{conj:permuted-frs} is identical to \cref{conj:permuted-rs} except that the Reed-Solomon code is replaced with an FRS code with appropriate parameters).

Using a similar approach to our PRC construction with Reed-Solomon codes, we can show that a pseudorandom code based off of \cref{conj:permuted-frs} has robustness to any channel which introduces a $\left( \frac 12 - \Omega(\alpha^2) \right)$-fraction of substitutions and a $c_1 \cdot \alpha^2$ fraction of edits, for some sufficiently small constant $c_1 > 0$. This is sufficient to obtain a watermarking scheme for language models with per-token entropy at least $\alpha$: for a PRC codeword $x$ used in the watermarking procedure $\Wat$, at most a fraction $\left( \frac 12 - \Omega(\alpha^2) \right)$ of substitutions are introduced in $x$ to obtain the output $\tok$ of $\Wat$, and an adversary may introduce an additional $c_1 \cdot \alpha^2$ fraction of \emph{edits} before the string is passed to the detection algorithm $\Detect$. See \cref{lem:hamedit-robustness} and \cref{lem:wat-edit-robustness} for the formal arguments. 

\section{The permuted codes assumption}
\label{sec:perm-puzz-to-codes}
In this section, we discuss the permuted codes assumption (\cref{def:permuted-codes}) and perform some cryptanalysis: first, in \cref{subsection:pp-pc-comp}, we show that the permuted codes assumption is implied by the permuted puzzles assumption of \cite{boyle2021security,blackwell2021note} (\cref{conj:perm-puzz}). In \cref{sec:stat-evidence}, we give some statistical evidence for \cref{def:permuted-codes} when the alphabet size is constant by showing that logarithmically many permuted codewords are statistically uniform. Finally, in \cref{sec:alphabet-permutation}, we show that the permuted codes assumption \emph{fails} if the alphabet permutation is dropped. 

\subsection{Comparison to permuted puzzles} \label{subsection:pp-pc-comp}
We now show that the permuted puzzles assumption with $m = \Theta(n)$ implies the permuted codes assumption, for any constant $\eta > 1 - \frac{m}{n}$.

\begin{proof}[Proof of \cref{thm:puzzles-from-codes}]
    Let $C$ be as specified in \Cref{def:permuted-codes} and \Cref{conj:perm-puzz}, and let $m(\secpar) \geq \delta n(\secpar)$ for some constant $\delta \in (0,1)$.
    Let $\eta$ be a constant such that $\delta > 1 - \eta$.
    
    We present an efficient algorithm $\adv$ that takes in samples from $\Dpp_{n, \F_q, C, T, m}$ or $\Dpp_{n, \F_q, \F_q^n, T, m}$, and converts them to samples from $\Dpc_{n, \F_q, C, \eta, T}$ or $\Dpc_{n, \F_q, \F_q^n, \eta, T}$ respectively.

    $\adv$ is given $T \geq \secpar (nq)^2$ samples 
    \[
    \{(\alpha^{(i)}_1, \ldots, \alpha^{(i)}_m)\}_{i \in [T]},
    \]
    where each $\alpha^{(i)}_j \in [n] \times \F_q$.
    $\adv$ initializes $n$ empty sets $S_1, \ldots, S_n$.

    For each $S_\ell$, $\adv$ does the following:
    \begin{itemize}
        \item For all $\alpha \in [n] \times \F_q$:
        \begin{itemize}
            \item If $\alpha$ already appears in some set, continue.
            \item If $S_\ell$ is empty, add $\alpha$ to $S_\ell$.
            \item If there exists $\alpha' \in S_\ell$ such that $\alpha$ and $\alpha'$ do not appear together in any of the $T$ samples, add $\alpha$ to $S_\ell$.
        \end{itemize}
    \end{itemize}

We'll argue that each $S_\ell$ is exactly $\{\alpha \in [n] \times \F_q : \exists j \in \F_q \text{ s.t. } \pi(i, j) = \alpha \}$ for some $i \in [n]$, where $\pi$ is the random permutation in \Cref{conj:perm-puzz}.

We'll first show that all symbols added to $S_\ell$ correspond to the same index.
Let $S_\ell$ contain some $\alpha \in [n] \times \F_q$, where $\alpha = \pi(i, j)$.
Let $\alpha' = \pi(i', j')$ for $i' \neq i$.
For a uniform $c \gets C$, since $C$ has dual distance at least 2, the marginal distribution $(c_{i'}, c_i)$ is uniform and $c_{i'} = j'$ and $c_i = j$ with probability  $1/q^2$.
The same holds for a uniform $c$.
(Very coarsely), with probability at least $1/n^2$, $i$ and $i'$ are both included in the subsample in step (b).
Therefore, for each sample, $\alpha$ and $\alpha'$ appear together with probability at least $1/(nq)^2$.
The probability that $\alpha$ and $\alpha'$ do not appear together at any of the $T$ samples is at most $(1 - 1/(nq)^2)^{\secpar(nq)^2}$, which is negligible.

On the other hand, if two symbols correspond to the same index, they will not appear together in any samples and they are indeed added to the same $S_\ell$.
Therefore, each $S_\ell$ has size exactly $q$.

$\adv$ completes its conversion by drawing a random permutation $\pi_\adv \gets S_{[n]}$, and random bijections $\pi_1 : S_1 \to \F_q, \ldots, \pi_n : S_n \to \F_q$.
For $\alpha \in [n] \times \F_q$, let $\phi(\alpha) := \pi_\adv(\ell)$, where $\ell$ is the index of the set $S_\ell$ containing $\alpha$.

$\adv$ then constructs its $i^{\text{th}}$ sample for $i \in [T]$ as follows:
\begin{enumerate}
    \item Initialize a random $\hat{c} \gets \F_q^n$.
    \item Sample $k \gets \text{Bin}(n, 1-\eta)$. Let $I := \{\ell \in [n]: \ \exists j \in [n] \text{ s.t. } \ell = \phi(\alpha_j^{(i)})\}$. If $\abs{I} < k$, output ``fail.''
    \item Sample a random subset $I' \subseteq I$, where $\abs{I'} = k$.
    \item For all $j \in [m]$, let $\ell_j = \phi(\alpha^{(i)}_j)$. If $\ell_j \in I'$, let $\hat{c}_{\ell_j} \gets \pi_{\ell_j}(\alpha^{(i)}_j)$.
    \item Output $\hat{c}$.
\end{enumerate}

The bulk of this work is in converting the $m$ indices sampled without replacement in the permuted puzzles distribution, into the indices left error-free by the substitution channel in the permuted codes distribution.
In each iteration, $I$ is the set of (permuted according to $\pi_\adv$) distinct indices appearing in the sampled codeword.
$k$ is the number of error-free indices we want to have in the permuted-codes sample we are constructing.
$I'$ is the set of indices left error-free by the substitution channel; we chose $m$ so that $I$ will be larger than $I'$ with overwhelming probability.

We first show that $\adv$ outputs ``fail'' with negligible probability.
The expectation of $k$ is $(1-\eta)n$, where $(1-\eta)$ is a constant less than $\delta$.
There is some constant $\delta'$ such that $(1-\eta) < \delta' < \delta$ and 
by standard binomial tail bounds, $k \leq \delta' n$ with overwhelming probability in $n$ (and therefore in $\secpar$).
Now, we need to show that $I$ has size at least $\delta' n$.
$I$ is the set of distinct indices in $[n]$ that were included the $m$ samples that were chosen with replacement.
Since $m = \delta n$, it follows from a Chernoff bound that $\abs{I} \geq \delta' n$ with overwhelming probability.

Provided that $\adv$ did not output ``fail,'' we argue that this algorithm indeed converts samples as desired.
If $\adv$ was given samples from $\Dpp_{n, \F_q, C, T, m}$, this distribution is exactly $\Dpc_{n, \F_q, C, \eta, T}$.
That is, $\adv$ has learned which symbols correspond to the same index $i$ and has applied a random permutation $\pi_i$ to those symbols.
It has then permuted the indices under a random permutation $\pi_\adv$.
Similarly, if $\adv$ was given samples from $\Dpp_{n, \F_q, \F_q^n, T, m}$, its output is distributed according to $\Dpc_{n, \F_q, \F_q^n, \eta, T}$.
\end{proof}

\subsection{Statistical evidence for small alphabets}
\label{sec:stat-evidence}
In this section we show that the permuted codes conjecture with any $\Omega(1)$ rate of random substitution errors holds for any code $C \subseteq \F_q^n$ with high dual distance and constant alphabet size $q$.
This holds even without the alphabet permutations $\pi_1, \dots, \pi_n$.
Our argument is based on a decomposition of Forbes and Kelley \cite{forbes2018pseudorandom}; we follow \cite[Section 5.4]{hatami2023theory} for our proof.

A read-once branching program (ROBP) is just a graph on vertices $V_0 \cup \dots \cup V_n$ where each node in $V_i$ has two edges (corresponding to 0 and 1) going to $V_{i+1}$.
The program is evaluated on $x$ by starting at a designated start vertex in $V_0$, and following the edges corresponding to the bits in $x$.
The program outputs $f(x) = 1$ if it ends in an ``accept'' vertex; otherwise it outputs $f(x) = 0$.

For $v \in V_i$, we write $f_{v \to}(x)$ to denote the program evaluated on $x_{i+1}, \dots, x_n$, starting at vertex $v$.
We write $f_{\to v}(x)$ to denote the program evaluated on $x_1, \dots, x_i$, where the only accept state in $V_i$ is $v$.

Our main tool is the following convenient decomposition of ROBPs in the Fourier basis.
It was originally stated in the work of \cite{forbes2018pseudorandom} for the case $q=2$, but the same statement and proof easily generalize to arbitrary $q$.
In the following, let $\norm{z}$ denote the number of non-zero coordinates in $z \in (\Z/q\Z)^n$.

\begin{lemma}[\cite{forbes2018pseudorandom}, adapted from \cite{hatami2023theory}] \label{lemma:forbes-kelley}
Let $\chi_z$ be the Fourier characters of $(\Z/q\Z)^n$, defined by $\chi_z(x) = e^{2 \pi i \langle x, z \rangle / q}$.
Let $f : (\Z/q\Z)^n \to \{0,1\}$ be a length-$n$ width-$w$ standard-order ROBP with layers $V_0,V_1,\ldots,V_n$.
Then for any $k \in [n]$ we have
\[
    f(x)=L_k(x)+H_k(x),
\]
where
\[
    L_k(x)=\sum_{z \in (\Z/q\Z)^n:\,\norm{z}<k}\widehat{f}(z)\,\chi_z(x),
\]
and
\[
    H_k(x)=\sum_{i=1}^{n}\sum_{v\in V_i} H_{k,v}(x)\,f_{v\to}(x),
\]
where
\[
    H_{k,v}(x)=\sum_{\substack{z \in (\Z/q\Z)^n, \norm{z} = k \\ i \in \supp(z) \subseteq [i]}}\widehat{f_{\to v}}(z)\,\chi_z(x).
\]
\end{lemma}

\begin{proof}
The proof is identical to that found in \cite{hatami2023theory}.
Note that $i(z)$ defined there should again be the $k$th-smallest index appearing in $\supp(z)$.
\end{proof}

For any $\eta \in [0,1]$, let $T_{\eta}$ be the noise operator defined by $(T_{\eta} f)(x) = f(\channSC_{\eta}(x))$.
That is,
\[
    (T_{\eta} f)(x) = \E_{y \sim N_\eta(x)}[f(y)]
\]
where the distribution of $y \sim N_\eta(x)$ is as follows: for each $i \in [n]$,
\[
    y_i =
    \begin{cases}
        x_i & \text{with probability $1-\eta$, and}\\
        \text{uniform from $[q]$} & \text{with probability $\eta$}.
    \end{cases}
\]

\begin{theorem} \label{theorem:fool-robps}
    Suppose that $\MD$ is a $2k$-wise uniform distribution over $(\Z/q\Z)^n$.
    If $x \sim \MD$ and $y \from \channSC_{\eta}(x)$, then $y$ is $n w \cdot (1-\eta)^k$-pseudorandom to any length-$n$ width-$w$ ROBP $f$, i.e.,
    \[
        \abs{\E_{x \sim \MD} (T_{\eta} f)(x) - \E_{x \from (\Z/q\Z)^n} f(x)} \le nw (1-\eta)^k.
    \]
\end{theorem}
\begin{proof}
    Applying \Cref{lemma:forbes-kelley},
    \begin{align*}
        & \abs{\E_{x \sim \MD}[(T_{\eta} f)(x)] - \E_{x \from (\Z/q\Z)^n}[f(x)]} \\
        &= \abs{\E_{x \sim \MD}[(T_{\eta} L_k)(x) + (T_{\eta} H_k)(x)] - \hat{f}(\emptyset)} \\
        &= \abs{\E_{x \sim \MD} (T_{\eta} H_k)(x)} \\
        &= (1-\eta)^k \cdot \abs{\E_{x \sim \MD} \sum_{i \in [n]} \sum_{v \in V_i} (T_{\eta} f_{v \to})(x) \sum_{\substack{z \in (\Z/q\Z)^n, \norm{z} = k \\ i \in \supp(z) \subseteq [i]}} \hat{f}_{\to v}(z) \chi_z(x)} \\
        &\le (1-\eta)^k n w \cdot \max_{i, v} \E_{x \sim \MD} \abs{(T_{\eta} f_{v \to})(x) \sum_{\substack{z \in (\Z/q\Z)^n, \norm{z} = k \\ i \in \supp(z) \subseteq [i]}} \hat{f}_{\to v}(z) \chi_z(x)} \\
        &\le (1-\eta)^k n w \cdot \max_{i, v} \E_{x \sim \MD} \abs{\sum_{\substack{z \in (\Z/q\Z)^n, \norm{z} = k \\ i \in \supp(z) \subseteq [i]}} \hat{f}_{\to v}(z) \chi_z(x)}.
    \end{align*}
    The second equality follows from the fact that $x \sim \MD$ is $2k > k$-wise uniform.
    Now
    \begin{align*}
        & \E_{x \sim \MD} \abs{\sum_{\substack{z \in (\Z/q\Z)^n, \norm{z} = k \\ i \in \supp(z) \subseteq [i]}} \hat{f}_{\to v}(z) \chi_z(x)}^2 \\
        &= \E_{x \sim \MD} \sum_{\substack{z,z' \in (\Z/q\Z)^n, \norm{z} = \norm{z'} = k \\ i \in \supp(z), \supp(z') \subseteq [i]}} \hat{f}_{\to v}(z) \hat{f}_{\to v}(z')^* \chi_z(x) \chi_{z'}(x)^* \\
        &= \sum_{\substack{z \in (\Z/q\Z)^n, \norm{z} = k \\ i \in \supp(z) \subseteq [i]}} \abs{\hat{f}_{\to v}(z)}^2 \\
        &\le \E_{x \from \{-1,1\}^n} \abs{f_{\to v}(x)}^2 \\
        &\le 1,
    \end{align*}
    completing the proof.
    The second equality here follows from the fact that $x \sim \MD$ is $2k$-wise uniform, so every term vanishes except the ones where $z=z'$.
    The first inequality follows from Parseval's identity.
\end{proof}

With \Cref{theorem:fool-robps} in hand, it is easy to see that randomly permuting the symbols of \emph{any code with high dual distance and small alphabet} is sufficient to guarantee that a few noisy codewords are jointly uniform.

\begin{corollary}
  \label{cor:stat-evidence}
    If $C \subseteq \BF_q^n$ is any code with dual distance $d$ and $\eta>0$, then
    \[
        \tvd{\MD_{n,\BF_q,C,\channSC_{\eta},T}}{\MU_{\BF_q^n,T}} \le n T \cdot \binom{n+q^T-1}{q^T-1} \cdot (1-\eta)^{d/2}
    \]
    where $\MU_{\BF_q^n,T}$ is $T$ uniform samples from $\BF_q^n$.
    In particular, assuming $d = n^{\Omega(1)}$ and $\eta = \Omega(1)$:
    \begin{itemize}
        \item if $q = O(1)$, then there exists $T = \Omega(\log d)$ such that the above is $\exp(-\Omega(d))$; and
        \item if $q^T = o(d/\log n)$, then the above is $\exp(-\Omega(d))$.
    \end{itemize}
\end{corollary}
\begin{proof}
    Consider grouping together the symbols across the $T$ samples in $\MD_{n,\BF_2,C,\channSC_{\eta},T}$ as symbols in $\Sigma = \BF_q^T$.
    Because of the permutation, it suffices to show that the resulting string in $\Sigma^n$ has the same multinomial weight distribution as a uniformly random string in $\Sigma^n$.

    There are $\binom{n+q^T-1}{q^T-1}$ possibilities for what the multinomial weight (i.e., the collection of symbols and their counts) could be.
    Therefore we can compute it using a length-$n T$, width-$\binom{n+q^T-1}{q^T-1}$ ROBP.
    The main result then follows from applying \Cref{theorem:fool-robps}.
    Finally, the special cases listed follow straightforwardly from the main result.
\end{proof}

\paragraph{Further consequence: relation to low-degree conjecture.} Of independent interest, we discuss a consequence of \cref{theorem:fool-robps} for the \emph{low-degree conjecture}. For distributions $P, Q $ on $\{0,1\}^n$, the \emph{low-degree advantage} between them is defined, for $d \in \BN$, as 
\begin{align}
\Adv_{\leq d}(P, Q) := \max_{f : \text{deg-$k$ polynomial}} \frac{|\E_P[f] - \E_Q[f]|}{\sqrt{\mathbf{Var}_Q(f)}}\nonumber.
\end{align}
The low-degree advantage is used to capture distinguishability by low-degree polynomials.

A special case of the \emph{low-degree conjecture} (e.g., the $k=1$ case of \cite[Conjecture 2.1]{buhai2025quasipolynomial}; see also \cite{hopkins2018sos,kunisky2019notes}) states that if a permutation-symmetric distribution is only $o(1)$ distinguishable from uniform by low-degree polynomials, then it is $o(1)$-close to uniform in statistical distance.

Formally it states the following, where we let $\MU_n$ denote the uniform distribution on $\{0,1\}^n$:\footnote{Interestingly, \cite{buhai2025quasipolynomial}, which is concurrent and independent work from the present paper, shows that the low-degree conjecture is false for $k > 1$.}\noah{typically low-deg conj is stated for advantage $O(1)$, which just seems false to me, unless I'm missing something?}
\sam{I think it is correct with $O(1)$: You proved that if $\Adv = o(1)$ then the distinguishing advantage is $o(1)$, but they say that if $\Adv = O(1)$ then the distinguishing advantage is $1 - \Omega(1)$ (i.e. it cannot be $1-o(1)$). I've updated the proof below.}

\sam{Changed it from ``If $P_n$ is a family of distributions on $\{0,1\}^n$'' to the following. That wording really confused me, I thought it was a family for each $n$.}
\dan{I still find this hard to parse. The ``If'' condition requires indistinguishabiltiy for arbitrarily high degrer polynomials, since we allow $d$ to be large? I don't think that's what it was supposed to mean. What does the subscript $n$ mean for asymptotics - what are asymptotics over?  Maybe also add a text description of the result beforehand to help here. }\noah{I rephrased}
\begin{corollary}[Special case of the low-degree conjecture]
    \label{cor:low-degree-consequence}
Suppose that $d_n = \omega(\log n)$ is a sequence of integers indexed by $n \in \BN$ and that $P_n$ is a distribution on $\{0,1\}^n$ that is invariant under any permutation. 
For every fixed $\eta > 0$, there is no 
algorithm that distinguishes between a sample from $T_\eta P_n$ and a sample from $\MU_n$ with probability $1 - \Adv_{\leq d_n}(P_n, \MU_n) - o_n(1)$. 
\end{corollary} 
In particular, if $\Adv_{\leq d_n}(P_n, \MU_n) \leq o_n(1)$, then no algorithm can distinguish between the two distributions with probability $1-o_n(1)$. 
\begin{proof}[Proof of \cref{cor:low-degree-consequence}]
    Fix $n \in \BN$. 
A slight variant of \cref{theorem:fool-robps} with $q=2$, $k = d_n/2$, gives the following: for any length-$n$ width-$w$ ROBP $f$, 
\begin{align}
\left| \E_{x \sim P_n} (T_\eta f)(x) - \E_{x \sim \MU_n} f(x) \right| \leq nw (1-\eta)^{d_n/2} + \Adv_{\leq d_n}(P_n, \MU_n) \cdot \sqrt{\mathbf{Var}_{\MU_n}(L_k)} \label{eq:robp-consequence}.
\end{align}
To establish the above, we cannot directly apply \cref{theorem:fool-robps}, since $P_n$ is not $d_n$-wise uniform. However, we instead remark that the only place where \cref{theorem:fool-robps} uses that the distribution $P_n$ is $d_n$-wise uniform is to establish that $\E_{x \sim P_n}[(T_\eta L_k)(x) - \hat f(0)] = 0$. Instead, we have by the fact that $L_k$ and thus $T_\eta L_k$ is of degree at most $k < d_n$ that
\begin{align}
\left| \E_{x \sim P_n}[(T_\eta L_k)(x) - \hat f(0)] \right| =  \Adv_{\leq d_n}(P_n, \MU_n) \cdot \sqrt{\mathbf{Var}_{\MU_n}(L_k)}\nonumber,
\end{align}
which yields \cref{eq:robp-consequence} by propagating the above quantity through the proof. 

Next, Parseval's equality gives that $\E_{x \sim \MU_n}[L(x)^2] = \sum_{z \in \{0,1\}^n} \hat L(z)^2 = \sum_{z \in \{0,1\}^n: \| z \| \leq k} \hat f(z)^2 \leq 1$. 
Moreover, for $d_n = \omega(\log n)$ we may apply the argument of \cref{cor:stat-evidence} with $q = 2$ and $T = 1$ (corresponding to a single sample), so that the width $w$ of the ROBP $f$ may be taken to be $w = O(n)$, which yields
\begin{align}
    \tvd{T_\eta P_n}{\MU_n} \leq O(n^2) \cdot (1-\eta)^{d_n/2} + \Adv_{\leq d_n}(P_n, \MU_n) \leq o_n(1) + \Adv_{\leq d_n}(P_n, \MU_n)\nonumber . 
\end{align}
\sam{If we want $d_n = \Omega(\log n)$ then we would have to choose $\eta$ depending on $d_n$. Or choose which $d_n = \Omega(\log n)$ we use depending on $\eta$.}
\end{proof}
We remark that a variant of \cref{cor:low-degree-consequence} which shows that no algorithm can distinguish with advantage $1 - o_n(1)$ whenever $\Adv_{\leq d_n}(P_n, \MU_n) \leq O(1)$ (i.e., for any constant $O(1)$) may be obtained by using the variational characterization $$\chi^2(P_n, \MU_n) = \sup_{f : \{0,1\}^n \to \BR} \frac{(\E_{P_n}[f] - \E_{\MU_n}[f])^2}{\mathbf{Var}_{\MU_n}(f)}$$ of the $\chi^2$-divergence.

\subsection{The alphabet permutation is necessary}
\label{sec:alphabet-permutation}
Let $\ME : \Sigma^* \to (\Sigma \cup \{\bot\})^*$ be some error channel, and $T \in \BN$ be an integer denoting the number of samples. We define\footnote{nap for No Alphabet Permutation.} $\Dnap_{n, \Sigma, C, \ME, T}$ as follows:
\begin{enumerate}
    \item Sample random permutation\footnote{In contrast, the structured distribution considered in the permuted codes assumption (\Cref{def:permuted-codes}) involved additionally sampling random permutations $\pi_1, \ldots, \pi_n \gets S_\Sigma$. Permutation $\pi_i$ was applied to the $i^{\text{th}}$ codeword symbol.}  $\sigma \gets S_{[n]}$.
    \item Repeat $T$ times:
    \begin{enumerate}
        \item Sample $c \gets C$.
        \item Define $\hat{c}$ by $\hat{c}_i \gets c_{\sigma(i)}$ for $i \in [n]$.
        \item Output $\ME(\hat{c})$.
    \end{enumerate}
\end{enumerate}
We demonstrate our attack on the above distribution instantiated with Reed-Solomon codes (see \cref{def:rs}).

\begin{theorem} \label{thm:rs-dist}
    Let $C$ be the Reed-Solomon code $\RS_{\F_q,n,k}$ with $k \leq n^{1-c}$ for some constant $c \in (0,1)$.
    Let $\ME = \channSC_\eta$ be the substitution channel over $\Sigma^*$ with some error probability $\eta \le 1 - (3/2)^{-c/4}$.
    Then there is an efficient distinguisher achieving inverse polynomial distinguishing advantage between the distributions $\MD_{n,\BF_q,C,\ME,T}$ and $\MD_{n,\BF_q,\BF_q^n,\ME,T}$ for $T = 4/c$.
\end{theorem}

If $X \in \F^{T \times n}$, then define $X^r \in \F^{\binom{r+T}{r} \times n}$ to be the matrix of all coordinate-wise products of up to $r$ rows of $X$, including products with repeated rows.
You can think of $X^r$ as the matrix of monomials of total degree at most $r$ in the rows of $X$.

\begin{fact} \label{fact:rs-low-rank}
    Suppose that $p_1, \dots, p_T \in \F[t]$ are polynomials of degree at most $k$, and $t_1, \dots, t_\ell \in \F$.
    If $\bar{X}$ is the matrix where $\bar{X}_{i,j} = p_i(t_j)$ and $r$ is any positive integer, then $\rank(\bar{X}^r) \le kr+1$.
\end{fact}
\begin{proof}
    The row span of $\bar{X}^r$ consists of evaluations of polynomials of degree at most $kr$.
    But the space of all polynomials of degree at most $kr$ has dimension $kr+1$.
    Therefore the row span of $\bar{X}^r$ has dimension at most $kr+1$.
\end{proof}

We'll apply this fact with $\bar{X}$ being the ``clean'' submatrix of $X$---that is, the submatrix of columns that don't contain any errors.

\begin{fact} \label{fact:uniform-high-rank}
    If $X$ is a uniformly random matrix from $\F^{T \times n}$ and $r^n \cdot q^{\binom{r+T}{r}-n} \ge 1$, then
    \[
        \E[\rank(X^r)] \ge n \cdot \left(1 - \frac{\ln r}{\ln q}\right) - \frac{1}{\ln q}.
    \]
\end{fact}
\begin{proof}
    Let $q = \abs{\F}$ and $M = \binom{r+T}{r}$.
    For $\alpha \in \F^{\binom{r+T}{r}} \setminus \{0\}$, each coordinate of $\alpha X^r$ is a polynomial of degree at most $r$ in a disjoint set of coordinates of $X$.
    Therefore by Schwartz-Zippel, we have
    \[
        \Pr_{X\from\F^{T \times n}}[\alpha X^r = 0] \le \left(\frac{r}{q}\right)^n.
    \]
    Therefore the expected number of linear dependencies among rows of $X^r$ is at most $1 + r^n q^{M-n}$, so
    \begin{align*}
        \E[\rank(X^r)] &= \E[M - \log_q(\text{\# row dependencies})] \\
        &\ge M - \log_q \E[\text{\# row dependencies}] \\
        &\ge M - \log_q\left(1 + r^n q^{M-n}\right) \\
        &\ge M - \log_q\left(r^n q^{M-n}\right) - \frac{1}{\ln q} \\
        &= n \cdot \left(1 - \frac{\ln r}{\ln q}\right) - \frac{1}{\ln q}. \tag*{\qedhere}
    \end{align*}
\end{proof}

We are now ready to prove the main result of this section.

\begin{proof}[Proof of \cref{thm:rs-dist}]
    Collect our $T$ samples into a matrix $X \in \F_q^{T \times n}$.
    Our distinguisher will simply compute $\rank(X^r)$ for $r = n^{c/2}$.
    Let us compute the expected rank in both cases.

    \paragraph{The Reed-Solomon case.}
    In the Reed-Solomon case where $X \from \MD_{n,\BF_q,C,\ME,T}$, we have $X = \bar{X} + E$ where $E$ is all the noise from $\ME$ and $\bar{X}$ has no noise.
    \cref{fact:rs-low-rank} says that $\rank(\bar{X}^r) \le kr + 1$.
    Since each column of $E$ contains an error with probability $1 - (1-\eta)^T$, the expected number of columns containing any noise is at most $n - n \cdot (1-\eta)^T$.
    Therefore
    \[
        \E[\rank(X^r)] \le n - n \cdot (1-\eta)^T + kr + 1
    \]
    in the Reed-Solomon case.

    \paragraph{The random case.}
    In the random case where $X \from \MD_{n,\BF_q,\BF_q^n,\ME,T}$, \cref{fact:uniform-high-rank} says that
    \[
        \E[\rank(X^r)] \ge n \cdot \left(1 - \frac{\ln r}{\ln q}\right) - \frac{1}{\ln q}
    \]
    assuming that $r^n \cdot q^{\binom{r+T}{T}-n} \ge 1$.
    This condition holds because $\binom{r+T}{T} \ge r^T = n^2$.

    \paragraph{Putting it together.}
    Since the rank is bounded by $\binom{r+T}{T} \le n^{O(1)}$, it suffices to show that these expectations differ by 1.
    Indeed, using the given bounds on $k, \eta, T$,
    \begin{align*}
        n - n \cdot (1-\eta)^T + kr + 1 &\le n - n \cdot ((3/2)^{-c/4})^{4/c} + n^{1-c/2} + 1 \\
        &= n/3 + n^{1-c/2} + 1 \\
        \intertext{and}
        n - n \cdot \frac{\ln r}{\ln q} - \frac{1}{\ln q} &\ge n - n \cdot \frac{c}{2} \cdot \frac{\ln n}{\ln n} - \frac{1}{\ln n} \\
        &\ge n/2 - 1.
    \end{align*}
    This concludes the proof.
\end{proof}




This attack falls into a class of \emph{algebraic distinguishing attacks} following Arora-Ge \cite{arora2011new}, where one attempts to find a low-degree polynomial $P : \F_q^{nT} \to \F_q$ annihilating the given vectors.
That is, given $\v_1, \ldots, \v_T$ where each $\v_i \in \F_q^n$, one attempts to find a nonzero ``simple'' polynomial $P$, such that $P(\v_1, \ldots, \v_T) = 0$.
If the $\v_i$'s are random, such a polynomial should not exist.
For certain structured distributions of interest (such as ours), $P$ does exist, resulting in a successful distinguisher.
We refer the reader to \cite[Section 7]{benhamouda2025encrypted} for further discussion of such attacks.

\subsection{A quasipolynomial-time distinguisher for a natural class of constructions} \label{sec:quasipoly-dist} 
In this section, we show that any PRC of a certain form is vulnerable to quasipolynomial-time attacks.
The form is fairly general, and encompasses several prior constructions~\cite{CG24,GM24}:
each PRC codeword contains a uniformly random string $x \in \{0,1\}^*$, and some predicate $f(x)$.
The predicate is noise-tolerant in that if $x' \gets \channSC_\rho(x)$, $f(x') = f(x)$ with probability significantly greater than $1/2$.

We show in a formal sense that \emph{any} such construction relies on a planted $O(\log n)$-size structure, regardless of the choice of $f$, or how $x$ and $f$ are interspersed in a longer codeword.
Of course, the constructions in this paper are not of this form.

\begin{theorem} \label{thm:weakly-correlated}
Consider any PRC where codewords are of length $n$, and of the form $\sigma(x || f(x) || y)$.\footnote{Abusing notation a bit, we mean that the $i^{\text{th}}$ bit of the codeword is the $\sigma(i)^{\text{th}}$ bit of $(x || f(x) || y)$.}
Here, $\sigma \in S_{[n]}$ is any permutation fixed across all codewords, $x \gets \{0,1\}^\ell$, $f : \{0,1\}^\ell \to \{0,1\}$, and $y$ any string in $\{0,1\}^{n - \ell - 1}$ (possibly depending on $x$).
Suppose that for a constant $\rho \in (0, 1/2)$,
\[
    \Pr_{\substack{
    x \gets \{0,1\}^\ell\\
    x' \gets N_\rho(x)
    }}
[f(x') = f(x)] \geq \frac{1}{2} + \frac{1}{q}
\]
for some $q = O(\poly(n))$.

Then there exists an algorithm $\adv$ running in time $n^{O(\log n)}$, distinguishing between the uniform distribution and the distribution of codewords with non-negligible advantage.
\end{theorem}

\begin{proof}
We'll argue that there is a small set $S$ for which $\hat{f}(S) = n^{-O(\log n)}$.
Estimating $\hat{f}(S)$ for all small $S$ will yield our quasipolynomial-time distinguisher.

First, recall~(see, e.g., \cite{o2003computational}) that
\begin{align*}
    \Pr_{\substack{
x \gets \{-1,1\}^\ell\\
x' \gets N_\rho(x)
}}
[f(x') = f(x)] = 1 - \NS_{\rho/2}(f), 
\end{align*}
where 
\[
    \NS_{\rho/2}(f) = \frac{1}{2} - \frac{1}{2} \sum_{S \subseteq [\ell]} (1-\rho)^{\abs{S}} \hat{f}(S)^2.
\]
Therefore, 
\begin{align}
    \frac{1}{2} \sum_{S \subseteq [\ell]} (1-\rho)^{\abs{S}} \hat{f}(S)^2 \geq 1/q. \label{eq:sum-bounded-by-q}
\end{align}

Now, let $t = \log_{(1-\rho)} \frac{1}{q} = O(\log n)$ and consider breaking the sum up as follows:
\begin{align*}
    \frac{1}{2} \sum_{S \subseteq [\ell]} (1-\rho)^{\abs{S}} \hat{f}(S)^2 &= 
    \frac{1}{2} \sum_{\substack{S \subseteq [\ell]\\ \abs{S} \leq t}} 
    (1-\rho)^{\abs{S}} \hat{f}(S)^2 + \frac{1}{2} \sum_{\substack{S \subseteq [\ell]\\ \abs{S} > t}} (1-\rho)^{\abs{S}} \hat{f}(S)^2\\
    &\leq \frac{1}{2} \binom{\ell}{\leq t} \max_{\substack{S \subseteq [\ell]\\ \abs{S} \leq t}} \hat{f}(S)^2 + \frac{1}{2} (1-\rho)^{\abs{S}} \text{ (Parseval's identity)}\\
    &\leq \frac{1}{2}(\ell^t + 1)\max_{\substack{S \subseteq [\ell]\\ \abs{S} \leq t}} \hat{f}(S)^2 + \frac{1}{2q}.
\end{align*}
In order for this quantity to be at least $1/q$, we must have
\[
\max_{\substack{S \subseteq [\ell]\\ \abs{S} \leq t}} \hat{f}(S)^2 \geq \frac{1}{(\ell^t + 1) \cdot q}, \text{ where } \ell \leq n \text{ and } t = O(\log n).
\]

Recall our task of distinguishing between strings of the form $\pi(x, f(x), y)$ and $r \gets \{0,1\}^n$.
In the latter case, for all $S \subseteq [n]$ and all $i \notin S$, $\E_{r \gets \{-1, 1\}^n}[\chi_S(r) \cdot r_i] = 0$ (mapping $r$ to $\{-1, +1\}^n)$.
In the former case, we just showed that there is a set $S$ of size $O(\log n)$, and an index $i$, such that this expectation is $1/n^{O(\log n)}$.
Therefore, there is a distinguisher running in time $n^{O(\log n)}$ that estimates this expectation for all sets $S$ of size at most $t$, and all $i$.
\end{proof}

\section{Edit-robust pseudorandom codes}\label{sec:edit-robustness}

In this section, we present our main technical consequence of the permuted codes assumption, namely edit-robust PRCs over a binary alphabet with strong adaptive robustness (and subexponential security, assuming such for the permuted codes assumption). In \cref{sec:codes-prelim} we present some preliminaries from coding theory, and in \cref{sec:edit-prelim} we do the same for edit distance. In \cref{sec:main-prc,sec:prc-analysis} we present our main PRC construction and its analysis, and in \cref{sec:prc-watermarking-formal} we show how to derive analogous guarantees for watermarking from our PRC. 

\subsection{Preliminaries on folded Reed-Solomon codes}\label{sec:codes-prelim}
  Consider a field $\BF_q$, and integers $s,n$ with $n \leq q-1$ such that $n$ is divisible by $s$. Let $\gamma$ be a generator of $\BF_q^\star$, and $k \in \BN$.
  \begin{definition}[Reed-Solomon code]
    \label{def:rs}
    The \emph{Reed-Solomon code} $\RS_{\BF_q, n,k}$ is a linear code of dimension $k+1$ and block length $n$ over $\BF_q$.
    It consists of the vectors of evaluations of polynomials of degree at most $k$ with coefficients in $\BF_q$ on a set of distinct elements $x_1, \dots, x_n \in \F_q$:
    \[
        \RS_{\BF_q, n,k} = \{(p(x_1), \dots, p(x_n)) \mid p \in \BF_q[x], \deg(p) \le k\}.
    \]
  \end{definition}
  
  \begin{definition}[Folded Reed-Solomon code]
    \label{def:frs}
The \emph{folded Reed-Solomon code} $\FRS_{\BF_q, \gamma,n, s, k}$ is a code of block length $N := n/s$ over $\BF_q^s$. For a message $p(X)$, a polynomial over $\BF_q$ of degree at most $k$, for $0 \leq j < N$, the $j$th symbol of the encoding is $(p(\gamma^{js}), p(\gamma^{js+1}), \ldots, p(\gamma^{js + s-1}))$. 
\end{definition}

We will need to use the fact that (folded) Reed-Solomon codes have very good \emph{list recovery} algorithms, defined below. 
\begin{definition}[List recovery]
  Fix $\zeta \in [0,1]$ and $\ell, L \in \BN$. A code $C \subset \Sigma^n$ is \emph{$(\zeta, \ell, L)$-list recoverable} if for every sequence of sets $S_1, \ldots, S_n \subset \Sigma$ satisfying $|S_i| \leq \ell$ for all $i \in [n]$, there are at most $L$ codewords $c \in C$ for which $c_i \in S_i$ for at least $\zeta n$ values of $i \in [n]$.

For some $T = T(n)$,  we say that $C$ is \emph{$(\zeta, \ell)$-list recoverable in time $T$} if there is some $L = L(n) \leq O(T(n))$ and a $O(T(n))$-time algorithm which, given $c \in \Sigma^n$, finds the at most $L$ codewords $c'$ for which $c_i' \in S_i$ for at least $\zeta n$ values of $i \in [n]$.
The \emph{agreement threshold} is defined as $\trec = \zeta n$.
\end{definition}

A classical fact in coding theory, due to \cite{guruswami1998improved}, is that Reed-Solomon codes have good list recovery algorithms:
\begin{theorem}[\cite{guruswami2025essential}, Theorem 12.3.4]
  \label{thm:rs-listrecov}
  Consider the Reed-Solomon code $\RS_{\BF_q, n,k}$ as above. Then for any $\zeta \in (0,1), \ell \in \BN$ satisfying
  \begin{align}
\zeta n \geq \sqrt{(k-1) \ell n}\nonumber,
  \end{align}
  we can $(\zeta, \ell)$-list recover the code $\RS_{\BF_q, n,k}$ in time $\poly(n)$; that is, there exists an efficient list recovery algorithm with $\trec = \zeta n$. 
\end{theorem}

It is possible to obtained improved list recovery algorithms by considering folded Reed-Solomon codes (which come with the caveat that their alphabet is increased to size $q^s$). In the below theorem, we fix a field $\BF_q$ and integers $n,s,k,\gamma$ parametrizing a folded Reed-Solomon code $\FRS_{\BF_q, \gamma, n,s,k}$ as in \cref{def:frs}; thus its block length is $N = n/s$. 
\begin{theorem}[\cite{guruswami2007explicit}, Theorem 4.4 \& Eq.~(10)]
  \label{thm:frs-listrecov}
  Consider the folded Reed-Solomon code $\FRS_{\BF_q, \gamma, n,s,k}$ as above. Then for any $\zeta \in (0,1), \ell \in \BN$ satisfying
  \begin{align}
\zeta N \geq k \cdot (N\ell)^{1/(s+1)} + 2\nonumber,
  \end{align}
  we can $(\zeta, \ell)$-list recover the code $\FRS_{\BF_q, \gamma, n,s,k}$ in time $N^{O(s)}$. 
\end{theorem}

\paragraph{Pseudorandomness of permuted FRS codes.} In order to obtain PRCs with a binary alphabet robust to some constant fraction of edits, it will suffice for us to use the permuted codes conjecture for Reed-Solomon codes (\cref{conj:permuted-rs}). However, in order to obtain an analogous \emph{watermarking scheme} that works for language models whose per-token entropy is an arbitrary constant, we shall need to slightly modify our PRC construction, to use \emph{folded} Reed-Solomon codes. As such codes are not linear over their alphabet, we cannot directly apply the permuted codes assumption of \cref{def:permuted-codes} to argue that randomly permuted codewords are pseudorandom. Nevertheless, such codes are linear over $\BF_q$, and they have high dual distance with respect to the linear structure over $\BF_q$, in the sense that for a uniformly random codeword, its marginal distribution on any $k/s$ symbols is uniform. Therefore, we believe the following generalization of \cref{conj:permuted-rs} is plausible for such codes: 
\begin{conjecture}[Permuted FRS conjecture] \label{conj:permuted-frs}
    Fix any constant $\eta > 0$. Let $\lambda$ be a security parameter and $C$ be the folded Reed-Solomon code $\FRS_{\F_q,\gamma,n,s,k}$ with
    \begin{itemize}
        \item $\gamma$ any generator of $\F_q$,
        \item $n = \lambda$,
        \item $q = \lambda+1$,
        \item $s = O(1)$ any constant, and
        \item $k = \lambda^{1/(s+1)}$.
    \end{itemize}
   Writing $r = q^s$, we have that the distributions $\MD_{n,[r],C,\eta,T}$ and $\Unif(([r]^{n/s})^T)$ are computationally indistinguishable for any $T = \poly(\lambda).$
\end{conjecture}

\subsection{Preliminaries on edit distance}\label{sec:edit-prelim}
Below we present some basic preliminaries regarding edit-distance bounded channels. 
\begin{definition}
Given strings $z,z' \in \Sigma^n$, the \emph{edit distance} between $z$ and $z'$, denoted $\ED(z,z')$, is defined as the minimum number of insertions and deletions needed to transform $z$ into $z'$. 
\end{definition}

The Hamming distance is denoted $\Ham(z,z')$.

\begin{definition}
  Given $z \in \Sigma^n$ and a real number $p \in [0,1]$, we let
  \begin{align}
\EDball(z, p) := \{ z' \in \Sigma^\star \ : \ \ED(z,z') \leq p \cdot |z| \}\nonumber
  \end{align}
  denote the edit distance ball of radius $p \cdot |z|$ centered at $z$. 
  Given real numbers $\pham,\pedit \in [0,1]$, we let
  \begin{align}
\SEDball(z,\pham,\pedit) := \{ z' \in \Sigma^\star \ : \ \exists z'' \in \Sigma^\star, \ \Ham(z,z'') \leq \pham \cdot |z|, \ED(z'', z') \leq \pedit \cdot |z| \}\nonumber
  \end{align}
  denote the set of strings into which $z$ can be transformed by first making $\pham \cdot |z|$ substitutions and then making $\pedit \cdot |z|$ edits. 
\end{definition}

\begin{definition}[Edit-bounded channel]
  Fix real numbers $\epham, \epedit \in [0,1]$. A mapping (i.e., ``channel'') $\ME : \Sigma^\st \to \Sigma^\st$ is defined to be \emph{$\epedit$-edit bounded} if for all $z \in \Sigma^\st$, $\ME(z) \in \EDball(z, \epedit)$. We denote the set of $\epedit$-edit bounded channels by $\Eedit_{\epedit}$.

  A channel $\ME : \Sigma^\st \to \Sigma^\st$ is defined to be \emph{$(\epham,\epedit)$-substitution-edit bounded} if for all $z \in \Sigma^\st$, $\ME(z) \in \SEDball(z, \epham, \epedit)$. We denote the set of $(\epham, \epedit)$-substitution-edit bounded channels by $\Ehamedit_{\epham,\epedit}$. 
\end{definition}

\begin{lemma}
  \label{lem:ed-ball-bound}
For any $z \in \Sigma^n$ and $d \in \BN$, $|\EDball(z,d/n)| \leq  \left( \frac{e(n+d) \cdot (|\Sigma| + 1)}{d} \right)^{d} $. In particular, if $d = \ep n$ for some $\ep \in (0,1)$, then we have $|\EDball(z,\ep)| \leq (2e (|\Sigma| + 1))^{n \ep \log(1/\ep)}$. 
\end{lemma}
\begin{proof}
  There are at most ${n+d \choose d}$ ways to choose the positions at which to make insertions and deletions; for each one, one can either make a deletion or an insertion of any of $|\Sigma|$ characters. Overall, the number of ways to make a sequence of $n+d$ edits is thus bounded above by
  \begin{align}
{n+d \choose d} \cdot (|\Sigma| + 1)^{d} \leq \left( \frac{e(n+d) \cdot (|\Sigma| + 1)}{d} \right)^{d} \nonumber,
  \end{align}
  which, in the case that $d = \ep n$ for some $\ep \in (0,1)$, may be bounded above by $(2e (|\Sigma| + 1))^{n \ep \log(1/\ep)}$.
\end{proof}

\begin{lemma}
  \label{lem:sed-ball-bound}
  For any $z \in \{0,1\}^n$ and $p, \ep \in (0,1)$, we have $|\SEDball(z, 1/2 -p, \ep)| \leq 2^{n \cdot (1 - p^2 + \log(6e) \ep \log(1/\ep))}$.
\end{lemma}
\begin{proof}
  The number of $z''$ for which $\Ham(z, z'') \leq (1/2 -p) \cdot n$ may be bounded above by
$
2^{n \cdot h_2(1/2 - p)} \leq 2^{n \cdot (1 - p^2)},
$
where $h_2(\cdot)$ denotes the binary entropy. Using \cref{lem:ed-ball-bound} with $|\Sigma| = 2$ we see that
\begin{align}
|\SEDball(z, (1/2 - p), \ep)| \leq 2^{n\cdot (1-p^2)} \cdot (6e)^{n \ep \log(1/\ep)} = 2^{n \cdot (1 - p^2 + \log(6e) \ep \log(1/\ep))}\nonumber.
\end{align}
\end{proof}

\begin{algorithm}
  \caption{Edit-robust PRC with binary alphabet}
  \label{alg:prc-edit-binary}
  \begin{algorithmic}[1]\onehalfspacing
    \Require Functions $q = q(\lambda), n = n(\lambda), k = k(\lambda), C = C(\lambda) \subset [q]^n$ denoting alphabet size, block length, code dimension, and the code, respectively (all a function of the security parameter $\lambda$). We assume $n(\lambda), q(\lambda)$ are powers of 2 for all $\lambda$. \emph{Additional parameters:} $m, \pDec, \epDec, \Lmax, \eta$, and a list recovery algorithm $\ListRecov$ for $C$ with agreement threshold $\trec$.  
    \Function{$\KeyGen$}{$1^\lambda$}
    \State \emph{Write $n = n(\lambda)$.}
    \State Let $\sigma \gets \Unif(S_{[n]}), \pi_1, \ldots, \pi_n \gets \Unif(S_{[q]})$ be chosen uniformly at random. 
    \State Choose $o \sim \Unif([q]^n)$. 
    \State \Return $(\sigma, \pi_1, \ldots, \pi_n, o)$. 
    \EndFunction

    \Function{$\Enc$}{$1^\lambda, \sk$}
    \State Write $\sk = (\sigma, \pi_1, \ldots, \pi_n, o)$. 
    \State Draw $c \gets \Unif(C)$ and write $c' \gets \channSC_\eta(c)$. \emph{($\channSC_\eta$ is substitution channel w/ noise rate $\eta$.)}\label{line:draw-x-e}
    \State Set $z \gets c' + o \pmod{q} \in [q]^{n}$.\label{line:draw-z}
    \State Sample i.i.d.~indices $i_1, \ldots, i_m \gets \Unif([n])$.\label{line:draw-w}
    \For{$j = 1, 2, \ldots, m$}
    \If{$i_j \neq i_{j'} \forall j' < j$}
    \State Define $z_j' := \pi_{i_j}(z_{\sigma(i_j)})$.
    \Else
    \State Sample $z_j' \gets \Unif([q])$.
    \EndIf
    \EndFor
    \State Let $y$ denote the concatenation: $\bin(i_1)\circ \bin(z_1')\circ \bin(i_2)\circ \bin(z_2')\circ \ldots\circ \bin(i_m)\circ \bin(z_m')$.
    \State \Return $y \in \{0,1\}^{m \cdot (\log q + \log n)}$. 
    \EndFunction

    \Function{$\Dec$}{$1^\lambda, \sk, y$}
    \State Write $\sk = (\sigma, \pi_1, \ldots, \pi_n, o)$, $\ell := \log q + \log n$.
    \State For each $j \in [n]$, initialize $\ML_j \gets \emptyset$. 
    \For{$1 \leq j \leq |y| - \ell$}\label{line:for-inds}
\State \label{line:recover-inds}For each pair $(i,z) \in [n] \times [q]$ for which the below holds, add $\pi_i^{-1}(z)$ to $\ML_{\sigma(i)}$:
\begin{align}
y_{j:j+\ell-1} \in \SEDball(\bin(i) \circ \bin(z), 1/2-\pDec, \epDec)\nonumber.
\end{align}
\EndFor
\State For each $i \in [n]$ with $|\ML_i| \geq \Lmax$, remove $|\ML_i| - \Lmax$ arbitrary elements of $\ML_i$.\label{line:remove-list-items}
\If{$\ListRecov(\trec, (\ML_i - o_i)_{i \in [n]}) \neq \emptyset$} \emph{($\ML_i - o_i$ indicates shifting $\ML_i$ by $-o_i$.)}\label{line:find-close-codewords}
\State \Return True
\Else
\State \Return False
\EndIf 
    \EndFunction
  \end{algorithmic}

\end{algorithm}

\subsection{Our PRC construction: \cref{alg:prc-edit-binary}}
\label{sec:main-prc}
\cref{alg:prc-edit-binary} shows our main PRC construction. It takes as input a code $C(\lambda) \subset [q(\lambda)]^{n(\lambda)}$ (for a security parameter $\lambda$), as well as additional parameters $\pDec, \epDec, \Lmax, \eta, \trec$, as discussed further below. The key generation algorithm $\KeyGen(1^\lambda)$ takes as input the security parameter $\lambda$ and returns a key $\sk = (\sigma, \pi_1, \ldots, \pi_n, o)$ consisting of permutations $\sigma, \pi_1, \ldots, \pi_n$ used to permute codewords, as well as a one-time pad $o$ (which is used to ensure soundness of the PRC). 

The encoding algorithm $\Enc(1^\lambda, \sk)$ draws a uniformly random codeword from $C$, and then permutes it according to $\sk$ and adds noise $\eta$ in a way consistent with the permuted codes assumption (\cref{conj:permuted-codes}). Finally, it outputs a subset of the symbols in the resulting (permuted and noisy) codeword, where each symbol is accompanied by its index, and both are written in binary.

Finally, the decoding algorithm $\Dec(1^\lambda, \sk, y)$ takes as input the key $\sk$ and a string $y$ and applies a list recovery algorithm (using the parameters $\Lmax, \trec, \pDec, \epDec$) to determine if $y$ is close to some codeword in $C(\lambda)$. We refer the reader to \cref{sec:edit-consequences} for further intuitive explanations regarding the design of the encoding and decoding algorithms.

\paragraph{Parameter settings.} We consider two distinct settings of parameters for our edit-robust pseudorandom codes. The first setting of parameters is used to obtain the guarantee of strong adaptive robustness with respect to any $\epedit$-edit bounded channel:
\begin{definition}[Edit parameters]
  \label{def:edit-params}
  Given $\epedit > 0$, we define the following parameters of \cref{alg:prc-edit-binary}, so as to obtain robustness to $\epedit$-edit distance bounded channels:
\begin{itemize}
\item For a security parameter $\lambda$, we choose the parameters of the code as follows:
  \begin{align}
    n(\lambda) = \lambda,\quad q(\lambda) = \lambda-1, \quad k(\lambda) = n^{1/5}\nonumber.
  \end{align}
  To simplify notation, we write $q = q(\lambda), n = n(\lambda), k = k(\lambda)$, and we let $\gamma = \gamma(\lambda)$ denote a generator of $\BF_q$. We next define the code $C = C(\lambda)$ by $ C := \RS_{\BF_q, n,k}$.
\item $m = 4n^{4/5}$.
\item $\CDec = 16$.
\item $\eta = 1/32$.
\item $\Lmax = n^{2/5}$.
\item $\epDec = 2\epedit \CDec$, $\pDec = 1/2$. 
  \item $\trec = \sqrt{k\Lmax n}$. 
  \item $\ListRecov$ is the algorithm described in~\Cref{thm:rs-listrecov}.
\end{itemize}
\end{definition}
The second setting of parameters is used to obtain the guarantee of strong robustness with respect to any $(1/2 - \pSub,\epedit)$-substitution-edit bounded channel, which will in turn be needed for our watermarking application in the \cref{sec:prc-watermarking-formal}. 
\begin{definition}[substitution-edit parameters]
  \label{def:hamming-edit-params}
Given $\pSub \in (0,1/2)$, we define the following parameters for \cref{alg:prc-edit-binary}, so as to obtain robustness to $(1/2-\pSub, \epedit)$-substitution-edit distance bounded channels, for an appropriate choice of $\epedit$ (specified below): 
\begin{itemize}
\item $s = 8/\pSub^2$ (representing the folding parameter for the FRS code).
\item For a security parameter $\lambda$, we choose the parameters of the code as follows:
  \begin{align}
    \bar n(\lambda) = \lambda,\quad \bar q(\lambda) = \lambda-1, \quad k(\lambda) = n(\lambda)^{1/(s+1)}, \quad n(\lambda) = \bar n(\lambda)/s, \quad q(\lambda) = \bar q(\lambda)^s\nonumber.
  \end{align}
  To simplify notation, we write $q = q(\lambda) n = n(\lambda), k = k(\lambda)$, and we let $\gamma = \gamma(\lambda)$ denote a generator of $\BF_q$. We next define the code $C = C(\lambda)$ by $ C := \FRS_{\BF_q,\gamma, n,s,k}$.
  Above $\bar n(\lambda), \bar q(\lambda)$ should be interpreted as the block length and alphabet size of the underlying Reed-Solomon code. 
\item $m = n^{s/(s+1)}$. 
\item $\CDec = 16/\pSub $. 
\item $\eta = \pSub/32$. 
\item $\Lmax = n^{s-3}$.
\item $\epedit = c_0 \cdot \pSub^3 \log(1/\pSub)$, where $c_0$ is a universal constant chosen sufficiently small so that $\epedit \log(1/\epedit) \log(6e) \leq \pSub^3/32$.
\item $\pDec = \pSub/2$, $\epDec = 2\epedit \CDec$. By our choice of $\epedit$ and $c_0$ above, it follows that $\epDec \log(6e) \log(1/\epDec) \leq \pSub^2/2$.
\item $\trec = k \cdot (n\Lmax)^{1/(s+1)}+2$. 
\end{itemize}
\end{definition}

\subsection{Analysis of \cref{alg:prc-edit-binary}}
\label{sec:prc-analysis}
In this section, we will establish the following theorem which shows that the PRC of \cref{alg:prc-edit-binary} obtains all of our desired properties:
\begin{theorem}[Main edit-robust PRC result]
\label{thm:prc-main}
Fix any $\epedit \leq 1/400$. Then the PRC construction of \cref{alg:prc-edit-binary} with the parameter settings of \cref{def:edit-params} is undetectable, sound, and has strong adaptive robustness to any $\epedit$-edit bounded channel, where undetectability relies on \cref{conj:permuted-rs}. Moreover, the constituent algorithms $\KeyGen, \Enc, \Dec$ all run in $\poly(\lambda)$ time, where $\lambda$ denotes the security parameter. 
\end{theorem}
\begin{proof}
Strong adaptive robustness to $\epedit$-edit bounded channels follows from \cref{lem:edit-robustness}. Soundness follows from \cref{lem:soundness}. Undetectability, under \cref{conj:permuted-rs}, follows from \cref{lem:edit-undetectability}.
\end{proof}
Additionally, to obtain a watermarking scheme in \cref{sec:prc-watermarking-formal}, we will need the following analogous theorem which applies to any substitution-edit bounded channel with appropriate parameters:
\begin{theorem}
\label{thm:prc-hamming-edit}
Fix any $\pSub \in (0,1/2)$. Then the PRC construction of \cref{alg:prc-edit-binary} with the parameter settings of \cref{def:hamming-edit-params} is undetectable, sound, and has strong robustness to any $(1/2 - \pSub, \epedit)$-substitution-edit bounded channel, for $\epedit = \tilde\Theta(\pSub^3)$, where undetectability relies on \cref{conj:permuted-frs}. Moreover, the constituent algorithms $\KeyGen, \Enc$ run in $\poly(\lambda)$ time,\footnote{In particular, this $\poly(\lambda)$ does not depend exponentially on $1/\pSub^2$.} and $\Dec$ runs $\lambda^{O(1/\pSub^2)}$ time, where $\lambda$ denotes the security parameter.
\end{theorem}
\begin{proof}
  Strong adaptive robustness to $\epedit$-edit bounded channels follows from \cref{lem:hamedit-robustness}. Soundness follows from \cref{lem:soundness}. Undetectability, under \cref{conj:permuted-rs}, follows from \cref{lem:edit-undetectability}. The bound on the running time follows immediately from \cref{thm:frs-listrecov}. 
  \end{proof}
We remark that one downside of \cref{thm:prc-hamming-edit} is that the PRC's decoding algorithm runs in time that grows exponentially in $1/\pSub^2$; this is due to our use of folded Reed-Solomon codes, which require time $n^{O(s)}$ to list recover (where $s = O(1/\pSub^2)$ is the folding parameter). It is an interesting question to obtain an improved construction which does not suffer this decay. 

We proceed with the proofs of the individual lemmas used to establish \cref{thm:prc-main,thm:prc-hamming-edit}. Our first lemma gives an upper bound on the size of the lists $\ML_j$ constructed in the decoding algorithm $\Dec$.
\begin{lemma}
  \label{lem:list-ub}
  The lists $\ML_j$ at termination of the for loop on \cref{line:for-inds} of $\Dec(1^\lambda, \sk, y)$ satisfy the following:
  \begin{align}
    \sum_{j=1}^n |\ML_j| \leq \begin{cases}
      |y| \cdot (nq)^{1 - \pDec^2 + \log(6e) \epDec \log(1/\epDec)} &: \pDec \in [0, 1/2) \\
      |y| \cdot (nq)^{\epDec \log(6e) \log(1/\epDec)} &  :\pDec = 1/2\nonumber.
      \end{cases}
  \end{align}
\end{lemma}
\begin{proof}
  Let the input to $\Dec$ be denoted $y$. For each $j \in [|y| - \ell]$, the number of tuples $(w,z)$ satisfying the conditions on \cref{line:recover-inds} may be bounded above using \cref{lem:sed-ball-bound} (for any $p \in [0,1/2]$) or, in the case that $p = 1/2$, by the tighter bound in \cref{lem:ed-ball-bound} (where the block length $n$ in each of these lemmas is taken to be $\ell = \log(nq)$).  This yields the bound claimed in the lemma statement. 
\end{proof}

Next, we show that for any string $y$ which is output by $\Enc$, if it is passed through an edit-bounded channel (or more generally, an edit-substitution bounded channel), then many of the lists $\ML_i$ constructed in $\Dec$ will still contain correct codeword symbols corresponding to the codeword generated by $\Enc$. 
\begin{lemma}
  \label{lem:list-preserve}
  Consider any $i_1, \ldots, i_m \in [n], z_1', \ldots, z_m' \in [q]$, and let $y \in \{0,1\}^{m \cdot \ell}$ denote the concatenation $\bin(i_1)\circ \bin(z_1')\circ \ldots\circ \bin(i_m)\circ \bin(z_m')$. Fix any $y' \in \SEDball(y, 1/2 - \pSub, \epED)$, and any value of $\sk$. 
  For $j \in [m]$, letting $\ML_{i_j}$ denote the value of the list in $\Dec(1^\lambda, \sk, y')$ at the conclusion of the function, then the number of values of $j \in[ m]$ for which $z_j' \in \ML_{i_j}$ is at least 
  \begin{align}
      m \cdot \left( \frac{\pSub}{2} - \frac{1}{\CDec} - \frac{2\ell \cdot (nq)^{1 - \pSub^2 + \epDec \log(6e) \log(1/\epDec)}}{\Lmax} \right)\quad \mbox{ if } \quad \pSub < 1/2 \nonumber\\
      m \cdot \left( 1 - \frac{1}{\CDec} - \frac{2\ell \cdot (nq)^{\epDec \log(6e) \log(1/\epDec)}}{\Lmax} \right)\quad \mbox{ if } \quad  \pSub = 1/2\nonumber.
  \end{align}
\end{lemma}
\begin{proof}
  Fix $y'' \in \{0,1\}^\st$ so that $\Ham(y, y'') \leq (1/2 - \pDec) \cdot m\ell$ and $\ED(y', y'') \leq \epDec \cdot m\ell$. For each $j \in [m]$, we define the following quantities:
  \begin{itemize}
  \item Let $f_j \geq 0$ denote the number of substitutions made at positions $(j-1) \, \ell + 1, \ldots, (j-1) \, \ell + \ell$ when transforming $y$ to $y''$. 
  \item Let $e_j \geq 0$ denote the number of edits (i.e., insertions and deletions) made at positions $(j-1) \, \ell + 1, \ldots, (j-1) \, \ell + \ell$ when transforming $y''$ to $y'$ (through an optimal sequence of edits).
  \end{itemize}
  Using the bounds on $\Ham(y,y'')$ and $\ED(y'', y')$, we have
  \begin{align}
f_1 + \cdots + f_m \leq (1/2 - \pSub) \cdot m\ell, \qquad e_1 + \cdots + e_m \leq \epED \cdot m\ell\nonumber.
  \end{align}
  Then the following statements are immediate:
  \begin{itemize}
  \item For at least $m \cdot (1-1/\CDec)$ values of $j \in [m]$, we have that $e_j \leq \epED \cdot \ell \CDec$.
  \item For at least $m \cdot \pSub/2$ values of $j \in [m]$, we have $f_j \leq (1/2 - \pSub/2) \cdot \ell$.  If in fact $\pSub = 1/2$, then $f_j = 0$ for all $j \in [m]$. 
  \end{itemize}
  Let $\MJ \subset [m]$ be the set of $j \in [m]$ for which $e_j \leq \epED \cdot \ell \CDec$ and $f_j \leq (1/2 - \pSub/2) \cdot \ell$. Then $|\MJ| \geq m \cdot (\pSub/2 - 1/\CDec)$ (and $|\MJ| \geq m \cdot (1 -1/\CDec)$ in the case that $\pSub = 1/2$). For each $j \in \MJ$, there is some $j' \in |y'|$ so that
  \begin{align}
y'_{j:j+\ell} \in \SEDball(\bin(i_j) \circ \bin(z_j'), 1/2 - \pSub/2, 2\epED \CDec)\nonumber.
  \end{align}
  Note that the edit distance term above is $2\epED\CDec$ since an additional $e_j \leq \epED \cdot \ell \CDec$ deletions may be needed following the edits to transform $y''$ to $y'$ (as in the above expression we are considering the fixed-length quantity $y'_{j:j+\ell}$).

  Since $\epDec \geq 2\epED \CDec$ and $\pDec \leq \pSub/2$ (or else $\pDec = \pSub = 1/2$), it follows that, for $j \in \MJ$, in \cref{line:recover-inds} of $\Dec$, $z_j'$ is added to $\ML_{i_j}$.

  Let us write
  \begin{align}
\hat L := \begin{cases}
      |y| \cdot (nq)^{1 - \pDec^2 + \log(6e) \epDec \log(1/\epDec)} &: \pDec < 1/2\\
      |y| \cdot (nq)^{\epDec \log(6e) \log(1/\epDec)} &  :\pDec = 1/2\nonumber.
      \end{cases}
  \end{align}
  By \cref{lem:list-ub} and the fact that $|y'| \leq(1+ \epED) \cdot m\ell \leq 2m\ell$, the lists $\ML_i$ at the termination of the for loop on \cref{line:for-inds} of $\Dec$ satisfy $\sum_{i=1}^n |\ML_i| \leq 2m\ell \cdot \hat L$. Thus, for $\eta := \frac{2m\ell \cdot \hat L}{n \Lmax}$, the number of lists $\ML_w$ for which some item is removed in \cref{line:remove-list-items} of $\Dec$ may be bounded above by $\eta \cdot n$. Overall, it follows that the number of values of $j \in [m]$ for which $z_j' \in \ML_{i_j}$ at termination of $\Dec$ is at least the quantity claimed in the lemma statement. 
\end{proof}

Using \cref{lem:list-preserve}, we next establish the strong adaptive robustness of our PRC, for the parameter settings of \cref{def:hamming-edit-params} (in \cref{lem:hamedit-robustness}) and \cref{def:edit-params} (in \cref{lem:edit-robustness}). 
\begin{lemma}
  \label{lem:hamedit-robustness}
Given the parameters of \cref{def:hamming-edit-params}, the PRC of \cref{alg:prc-edit-binary} is strongly robust to $(1/2-\pSub,\epedit)$-substitution-edit distance bounded channels.
\end{lemma}
\begin{proof}
  We will show that for any value of $\sk$, 
  \begin{align}
\Pr_{y \gets \Enc(1^\lambda, \sk) } \left( \forall \ME \in \Ehamedit_{1/2-\pham,\epED},\ \Dec(1^\lambda,\sk, \ME(y)) = \True \right) \geq 1 - \exp(-n^{\Omega(1)}).
  \end{align}
  Fix any value for $\sk$. Let $\Egood$ be the event that the following conditions hold:
  \begin{itemize}
  \item Amongst the indices $i_1, \ldots, i_m$ generated on \cref{line:draw-w} of $\Enc$, there are at least $m - 2m^2/n$ distinct elements.
  \item The substitution channel $\channSC_\eta$ on \cref{line:draw-x-e} of $\Enc$ corrupts at most $2\eta m$ of the $m$ symbols of $c$ given by $c_{\sigma(i_1)}, \ldots, c_{\sigma(i_m)}$. 
  \end{itemize} 
  The probability that each $i_j$ is distinct from $i_1, \ldots, i_{j-1}$ is at least $1-(j-1)/n \geq 1-m/n$, and thus Azuma's inequality ensures that the probability of the first item above is at least $ 1 - \exp(-\Omega(m^3/n^2)) \geq 1 - \exp(-n^{\Omega(1)})$, since $m\geq n^{3/4}$. Moreover, by a Chernoff bound, the probability that the second item above fails is at most $\exp(-\Omega(\eta m)) \leq \exp(-n^{\Omega(1)})$, since $\eta = \Omega(1)$. Overall, it follows that $\Pr(\Egood) \geq 1 - \exp(-n^{\Omega(1)})$. 
  Let $\MI := \{ i_1, \ldots, i_m \} \subset [n]$, so that $|\MI| \geq m - 2m^2/n$ whenever the event $\Egood$ occurs.

  The output of $\Enc(1^\lambda, \sk)$, which we denote by $y \in \{0,1\}^{m\ell}$, is the concatenation of $\bin(i_1), \bin(z_1'),\ldots$ $ \bin(i_m), \bin(z_m')$, where $i_1, \ldots, i_m \in [n]$ and $z_1', \ldots, z_m' \in [q]$. Under the event $\Egood$, we have that, for at least $|\MI| - 2\eta m \geq m-2m^2/n-2\eta m$ values of $j \in [m]$, $z_j' = \pi_{i_j}(c_{\sigma(i_j)} + o_{\sigma(i_j)})$, where we recall that $c \in [q]^n$ is defined on \cref{line:draw-x-e} of $\Enc$. 

  Moreover, by \cref{lem:list-preserve}, for \emph{any} channel $\ME \in \Ehamedit_{1/2-\pham,\epedit}$, letting $y' = \ME(y)$ (so that $y' \in \SEDball(y, 1/2 - \pham, \epedit)$), there are at least  $m \cdot \left( \frac{\pSub}{2} - \frac{1}{\CDec} - \frac{2\ell \cdot (nq)^{1 - \pSub^2 + \epDec \log(6e) \log(1/\epDec)}}{\Lmax} \right)$ values of $j \in [m]$ for which $z_j' \in \ML_{i_j}$, where $\ML_{i_j}$ denotes the value of the list at the end of $\Dec(1^\lambda, \sk, y')$. Overall, the number of values of $i\in [n]$ for which $i_j = i$ for some $j \in [m]$ and $c_i + o_i \in \ML_{i}$  is at least
  \begin{align}
 m \cdot \left( \frac{\pSub}{2} - \frac{1}{\CDec} - \frac{2\ell \cdot (nq)^{1 - \pSub^2 + \epDec \log(6e) \log(1/\epDec)}}{\Lmax} -\frac{2m}{n} - 2\eta  \right)   \label{eq:recovered-inds-subham}.
  \end{align}
  Using the fact that $\CDec \leq \pSub/16$, that $2\eta \leq \pSub/16$, that $ 2m/n \leq \pSub/16$ for sufficiently large security parameter $\lambda$ (as $\pSub$ is a constant), that $q \leq (sn)^s$ (recalling the definition of $q = q(\lambda), n =n(\lambda)$), and that $\epDec \log(6e) \log(1/\epDec) \leq \pSub^2/2$ (by the choice of $\epDec$) we see that the quantity in \cref{eq:recovered-inds-subham} is at least
  \begin{align}
        m \cdot \left( \frac{\pSub}{4} - \frac{2\ell \cdot (sn)^{(s+1) \cdot (1-\pSub^2/2)}}{\Lmax} \right) \geq & m \cdot \left( \frac{\pSub}{4} - \frac{n^{s-4}}{\Lmax}\right) \geq \frac{m \cdot \pSub}{8}\nonumber,
  \end{align}
  where the first inequality uses that $2\ell \cdot (sn)^{(s+1) \cdot (1-\pSub^2/2)} \leq 2\ell s^{s+1} \cdot n^{s-5} \leq  n^{s-4}$ for sufficiently large $n$ (which in turn uses our choice of the constant $s$), and the second inequality uses that $\Lmax =n^{s-3} \geq n^{s-4} \cdot 8/\pSub$ for sufficiently large $n$ by our choice of $\Lmax$.
  
  Moreover, the lists $\ML_1, \ldots, \ML_n$ passed to the list recovery algorithm on \cref{line:find-close-codewords} have size bounded by $\Lmax =n^{s-3}$, which means, by \cref{thm:frs-listrecov}, as long as
  \begin{align}
\frac{m \cdot \pSub}{8} \geq k \cdot \left( n \cdot n^{s-3} \right)^{1/(s+1)} + 2,\nonumber
  \end{align}
  the list recovery algorithm will return at least one codeword (namely, the codeword chosen by $\Enc(1^\lambda, \sk)$ to produce $y$). In turn, the above inequality holds for sufficiently large $n$ since we have chosen $m = n^{s/(s+1)}$ and $k = n^{1/(s+1)}$. 
\end{proof}

\begin{lemma}
  \label{lem:edit-robustness}
Suppose that $\epedit \leq 1/400$. Given the parameters of \cref{def:edit-params}, the PRC of \cref{alg:prc-edit-binary} is strongly robust to $\epedit$-edit distance bounded channels.
\end{lemma}
\begin{proof}
  We will show that for any value of $\sk$, 
  \begin{align}
\Pr_{y \gets \Enc(1^\lambda, \sk)} \left( \forall \ME \in \Eedit_{\epED},\ \Dec(1^\lambda,\sk, \ME(y)) \neq \emptyset \right) \geq 1 - \exp(-n^{\Omega(1)}).
  \end{align}
  Fix any value for $\sk$. Let $\Egood$ be the event that the following events hold:
  \begin{itemize}
  \item Amongst the indices $i_1, \ldots, i_m$ generated on \cref{line:draw-w} of $\Enc$, there are at least $m - 2m^2/n$ distinct elements.
  \item The substitution channel $\channSC_\eta$ on \cref{line:draw-x-e} of $\Enc$ corrupts at most $2\eta m$ of the $m$ symbols of $c$ given by $c_{\sigma(i_1)}, \ldots, c_{\sigma(i_m)}$.
  \end{itemize} 
  As in the proof of \cref{lem:hamedit-robustness}, we have $\Pr(\Egood) \geq 1 - \exp(-n^{\Omega(1)})$, and we let $\MI := \{ i_1, \ldots, i_m \} \subset [n]$.

  The output of $\Enc(1^\lambda, \sk)$, which we denote by $y \in \{0,1\}^{m\ell}$, is the concatenation of $\bin(i_1), \bin(z_1'),\ldots$ $ \bin(i_m), \bin(z_m')$, where $i_1, \ldots, i_m \in [n]$ and $z_1', \ldots, z_m' \in [q]$. Under the event $\Egood$, we have that, for at least $|\MI| - 2\eta m \geq m-2m^2/n - 2\eta m$ values of $j \in [m]$, $z_j' = \pi_{i_j}(c_{\sigma(i_j)} + o_{\sigma(i_j)})$, where we recall that  $c \in [q]^n$ was generated on \cref{line:draw-x-e} of $\Enc$. 

  Moreover, by \cref{lem:list-preserve}, for \emph{any} channel $\ME \in \Eedit_{\epedit}$, letting $y' = \ME(y)$ (so that $y' \in \EDball(y,\epedit)$), there are at least  $m \cdot \left( 1 - \frac{1}{\CDec} - \frac{2\ell \cdot (nq)^{\epDec \log(6e) \log(1/\epDec)}}{\Lmax} \right)$ values of $j \in [m]$ for which $z_j' \in \ML_{i_j}$, where $\ML_{i_j}$ denotes the value of the corresponding list at the end of $\Dec(1^\lambda, \sk, y')$. Overall, the number of values of $i\in [n]$ for which $i_j = i$ for some $j \in [m]$ and $c_i + o_i \in \ML_{i}$ is at least
  \begin{align}
 m \cdot \left( 1 - \frac{1}{\CDec} - \frac{2\ell \cdot (nq)^{\epDec \log(6e) \log(1/\epDec)}}{\Lmax} -\frac{2m}{n} -2\eta \right)   \label{eq:recovered-inds}.
  \end{align}
  Using the fact that $\CDec \leq 1/16$, that $2\eta \leq 1/16$, that $ 2m/n \leq \pSub/16$ for sufficiently large security parameter $\lambda$ (as $\pSub$ is a constant), that $q \leq n$ (recalling the definition of $q = q(\lambda), n =n(\lambda)$), and that $\epDec \log(6e) \log(1/\epDec) \leq 1/10$ (by the choice of $\epDec = 2\epedit \CDec = 16\epedit$ and $\epedit \leq 1/400$) we see that the quantity in \cref{eq:recovered-inds} is at least
  \begin{align}
    m \cdot \left( \frac{1}{2} - \frac{2\ell \cdot n^{1/5}}{\Lmax} \right) \geq \frac{m}{4}\nonumber,
  \end{align}
where we have used that $\Lmax = n^{2/5} \geq 2\ell \cdot n^{1/5}$ for sufficiently lage $n$. 
  
  Moreover, the lists $\ML_1, \ldots, \ML_n$ passed to the list recovery algorithm on \cref{line:find-close-codewords} have size bounded by $\Lmax =n^{2/5}$, which means, by \cref{thm:rs-listrecov}, as long as
  \begin{align}
\frac{m}{4} \geq \sqrt{k n^{2/5} n} = n^{4/5}\nonumber,
  \end{align}
  the list recovery algorithm will return at least one codeword (namely, the codeword chosen by $\Enc(1^\lambda, \sk)$ to produce $y$). In turn, the above inequality holds by our definition of $m$. 
\end{proof}

\begin{lemma}
  \label{lem:soundness}
The PRC of \cref{alg:prc-edit-binary} (with the parameter settings of either \cref{def:edit-params} or \cref{def:hamming-edit-params}) is sound.
\end{lemma}
\begin{proof}
  Fix any string $y \in \{0,1\}^\star$. Note that the lists $\ML_1, \ldots, \ML_n$ as constructed in $\Dec(1^\lambda, \sk, y)$ do not depend on $o$. Recall that we use $C= C(\lambda)$ to denote the code used in \cref{alg:prc-edit-binary}. Thus, it suffices to show that, for any sets $\ML_1, \ldots, \ML_n$, each of size at most $\Lmax$, we have 
  \begin{align}
\Pr_{o \sim \Unif([q]^n)} \left( \exists y^\st \in C \mbox{ s.t. } y^\st_w + o_w \in \ML_w \mbox{ for at least } \trec \mbox{ values of } w \in [n] \right) \leq \exp(n^{-\Omega(1)})\label{eq:soundness-main}.
  \end{align}
  In turn, we may verify \cref{eq:soundness-main} by a simple counting argument. The number of strings $y \in [q]^n$ which satisfy $y_i \in \ML_i$ for at least $\trec$ values of $i \in [n]$ is bounded above by
  \begin{align}
q^{n - \trec} \cdot \Lmax^{\trec} \cdot {n \choose \trec} \leq q^{n - \trec} \cdot \left( \frac{en\Lmax}{\trec}\right)^{\trec} = q^n \cdot \left( \frac{e n \Lmax}{q\trec} \right)^{\trec}\nonumber.
  \end{align}
  Then \cref{eq:soundness-main} follows using  the fact that $o \sim \Unif([q]^n)$ together with a union bound over the $q^k$ choices of $y^\st \in C$, as we work out for each of the two parameter settings below: 
  \begin{itemize}
    \item For the parameter settings in \cref{def:edit-params}, we have $\trec = \sqrt{k \Lmax n} = n^{4/5}$, and so
      \begin{align}
        q^k \cdot \left( \frac{en \Lmax}{q\trec} \right)^{\trec} \leq n^{n^{1/5}} \cdot \left( \frac{2en \cdot n^{2/5}}{n \cdot n^{4/5}} \right)^{n^{4/5}} \leq \exp(-n^{\Omega(1)})\nonumber.
  \end{align}
\item For the parameter settings in \cref{def:hamming-edit-params}, we have $\trec = n^{(s-1)/(s+1)} + 2$ and $n^s \leq q \leq (sn)^s$, so
  \begin{align}
q^k \cdot \left( \frac{en \Lmax}{q\trec} \right)^{\trec} \leq (sn)^{s \cdot n^{1/(s+1)}} \cdot \left( \frac{en \cdot n^{s-3} }{n^s } \right)^{n^{(s-1)/(s+1)}} \leq \exp(-n^{\Omega(1)})\nonumber.
  \end{align}
\end{itemize}
\end{proof}

\begin{lemma}
  \label{lem:edit-undetectability}
The PRC of \cref{alg:prc-edit-binary} is undetectable under the following assumptions:
\begin{itemize}
\item For the parameter settings in \cref{def:edit-params}, under \cref{conj:permuted-rs};
\item For the parameter settings in \cref{def:hamming-edit-params}, under \cref{conj:permuted-frs}.
\end{itemize}
\end{lemma}
\begin{proof}
Consider any algorithm $\Adv$ running in time $T = T(\lambda)$ and which satisfies
  \begin{align}
\left| \Pr_{\sk \gets \KeyGen(1^\lambda)} \left( \Adv^{\Enc(1^\lambda, \sk)}(1^\lambda) = 1 \right) - \Pr_{\MU} \left( \Adv^{\MU}(1^\lambda) = 1 \right) \right| = \nu(\lambda)\nonumber,
  \end{align}
  for some functions $\nu, T : \BN \to \BR_{\geq 0}$. Here $\MU$ denotes the oracle that returns a uniformly random string of the same length as $\Enc(1^\lambda, \sk)$. We will use $\Adv$ to construct an adversary $\Adv'$  which distinguishes between the distributions $\MD_{n,[q],C,\eta,T}$ and $\Unif(([q]^n)^T)$; here the parameters $n,q,m$ and the code $C$ are given by either \cref{def:edit-params} or \cref{def:hamming-edit-params}. Note that in either case, the definitions of our parameters ensure that the conjecture in question (i.e., \cref{conj:permuted-rs} or \cref{conj:permuted-frs}) is exactly that the distributions $\MD_{n,[q],C,\eta,T}$ and $\Unif(([q]^n)^T)$ are computationally indistinguishable. (We omit the dependence of the various parameters on $\lambda$.)
  
  Note that $\Adv$ requires oracle access to an oracle $\MO$ (to be interpreted as either $\Enc(1^\lambda, \sk)$ or the uniform distribution) which when called outputs a sample in $\{0,1\}^{m \cdot (\log q + \log n)}$. The algorithm $\Adv'$ simulates $\Adv$ as follows: first, $\Adv'$ draws $o \sim \Unif([q]^n)$. Then, for $t \in [T]$, for the $t$th call that $\Adv$ makes to the oracle $\MO$, $\Adv'$ uses the $t$th sample from its distribution (either $\MD_{n,[q],C,\eta,T}$ or $\Unif(([q]^n)^T)$) which is a codeword $c \in [q]^n$, to simulate the oracle call $\MO$, as follows: 
  \begin{itemize}
    \item It draws $i_1, \ldots, i_m \sim \Unif([n])$ independently. 
    \item For each $j \in [m]$, if there is $j' < j$ with $i_j = i_{j'}$, then $\Adv'$ sets $z_j' \gets \Unif([q])$. Otherwise, it sets $z_j' \gets c_{i_j} + o_{i_j}$. 
  \item $\Adv'$ then uses $\bin(i_1) \circ \bin(z_1') \circ \cdots \circ \bin(i_m) \circ \bin(z_m')$ as the output of the oracle $\MO$. 
  \end{itemize}
  It is clear that the running time of $\Adv'$ is at most $T(\lambda) \cdot \poly(\lambda)$. We now make the following observations:
  \begin{itemize}
  \item If the distribution that $\Adv'$ is given is $\MD_{n,[q],C,\eta,T}$, then $\Adv'$ simulates $\Adv^{\Enc(1^\lambda, \sk)}(1^\lambda)$, where $\sk = (\sigma, \pi_1, \ldots, \pi_n, \sigma \circ o)$,\footnote{Here $\sigma \circ o$ denotes $(o_{\sigma(1)}, \ldots, o_{\sigma(n)})$.} with $\sigma, \pi_1, \ldots, \pi_n$ being given by the (secret) sampled permutations in $\MD_{n,[q],C,\eta, T}$. 
  
  Indeed, in the simulation of the oracle $\MO = \Enc(1^\lambda, \sk)$ as above, the codeword $c$ (representing one of the $T$ samples of $\MD_{n,[q],C,\eta,T}$) is given by $c \gets \channSC_\eta(\hat c)$ where $\hat c_i = \pi_i(\tilde c_{\sigma(i)})$, where $\tilde c \gets C$ is a uniformly random codeword in $C$. Then the result of the sequence of operations $\tilde c \mapsto \hat c \mapsto c \mapsto (i_1, z_1', \ldots, i_m', z_m')$ in the simulation of $\Adv'$ has the same distribution as the sequence of operations $c \mapsto c' \mapsto z \mapsto (i_1, z_1', \ldots, i_m, z_m')$ in the execution of \cref{alg:prc-edit-binary} (where $c \gets \Unif(C)$, $c' \gets \channSC_\eta(c)$, $z \gets c' + o$, and $z_j' \gets \pi_{i_j}(z_{\sigma(i_j)})$). 
  \item If the distribution that $\Adv'$ is given is $\Unif(([q]^n)^T)$, then $\Adv'$ simulates $\Adv^{\MU}(1^\lambda)$.
  
  Indeed, in the simulation of the oracle $\MO = \MU$ as above, the codeword $c$ is uniformly random, which means that so is the sequence $(i_1, z_1', \ldots, i_m, z_m')$. 
  \end{itemize}  Thus, under the assumption that there is no subexponential-time algorithm distinguishing $\MD_{n,[q],C,\eta,T}$ and $\Unif(([q]^n)^T)$, we must have that either $T(\lambda) \geq \exp(\lambda^{\Omega(1)})$ or $\nu(\lambda) \leq \exp(-\lambda^{\Omega(1)})$.
\end{proof}

\subsection{From PRCs to watermarking: background}
\label{sec:prc-watermarking-formal}
In this section, we give some background on watermarking schemes for language models; we refer the reader to \cite{CG24} for further details and explanations. 
\begin{definition}
An \emph{autoregressive language model} $\Model$ over alphabet (i.e., token set) $\Sigma$ is a deterministic algorithm that takes as input a prompt $\prompt \in \Sigma^\st$ and a sequence of previous tokens $\tok_1, \ldots, \tok_{i-1}$ and outputs a distribution $p_i := \Model(\prompt, \tok_{1:i-1}) \in \Delta(\Sigma)$. 
\end{definition}
In the event that $\Sigma = \{0,1\}$, we abuse notation slightly and write $p_i = \Model(\prompt, \tok_{1:i-1}) \in [0,1]$ to denote the probability that the model places on the next token being $1$. 

We assume that $\prompt$ encodes the length of desired text output by $\Model$. Given a model $\Model$, a prompt $\prompt$ which specified some length $\ell$, we let $\Modelo(\prompt)$ denote the random variable $\tok \in \Sigma^\ell$ generated in the natural way, i.e., given $\tok_{1:i-1}$, we draw $\tok_i \sim \Model(\prompt, \tok_{1:i-1})$, and repeat for $\ell$ steps.

In order to obtain watermarking procedures for language models, we need to assume that the model has sufficient \emph{entropy}, formalized below: 
\begin{definition}
  \label{def:empircal-entropy}
  Given a language model $\Model$, a sequence $\tok\in \Sigma^\st$, and $i,j \in [|\tok|]$, we define the \emph{empirical entropy} of $\Model$ on $\tok$ on the subsequence $[i,j]$ by
  \begin{align}
\Hemp{[i:j]}(\tok, \Model) := \sum_{a=i}^j - \log \Model(\tok_a \mid \tok_{1:a-1}) \nonumber.
  \end{align}
\end{definition}
When the choice of $\Model$ is clear from context, we will often write $\Hemp{[i:j]}(\tok) := \Hemp{[i:j]}(\tok, \Model)$. Moreover, we write $\Hemp{[i:j)}(\tok) = \Hemp{[i:j-1]}(\tok, \Model)$.

A \emph{watermarking scheme} $\MW$ for a model $\Model$ is a tuple of polynomial-time algorithms $\MW = (\Setup, \Wat, \Detect)$, where:
\begin{itemize}
\item $\Setup(1^\lambda)$ outputs a secret key $\sk$ of length polynomial in $\lambda$, where $\lambda \in \BN$ denotes a security parameter.
\item $\Wat(\sk, \prompt)$ takes as input a prompt and $\sk$ and outputs a response $\tok \in \Sigma^\ell$ of length $\ell$. 
\item  $\Detect(\sk, \tok)$ takes as inptut $\sk$ and a sequence $\tok \in \Sigma^\st$ and outputs a response in $\{ \True, \False \}$ indicating whether $\Detect$ detects the sequence $\tok$ as being watermarked according to $\sk$. 
\end{itemize}

The main desired properties of watermarking schemes parallel those of PRCs, and are formalized below: 
\begin{definition}[Undetectability (\cite{CGZ24})]
  \label{def:wat-undetect}
  A watermarking scheme $\MW$ is \emph{undetectable} if for every security parameter $\lambda$, prompt $\prompt$, and any polynomial-time distinguisher $\Dist$, it holds that
  \begin{align}
\left| \Pr\left( \Dist^{\Modelo}(1^\lambda) = \True \right) - \Pr \left( \Dist^{\Wat(\sk, \cdot)}(1^\lambda) = \True \right) \right| \leq \negl(\lambda)\nonumber,
  \end{align}
  where $\Dist^{\MO}$ means that the distinguisher can query $\MO$ with an arbitrary prompt $\prompt$. 
\end{definition}

\begin{definition}[Soundness]
  \label{def:wat-sound}
  A watermarking scheme $\MW$ is \emph{sound} if for every security parameter $\lambda$ and every fixed token sequence $\tok \in \Sigma^\st$ of length $\poly(\lambda)$,
  \begin{align}
\Pr_{\sk \gets \Setup(1^\lambda)} \left( \Detect(\sk, \tok) = \True \right) \leq \negl(\lambda)\nonumber.
  \end{align}
\end{definition}

\begin{definition}[Substring Robustness]
  \label{def:wat-robust}
  Fix a channel $\ME : \{0,1\}^\st \times \Sigma^\st \to \Sigma^\st$ that takes as input a codeword $x \in \Sigma^\st$ and some auxiliary information $\sk \in \{0,1\}^\st$. Consider a family of functions $\beta_\lambda : \BN \to \BN$ indexed by the security parameter $\lambda$. Then a watermarking scheme $\MW$ is defined to be \emph{$\beta$-robust to $\ME$} if for each $\lambda \in \BN$ and each prompt $\prompt$, 
  \begin{align}
\Pr_{\substack{\sk \gets \Setup(1^\lambda) \\ \tok \gets \Wat(\sk, \prompt),\ \tok' \gets \ME(\sk,\tok)}} \left( \exists i,j \in [|\tok|],\ \mbox{ s.t. } \Detect(\sk, \tok') = \False \mbox{ and } \Hemp{[i:j)}(\tok) \geq \beta_\lambda(j-i) \right) \leq \negl(\lambda)\nonumber.
  \end{align}
\end{definition}

\begin{algorithm}
  \caption{Binary-alphabet watermarking scheme from binary-alphabet PRC \cite{CG24}}
  \label{alg:watermarking-edit}
  \begin{algorithmic}[1]\onehalfspacing
    \Require Pseudorandom code $\PRC = (\PRC.\KeyGen, \PRC.\Enc, \PRC.\Dec)$. Maximum length $\Imax$ for the watermarked text as encoded by any possible prompt.  
    \Function{$\Setup$}{$1^\lambda$}
    \State $\PRC.\sk \gets \PRC.\KeyGen(1^\lambda)$.
    \State $\sigma_1, \ldots, \sigma_{\Imax} \gets \Unif(\{0,1\}^{n(\lambda)})$.   \State \Return $\sk := (\PRC.\sk, \sigma_{1:\Imax})$. 
    \EndFunction

    \Function{$\Wat$}{$\sk,\prompt$}
    \State Let $\ell \in \BN$ denote the desired length encoded by $\prompt$ and $n = n(\lambda)$ denote the block length of the PRC corresponding to $\PRC.\sk$. 
    \State Write $\sk = (\PRC.\sk, \sigma_{1:\Imax})$.
    \For{$1 \leq i \leq \ell$}
    \If{$i \equiv 1 \pmod{n}$}
    \State Set $x \gets \sigma_{\lceil i/n\rceil} \oplus \PRC.\Enc(\sk)  \in \{0,1\}^n$.\label{line:wat-prc}
    \EndIf
    \State $p_i \gets \Model(\prompt, \tok_{1:i-1})$.
    \State Draw $\tok_i \gets \Ber(p_i - (-1)^{x_i}\cdot \min\{ p_i, 1-p_i \})$. \label{line:draw-toki}
    \EndFor
    \EndFunction

    \Function{$\Detect$}{$\sk, \tok$}
    \For{$i \in [|\tok|], j \in [i, \min\{i+n,|\tok|]$}
    \If{$\PRC.\Dec(\sk, \tok_{i:j-1}) = \True$}
    \State \Return \True
    \EndIf
    \EndFor
    \State \Return \False
    \EndFunction
  \end{algorithmic}

\end{algorithm}

\subsection{From PRC to watermarking: soundness, undetectability, \& robustness}
\label{sec:prc-to-watermarking}
In this section, we establish the following theorem, which shows that we can obtain watermarking schemes for language models which have strong adaptive robustness to edit-bounded channels whenever the per-token entropy of the language model is at least a constant. For simplicity, we focus on binary-alphabet language models, but a straightforward reduction (see \cite{CG24}) shows that this is without loss of generality. 
\begin{theorem}[Main edit-robust watermarking result]
\label{thm:watermarking-main}
Fix any constant $\alpha > 0$ and write, for security parameter $\lambda$, $\beta_\lambda(\ell) = 8\alpha \cdot \ell + 2\sqrt{2} \cdot \lambda$. Then there is a watermarking scheme for any language model over a binary alphabet which is $\beta$-robust to any $\tilde O(\alpha^7)$-edit-bounded channel and which satisfies soundness and undetectability, where the latter holds under \cref{conj:permuted-frs}. The key generation and watermarking procedures run in time $\poly(\lambda)$, and the detection procedure runs in time $\lambda^{O(1/\alpha^4)}$. 
\end{theorem}
Intuitively, the guarantee of $\beta$-robustness to $\tilde O(\alpha^7)$-edit-bounded channels in the context of \cref{thm:watermarking-main} means the following: as long as the empirical entropy of the language model for a sequence of length $\ell \gg \lambda$ is at least an $\Omega(\alpha)$-fraction of the length, then we can still detect the watermark even after the sequence has been passed through an arbitrary channel which can make a $\tilde O(\alpha^7)$-fraction of edits.  Improving the exponent in the bound on the fraction of edits and improving the running time of the detection procedure are both interesting open questions.

The watermarking procedure of \cref{thm:watermarking-main} is presented in \cref{alg:watermarking-edit}. It depends on a PRC $\PRC$; given such a PRC, we let $\MW[\PRC]$ denote the watermarking scheme of \cref{alg:watermarking-edit}. It also takes as input a parameter $\Imax$ denoting the maximum possible length of a sequence output by the model, as encoded by any possible prompt; we assume that $\Imax$ is polynomial in the security parameter. 
The construction is essentially identical to that of \cite{CG24}; moreover, \cite{CG24} showed that the properties of soundness and undetectability of $\MW[\PRC]$ follow from those of $\PRC$, and these guarantees carry over in a straightforward manner to our setting:
\begin{lemma}[\cite{CG24}, Lemma 18] \label{lem:wat-edit-soundness}
Suppose that $\PRC$ is sound. Then $\MW[\PRC]$ is sound.
\end{lemma}

\begin{lemma}[\cite{CG24}, Lemma 19] \label{lem:wat-edit-undetectability}
Suppose that $\PRC$ is undetectable. Then $\MW[\PRC]$ is undetectable.
\end{lemma}
The next lemma shows that if $\PRC$ has strong adaptive robustness to the set of substitution-edit bounded channels (for a sufficiently large substitution error rate), then $\MW[\PRC]$ has strong adaptive robustness to the set of edit-bounded channels. 
\begin{lemma}
  \label{lem:wat-edit-robustness}
Let $\epham, \epedit > 0$ be constants, and suppose that $\mathscr{E}$ is the set of channels $\ME : \Sigma^\st \to \Sigma^\st$ which are $(1/2 - \epham, \epedit)$-substitution-edit bounded. If $\PRC$ is robust to every channel in $\mathscr{E}$, then for $\beta_\lambda(\ell) := 8\sqrt{\epham} \cdot \ell + 2\sqrt{2} \cdot n(\lambda)$, we have that $\MW[\PRC]$ is $\beta$-robust to any $\epedit \sqrt{\epham}/3$-edit-bounded channel. 
\end{lemma}
\begin{proof}
  Fix a prompt $\prompt$, and let $\tok \gets \Wat(\sk, \prompt)$ denote the output of the watermarking procedure given $\prompt$. Write $\epedit' := \epedit \sqrt{\epham}/3$, and fix any $\epedit'$-edit bounded channel $\ME$, and let $\tok' = \ME(\tok)$. 
  Let us write $\ell := |\tok|$, and fix any $i,j \in [\ell]$ with $j \geq i$. Let $i'$ denote the smallest value which is  $1 \pmod{n}$ and at least $i$ and let $j'$ denote the largest value which is $1 \pmod{n}$ and at most $j$. Let $i_1 = i', i_2, \ldots, i_{g+1} = j'$ denote the values which are $1 \pmod{n}$ and between $i'$ and $j'$. Moreover let $x_1, \ldots, x_g \in \{0,1\}^n$ denote the PRC codewords which are generated on \cref{line:wat-prc} of \cref{alg:watermarking-edit} corresponding to positions $i_1, \ldots, i_{g+1}$ (after the addition with the one-time pads $\sigma_i$). For each $f \in [g]$, we consider the following \emph{embedding channel} $\Eemb_f : \{0,1\}^n \to \{0,1\}^\st$. $\Eemb_a$ maps the string $x_f \in \{0,1\}^n$ to the substring $\tok_{i_f:i_{f+1}-1}$ of $\Model$'s output. Note that $\Eemb_f$ is randomized, as it depends on the other PRC codewords as well as the randomness of the draws of $\tok_i$ on \cref{line:draw-toki} of \cref{alg:watermarking-edit}. 

  By \cite[Lemma 20]{CG24}, for any $k,k' \in [\ell]$ with $k' - k = n$, with probability $1-\negl(n)$, we have $\Hemp{[k,k')}(\tok) \leq \sqrt{2} \cdot n$. It follows in particular that with probability $1-\negl(n)$,  $\Hemp{[i,i')}(\tok) \leq \sqrt{2} \cdot n$, and $\Hemp{[j', j)}(\tok) \leq \sqrt{2}\cdot n$. Thus, with probability $1-\negl(n)$, $\Hemp{[i', j')}(\tok) \geq 8\sqrt{\epham} \cdot \ell \geq 8 \sqrt{\epham} \cdot ng$. Hence there are at least $\sqrt{\epham} \cdot g$ values of $f \in [g]$ for which $\Hemp{[i_f, i_{f+1})}(\tok) > 4 \sqrt{\epham} \cdot n$. Let the set of such $f$ be denoted by $\MS \subset [g]$. 

  For each $f \in [g]$, let $e_f$ denote the number of edits (i.e., insertions and deletions) made at positions $[i_f, i_{f+1})$ to obtain $\tok'$ from $\tok$ (via $\ME$). Note that at least one of these values of $f \in \MS$ must have the property that $e_f \leq n \cdot 3 \epedit'/\sqrt{\epham} =n\cdot \epedit$ (as otherwise the total number of edits made to obtain $\tok'$ from $\tok$ would be greater than $3ng \cdot \epedit' \geq \ell \cdot \epedit'$). Let $f^\st \in [g]$ denote the random variable which is the smallest such value of $f$. Due to the addition of the one-time pads $\sigma_f$ to each $x_f$, each of the codewords $x_f$ is drawn uniformly from $\{0,1\}^n$, independently from other $x_{f'}$ and also from the channel $\Eemb_f$. We now need the following result, which is a direct consequence of Lemma 21 of \cite{CG24} together with the fact we have just stated that $x_f$ is uniform and independent of $\Eemb_f$ (in particular, we use here that the prompt $\prompt$ and therefore the model's distribution do not depend on $\sk$): 
  \begin{lemma}[Lemma 21 of \cite{CG24}]
    \label{lem:ham-hemp}
    For any constant $c > 0$ and value of $f \in [g]$, it holds that
    \begin{align}
\Pr_{\substack{\sk \gets \KeyGen(1^\lambda) \\ x \gets \PRC.\Enc(\sk) \\ x' \gets \Eemb_f(x)}}\left( \Ham(x, x') > \left( \frac 12 - \frac{c^2}{16} \right) \cdot n \mbox{ and } \Hemp{[i_f:i_{f+1})}(x') > c \cdot n  \right) \leq \negl(n)\nonumber.
    \end{align}
  \end{lemma}
  It follows from \cref{lem:ham-hemp} (with $c = 4\sqrt{\epham}$) that with probability $1-\negl(n)$, we have that $\Ham(x_{f^\st}, \tok_{i_{f^\st}:i_{f^\st+1}-1}) \leq \left( \frac{1}{2} - \epham \right) \cdot n$. Moreover, by choice of $f^\st$, we also have $\ED(\tok_{i_{f^\st}:i_{f^\st+1}-1},\tok'_{a:b}) \leq \epedit \cdot n$. Thus, since $\PRC$ is strongly robust to every channel which is $(1/2 - \epham, \epedit)$-substitution-edit bounded, with probability $1-\negl(n)$, we have that $\PRC.\Dec(\sk, \tok'_{a:b-1}) = \True$ for some $a, b \in [\ell]$. 
\end{proof}

Finally we are ready to prove \cref{thm:watermarking-main}.
\begin{proof}[Proof of \cref{thm:watermarking-main}]
We use the PRC of \cref{alg:prc-edit-binary} with the parameter settings of \cref{def:hamming-edit-params}, so that \cref{thm:prc-hamming-edit} applies. The soundness and undetectability guarantees of \cref{thm:prc-hamming-edit}, together with \cref{lem:wat-edit-soundness} and \cref{lem:wat-edit-undetectability}, imply that $\MW[\PRC]$ satisfies soundness and undetectability (under \cref{conj:permuted-frs}). Finally, taking $\beta_\lambda(\ell) = 8\alpha \cdot \ell + 2\sqrt{2} \cdot n(\lambda)$, by applying \cref{lem:wat-edit-robustness} with $\epham = \alpha^2$ and \cref{thm:prc-hamming-edit} with $\pSub = \epham = \alpha^2$ (which requires us to take $\epedit \leq \tilde O(\pSub^3) = \tilde O(\alpha^6)$), we obtain that $\MW[\PRC]$ is $\beta$-robust to any $\tilde O(\alpha^7)$-edit-bounded channel. The running time guarantees follow directly from those of \cref{alg:prc-edit-binary}.
\end{proof}

\section{Substitution-robust pseudorandom codes} \label{sec:sub-robustness}
In this section, we give a straightforward construction which converts any family of codes satisfying the permuted codes assumption \cref{def:permuted-codes} into a PRC that has strong adaptive robustness to any substitution channel introducing at most a constant fraction of errors. (Under \cref{conj:permuted-codes}, this construction works for any code family with large dual distance.) Of course, such a result follows from the construction based on Reed-Solomon codes in \cref{sec:edit-robustness}, under the permuted codes assumption for Reed-Solomon codes (\cref{conj:permuted-rs}). 
The construction in this section, however, is secure as long as the conjecture holds for \emph{some} family of codes with high dual distance and efficient decoding from a constant error rate. 

\begin{proposition}
  \label{prop:prc-subst}
  Let $\lambda \in \BZ^+$ be a security parameter, let $n = n(\lambda), k = k(\lambda), q = q(\lambda)$ be polynomially-bounded functions in $\lambda$ denoting the block length, dimension, and alphabet size, respectively. Let $ C = \{ C(\lambda) \}_{\lambda \in \BN}$ be a family of linear codes with $C(\lambda) \subset \BF_{q(\lambda)}^{n(\lambda)}$. Suppose that $C$ comes equipped with efficient algorithms $\Enc_C : \BF_q^k \to \BF_q^n$ and $\Dec_C : \BF_q^n \to \BF_q^k \cup \{ \perp \}$ which correct up to a fraction $\delta \leq 1-1/q - \Omega(1)$ of worst-case substitutions. Under the permuted codes assumption for $C$ with error $\eta = \delta/3$ (\cref{def:permuted-codes}), the PRC of \cref{alg:prc-subst} for the family $C$ is undetectable, sound, and has strong adaptive robustness to any channel $\ME$ introducing a fraction of at most $\delta/2$ of substitutions.
\end{proposition}
We remark that the robustness of the PRC in \cref{prop:prc-subst} can be improved to essentially the same rate $\delta$ as is achieved by the underlying code $C$, decreasing the error rate $\eta$ to a vanishing fraction of $\delta$. Moreover, it is possible to further improve the robustness by using a code $C$ which has an efficient \emph{list-decoding} algorithm. 

Below we outline some concrete families of code $C$ which have polynomial dual distance, and thus, under \cref{conj:permuted-codes}, satisfy the hypotheses of \cref{prop:prc-subst}. We emphasize that the below families all have alphabet size $q(\lambda)$ which is a \emph{constant} (i.e., does not depend on $\lambda$).

\paragraph{Example: AG codes.} We follow the terminology and notation of \cite{hoholdt1998algebraic}. Let $q$ be the square of a prime; then \cite{garcia1995tower} shows the existence of, for each $m \in \BN$, an irreducible nonsingular projective curve $\MX$ over $\BF_q$ with genus less than $g := q^{(m-1)/2} \cdot (\sqrt{q} + 1)$ and with more than $n := q^{(m-1)/2} \cdot (q-1)$ rational points. Letting $n+1$ of these rational points be denoted $Q, P_1, \ldots, P_n \in \MX$, we define the divisors $D := P_1 + \cdots + P_n$ and $G := a \cdot Q$ for some integer $a > 2g-2$. Let $\ML(G) := \{ f \in \BF(\MX)^\st \mid (f) + G \geq 0\} \cup \{ 0 \}$ denote the vector space of functions on $\MX$ with poles only at $Q$ of order at most $a$. We consider the AG code $C = C(D,G)$, defined as the image of the linear mapping $E : \ML(G) \to \BF_q^n$  defined by $E(f) = (f(P_1), \ldots, f(P_n))$. \cite[Theorems 2.65 \& 2.69]{hoholdt1998algebraic} give the following bounds on the dimension $k$, distance $d$, and dual distance $d^\star$ of $C(D,G)$:
\begin{align}
k = a-g + 1 , \qquad d \geq n - a, \qquad d^\star \geq a-2g+2\nonumber.
\end{align}
Moreover, there are efficient algorithms to decode $C(D,G)$ up to $\lfloor(d-1)/2\rfloor$ substitutions \cite[Theorem 6.12]{hoholdt1998algebraic}.\footnote{See also, e.g., \cite[Corollary 5.2]{shokrollahi1998decoding} for a particularly simply algorithm which decodes up to $\lfloor(d-1)/2 \rfloor - g$ substitutions. Additionally, we remark that all of these algorithms assume the existence of a succinctly-represented basis for $\ML(G)$; we assume that such a basis exists. See e.g., \cite{hess2002computing}.} Taking $q$ to be a constant (e.g., any prime square at least $25$ will suffice), for each $m \in \BN$, we take the degree of the divisor $Q$ to be $a_m := 2g_m-2$, and we obtain a code $C = C_m$ with block length $n_m = q^{(m-1)/2} \cdot (q-1) = g_m \cdot (\sqrt{q}-1)$, whose distance $d_m$ and dual distance $d_m^\st$ satisfy
\begin{align}
d_m^\st = a_m - 2g_m + 2 = g_m = \frac{n_m}{\sqrt{q}-1} = \Omega(n_m), \qquad d_m \geq n_m - a_m = g_m \cdot \left(\sqrt{q} - 1  - 3 \right) = \Omega(n_m)\nonumber,
\end{align}
i.e., both are a constant fraction of the block length. Thus, under \cref{conj:permuted-codes}, the resulting family of codes satisfies the assumptions of \cref{prop:prc-subst}. Notice also that we can get a similar result (namely, linear distance and dual distance) over a binary alphabet by concatenating with the trivial code, though the distance and dual distance will of course degrade by a constant factor. 

\paragraph{Example: Raw Reed-Solomon Codes.} As another example, we consider the \emph{Raw Reed-Solomon} codes introduced in \cite{silbak2019quasilinear}. For a positive integer $m$, we let $r = 2^m$, we write $\tilde n := r -1$, and fix some $\tilde k < \tilde n$ denoting the dimension (which will ultimately be taken to be $k = \tilde n^\alpha$ for some $\alpha \in (0,1/2)$). We first consider the Reed-Solomon code $\tilde C := \RS_{\BF_r,\tilde n,\tilde k}$ over $\BF_r$ with evaluation points $\BF_r^\st$ (so that the block length is $\tilde n=r-1$). Interpreting $\BF_r$ as a dimension-$m$ vector space over $\BF_2$ via a bijection $\Phi : \BF_r \simeq \BF_2^m$ allows us to interpret $\tilde C$ as a code $C \subset \BF_2^{\tilde nm}$ with dimension $k := \tilde km$. While the distance of $\tilde C$ is only guaranteed to be at least $\tilde n - \tilde k$, which is \emph{sublinear} in the block length $\tilde nm$ as $m = \Theta(\log \tilde n)$, a remarkable fact shown in \cite{silbak2019quasilinear} states that if we consider the subcode of $\tilde C$ induced by only encoding polynomials with constant term equal to $0$, then the resulting binary distance has linear distance; in fact, its relative distance approaches $1/2$. In particular, as shown in \cite[Theorem 3.1]{silbak2019quasilinear}, this construction gives, for each $m$, and $\alpha \in (0,1/2)$, a code $C$ with block length $n_m = (2^m-1) \cdot m$, dimension $k_m = n_m^\alpha$, and whose distance $d_m$ and dual distance $d_m^\st$ satisfy
\begin{align}
d_m^\st = \Omega \left( \frac{n_m^\alpha}{\log n_m} \right), \qquad d_m = n_m \cdot \left( \frac 12 - O \left( \left(\frac{\log n}{n}\right)^{1/2 - \alpha} \right)\right) = \left( \frac 12 - o(1) \right) \cdot n_m\nonumber.
\end{align}
Moreover, this code is efficiently decodable up to half its distance; thus, under \cref{conj:permuted-codes}, the resulting family of codes satisfies the assumptions of \cref{prop:prc-subst}.

\subsection{Proof of \cref{prop:prc-subst}}
The proof of \cref{prop:prc-subst} is straightforward and follows similar (but simpler) lines to those in \cref{sec:edit-robustness}: in particular, we prove undetectabilty, soundness, and robustness in the below lemmas:

\begin{lemma}
\label{lem:prc-subst-soundness}
In the setting of \cref{prop:prc-subst}, the PRC of \cref{alg:prc-subst} is sound.
\end{lemma}
\begin{proof}
Fix any string $y \in \{0,1\}^\st$. It suffices to show that 
\begin{align}
\Pr_{o \sim \Unif(\BF_q^n)} \left( \exists y^\st \in C \mbox{ s.t. } \Ham(y^\st, y - o) \leq \delta \cdot n \right) \leq \exp(-\Omega(n))\label{eq:soundness-subst}.
\end{align}
 The number of strings $y' \in \BF_q^n$ which satisfy $\Ham(y', y - o) \leq \delta \cdot n$ is at most $q^{H_q(\delta) \cdot n}$, where $H_q(\cdot)$ denotes the $q$-ary entropy function. Since we have assumed $\delta \leq 1-1/q - \Omega(1)$, then this number is at most $q^{(1 - \Omega(1))\cdot n}$, meaning that the probability of the event in \cref{eq:soundness-subst} is at most $q^{-\Omega(n)}$, as desired. 
\end{proof}

\begin{lemma}
\label{lem:prc-subst-undetectability}
In the setting of \cref{prop:prc-subst}, the PRC of \cref{alg:prc-subst} is undetectable.
\end{lemma}
\begin{proof}
The proof closely follows that of \cref{lem:edit-undetectability}. Consider any algorithm $\Adv$ running in time $T = T(\lambda)$ and which satisfies
\begin{align}
\left| \Pr_{\sk \gets \KeyGen(1^\lambda)} \left( \Adv^{\Enc(\sk)}(1^\lambda) = 1 \right) - \Pr_{\MU} \left( \Adv^{\MU}(1^\lambda) = 1 \right) \right| = \nu(\lambda)\nonumber,
\end{align}
for some functions $\nu, T : \BN \to \BR_{\geq 0}$. Here $\MU$ denotes the oracle that returns a uniformly random string of the same length as $\Enc(\sk)$. We will use $\Adv$ to construct an adversary $\Adv'$  which distinguishes between the distributions $\MD_{n,\BF_q,C,\delta/3,T}$ and $\Unif((\BF_q^n)^T)$. (We omit the dependence of the various parameters on $\lambda$.)

Note that $\Adv$ requires oracle access to an oracle $\MO$ (to be interpreted as either $\Enc(\sk)$ or the uniform distribution) which when called outputs a sample in $\BF_q^n$. The algorithm $\Adv'$ simulates $\Adv$ as follows: first, $\Adv'$ draws $o \sim \Unif(\BF_q^n)$. Then, for $t \in [T]$, for the $t$th call that $\Adv$ makes to the oracle $\MO$, $\Adv'$ uses the $t$th sample from its distribution (either $\MD_{n,\BF_q,C,\eta,T}$ or $\Unif((\BF_q^n)^T)$) which is an element $c \in \BF_q^n$, to simulate the oracle call $\MO$ by simply returning $c + o$. 
It is clear that the running time of $\Adv'$ is at most $T(\lambda) \cdot \poly(\lambda)$. 

Note that, in the event that $\Adv'$ is given a sample from $\MD_{n,\BF_q, C, \eta, T}$, then it exactly simulates $\Adv^{\Enc(\sk)}(1^\lambda)$, for $\sk = (\sigma, \pi_1, \ldots, \pi_n, o)$, where $\sigma, \pi_1, \ldots, \pi_n$ are the (secret) permutations used in the definition of $\MD_{n,\BF_q,C,\eta,T}$. This holds because the definition of $\Enc(\sk)$ in \cref{alg:prc-subst} exactly parallels the definition of $\MD_{n, \BF_q, C, \eta, T}$ in \cref{sec:comp-asm} (with the addition of the one-time pad $o$).

On the other hand, in the event that $\Adv'$ is given a sample from $\Unif((\BF_q^n)^T)$, then it exactly simulates $\Adv^{\MU}(1^\lambda)$ because each of the simulated oracle calls to $\MU$ returns a uniformly random element of $\BF_q^n$. Thus, under the assumption that the distributions $\MD_{n,\BF_q,C,\eta,T}$ and $\Unif((\BF_q^n)^T)$ are computationally indistinguishable to sub-exponential time algorithms, we must have that either $T(\lambda) \geq \exp(\lambda^{\Omega(1)})$ or $\nu(\lambda) \leq \exp(-\lambda^{\Omega(1)})$.
\end{proof}

\begin{lemma}
\label{lem:prc-subst-robustness}
In the setting of \cref{prop:prc-subst}, the PRC of \cref{alg:prc-subst} has strong adaptive robustness to any channel $\ME$ introducing a fraction of at most $\delta/3$ of substitutions.
\end{lemma}
\begin{proof}
Fix any $\sk$, let us consider $x, c, \hat c, c'$ as computed in the execution of $\Enc(\sk)$. 
  Since $\eta = \delta/3$ is a constant, with probability $1- \exp(-\Omega(n))$ over the randomness of the channel $\channSC_\eta$ in $\Enc(\sk)$, we have that $\Ham(c', \hat c) \leq \frac{\delta}{2} \cdot n$. Letting $y = c' + o$ denote the output of $\Enc(\sk)$ and $y' = \ME(y)$ for any channel $\ME$ introducing at most a fraction of $\delta/3$ of substitutions (which can depend on $\sk$), it follows that $\Ham(y' - o, \hat c) \leq \delta \cdot n$ with probability $1-\exp(-\Omega(n))$. This in particular implies that $\Ham(y'', c) \leq \delta \cdot n$, where $y''$ is as computed in $\Dec(\sk, y)$ and $c$ is as computed in $\Enc(\sk)$. Thus, since $\Dec_C$ corrects up to a fraction $\delta$ of substitutions, we have that $\Dec(\sk, y') = \True$ with probability $1 - \exp(-\Omega(n))$, as desired.
\end{proof}

\begin{algorithm}
  \caption{Substitution-robust PRC from permuted codes}
  \label{alg:prc-subst}
  \begin{algorithmic}[1]\onehalfspacing
    \Require Code family $C = \{ C(\lambda) \}_{\lambda \in \BN}$ with $C(\lambda) \subseteq \BF_{q(\lambda)}^{n(\lambda)}$, encoding algorithm $\Enc_C : \BF_{q(\lambda)}^{k(\lambda)} \to \BF_{q(\lambda)}^{n(\lambda)}$ and decoding algorithm $\Dec_C : \BF_{q(\lambda)}^{n(\lambda)} \to \BF_{q(\lambda)}^{k(\lambda)} \cup \{ \perp \}$ for functions $q, n, k : \BN \to \BN$. Constant $\delta > 0$ determining noise rate. \emph{(Dependence on $\lambda$ omitted below.)}
    \Function{$\KeyGen$}{$1^\lambda$}
    \State Let $\sigma \gets \Unif(S_{[n]}), \pi_1, \ldots, \pi_n \gets \Unif(S_{\BF_q})$ be chosen uniformly at random. 
    \State $o \gets \Unif(\BF_q^n)$.
    \State \Return $\sk := (\sigma, \pi_1, \ldots, \pi_n, o)$.
    \EndFunction

    \Function{$\Enc$}{$\sk$}
    \State Write $\sk = (\sigma, \pi_1, \ldots, \pi_n, o)$. 
    \State Draw $x \gets \Unif(\BF_q^k)$, and $c \gets \Enc_C(x) \in \BF_q^n$.
    \State Define $\hat c \in \BF_q^n$ by $\hat c_i = \pi_i(c_{\sigma(i)})$ for each $i \in [n]$.
    \State Let $c' \gets \channSC_{\delta/2}(\hat c)$. 
    \State \Return $y := c' + o$.
    \EndFunction

    \Function{$\Dec$}{$\sk, y$}
    \State Write $\sk = (\sigma, \pi_1, \ldots, \pi_n, o)$. 
    \State Let $y' := y - o$.
    \State Define $y'' \in \BF_q^n$ by $y''_i = \pi_i^{-1}(y'_{\sigma^{-1}(i)})$ for each $i \in [n]$.
    \State Let $\hat y \gets \Dec_C(y'')$.
    \If{$\hat y \in C$ and $\Ham(\hat y, y') \leq \delta \cdot n$}
    \State \Return $\True$
    \Else
    \State \Return $\False$
    \EndIf
    \EndFunction
  \end{algorithmic}
\end{algorithm}

\paragraph{Acknowledgments.}
We thank Yuval Ishai for suggesting that we look into the permuted puzzles assumption. 
We thank Vinod Vaikuntanathan for helpful discussions during the early stages of this work.

Miranda Christ was partially supported by a Google CyberNYC grant, an Amazon Research Award, and NSF grants CCF-2312242, CCF-2107187, and CCF-2212233. This work was done in part while some of the authors were visiting the Simons Institute for the Theory of Computing.
Ankur Moitra was supported in part by a Microsoft Trustworthy AI Grant, NSF-CCF 2430381, an ONR grant, and a David
and Lucile Packard Fellowship.
Daniel Wichs was supported by NSF CNS-2349972 and CNS-2055510.

\newpage
\bibliographystyle{alpha}
\bibliography{sources}

\end{document}